\documentclass[sn-basic]{sn-jnl}

\usepackage{graphicx}%
\usepackage{multirow}%
\usepackage{amsmath,amssymb,amsfonts}%
\usepackage{amsthm}%
\usepackage{mathrsfs}%
\usepackage[title]{appendix}%
\usepackage[dvipsnames,table]{xcolor}%
\usepackage{textcomp}%
\usepackage{manyfoot}%
\usepackage{booktabs}%
\usepackage{algorithm}%
\usepackage{algorithmicx}%
\usepackage{algpseudocode}%
\usepackage{listings}%
\usepackage{subcaption}
\usepackage{tikz}
\usetikzlibrary{tikzmark,calc}
\usepackage{tkz-euclide}
\usepackage{mathtools}
\usepackage{thmtools,thm-restate}
\usepackage{enumitem}
\usepackage{hyperref}
\usepackage{cleveref}
\usepackage{varwidth}
\usepackage{orcidlink}
\usepackage{setspace}
\usepackage{caption}
\usepackage[sectionbib]{bibunits}
\usepackage{doi}

\defaultbibliographystyle{plainnat} 
\defaultbibliography{bibliography} 

\theoremstyle{thmstyleone}%
\newtheorem{theorem}{Theorem}
\newtheorem{proposition}[theorem]{Proposition}%
\newtheorem{corollary}[theorem]{Corollary}

\theoremstyle{thmstyletwo}%

\theoremstyle{thmstylethree}%
\newtheorem{definition}{Definition}%

\theoremstyle{thmstyleone}%
\newtheorem{thm}{Theorem}
\newtheorem{prop}[theorem]{Proposition}%

\DeclareMathOperator{\interior}{int}

\begin{document}

\title{Convergence-divergence models: Generalizations of phylogenetic trees modeling gene flow over time}


\author*[1,2]{\fnm{Jonathan D.}\sur{Mitchell}\orcidlink{0000-0003-3945-8150}}\email{jonathanmitchell88@gmail.com}

\author[1,2]{\fnm{Barbara R.}\sur{Holland}\orcidlink{0000-0002-4628-7938}}

\affil[1]{\orgdiv{School of Natural Sciences (Mathematics)}, \orgname{University of Tasmania}, \city{Hobart}, \state{TAS}, \country{Australia}}

\affil[2]{\orgdiv{ARC Centre of Excellence for Plant Success in Nature and Agriculture}, \orgname{University of Tasmania}, \city{Hobart}, \state{TAS}, \country{Australia}}



\abstract{Phylogenetic trees are simple models of evolutionary processes. They describe conditionally independent divergent evolution of taxa from common ancestors. Phylogenetic trees commonly do not have enough flexibility to adequately model all evolutionary processes. For example, introgressive hybridization, where genes can flow from one taxon to another. Phylogenetic networks model evolution not fully described by a phylogenetic tree. However, many phylogenetic network models assume ancestral taxa merge instantaneously to form ``hybrid'' descendant taxa. In contrast, our convergence-divergence models retain a single underlying ``principal'' tree, but permit gene flow over arbitrary time frames. Alternatively, convergence-divergence models can describe other biological processes leading to taxa becoming more similar over a time frame, such as replicated evolution. Here we present novel maximum likelihood-based algorithms to infer most aspects of $N$-taxon convergence-divergence models, many consistently, using a quartet-based approach. The algorithms can be applied to multiple sequence alignments restricted to genes or genomic windows or to gene presence/absence datasets.}

\keywords{phylogenetics, convergence-divergence models, phylogenetic networks, gene flow, convergence, Markov models}



\maketitle

\begin{bibunit}

\section{Convergence-divergence models as alternatives to phylogenetic networks}

By representing evolution on a phylogenetic tree, taxa are assumed to evolve conditionally independently from common ancestors. Independently evolving taxa diverge, becoming more different over time. From the beginning of quantitative inference of phylogenetic trees \citep{michener1957quantitative}, the fundamental assumption was that evolutionary divergence is the product of evolutionary rate and time. Therefore, ``degree of difference can give an estimate of evolutionary divergence'' from a common ancestor. For example, \cite{zuckerkandl1965evolutionary} introduced the notation of a ``molecular evolutionary clock'', which describes evolutionary time as being proportional to the number of sequence differences. This assumption of taxa evolving conditionally independently from common ancestors at a constant rate is often too simplistic.

Various biological assumptions do not meet the assumptions underpinning the representation of evolution with a phylogenetic tree. Evolutionary processes may not be independent nor divergent. Phylogenetic trees can be poor models when these assumptions are violated, for example, in the presence of gene flow \citep{leache2014influence}. Introgressive hybridization, horizontal gene transfer, recombination and replicated evolution --- the independent evolution of similar characteristics due to similar selective pressures \citep{james2023replicated} --- can violate these assumptions. Despite evolving independently, taxa undergoing genotypic replicated evolution display similarities in their genomes due to molecular convergence, for example, sites in multiple sequence alignments where the taxa have the same nucleotide. An alternative to phylogenetic trees is required to adequately model these biological processes.

There is a burgeoning body of literature on phylogenetics networks to address the limitations of phylogenetic trees. Phylogenetic networks model evolutionary processes leading to gene flow, for example, introgressive hybridization, horizontal gene transfer and recombination. See \cite{kong2022classes} for a thorough review of the classes of phylogenetic networks. Phylogenetic networks have ``hybrid'' nodes modeling gene flow --- not necessarily hybridization --- between taxa. However, on phylogenetic networks gene flow is assumed to be instantaneous at hybrid nodes. Many evolutionary processes can cause gene flow over a time interval. For example, introgressive hybridization, where hybrids of two taxa are repeatedly backcrossed with at least one of the taxa. Over a sufficiently long time interval, introgressive hybridization can lead to ``de-speciation'', where the species can no longer be distinguished from each other. One such example is two three-spined stickleback species in Enos Lake, British Columbia \citep{seehausen2008speciation}. Morphological and genetic analyses revealed the progressive de-speciation of the two species into a single hybrid species \citep{kraak2001increased,taylor2006speciation}. Such a scenario is not adequately modeled by most phylogenetic networks, which do not model gene flow over a time interval.

Alternatives to phylogenetic networks have been developed to model gene flow over a time interval. Isolation with migration models \citep{hey2010isolation} permit migration of individuals across otherwise isolated populations at some rate over a time interval. However, they are limited to datasets with several individuals per population. Furthermore, the algorithms tend to be very slow on datasets with many taxa. The ABBA-BABA test \citep{green2010draft} analyzes biallelic --- ``A'' for ancestral and ``B'' for derived --- sites of multiple sequence alignments (MSAs) for evidence of gene flow. A significant difference between counts of ``ABBA'' and ``BABA'' patterns across four taxa is interpreted as support for gene flow between taxa. However, the tests are limited to only $4$-taxon sets.

\bigskip

\emph{Convergence-divergence models} (CDMs) are an alternative to phylogenetic networks. They generalize phylogenetic trees in a different way to how phylogenetic networks generalize phylogenetic trees. Phylogenetic networks introduce hybrid nodes; the phylogenetic network does not generally display a single phylogenetic tree. In contrast, CDMs have a single underlying phylogenetic tree.

Distinct from phylogenetic trees, CDMs permit non-independent \emph{convergence}\footnote{Convergence, as defined here, is a generic term for any biological process that causes taxa to become more similar over a time interval. This includes, but is not limited to, convergent evolution.} of some subsets of taxa. As with some phylogenetic tree models, CDMs have an associated Markov model and rate matrices prescribing rates of substitutions. On a CDM, a single rate matrix prescribes the rates of substitutions between \emph{combinations} of states across the multiple converging taxa. For a set of converging taxa, only substitutions that take an arbitrary combination of states to identical states are allowed. For example, suppose two taxa are converging and have the combination of states $AT$ --- $A$ for taxon $1$ and $T$ for taxon $2$ --- at some site in an MSA. Then only substitutions to $AA$ or $TT$ are permitted by the rate matrix. By only permitting substitutions to identical states for converging taxa, the converging taxa become more similar in their associated sequence alignments over time.

CDMs build on the phylogenetic \emph{epoch models} of \cite{sumner2012algebra}, which envisage evolution occurring in a series of time intervals or epochs. Similar to isolation with migration models and in contrast to phylogenetic networks, CDMs model gene flow between otherwise isolated taxa over a time interval. In contrast to isolation with migration models, CDMs do not require multiple individuals per taxon; CDMs can be inferred from datasets with a single individual per taxon. Distinct from ABBA-BABA tests, CDMs can be inferred from datasets with any number of taxa. For ABBA-BABA tests, rejection of the null hypothesis --- a phylogenetic tree --- is assumed to be due to gene flow, with no explicit model of gene flow. In contrast, CDMs are explicit models of convergence, with the flexibility to model datasets of one to multiple individuals and many taxa, with gene flow between some taxa over a time period.

CDMs can also model replicated evolution. If there is genotypic replicated evolution in a system, then some parts of the genome evolve in similar ways in the taxa under similar selective pressures. Thus, these parts of the genome become more similar over time in those taxa and are modeled on the CDM as converging.

\bigskip

The mathematical properties of CDMs on three and four taxa were explored by \cite{mitchell2016distinguishing} and \cite{mitchell2018distinguishing}. \cite{holland2024distance} explored distance metric properties of hypothetical convergence models, without an explicit model assumed. They assume that converging taxa have smaller distances between them than if they had always diverged. In this article we develop algorithms to generalize inference of CDMs from previous studies. Inference is generalized by: 1) not assuming a molecular clock, 2) using the $2$-state general Markov model instead of the binary symmetric model and 3) increasing from $3$- or $4$-taxon to $N$-taxon datasets. These algorithms make CDM inference more widely accessible on empirical datasets. CDMs can be applied to a range of biological datasets, including gene or genomic window MSAs, gene (or gene family) presence/absence datasets and Diversity Arrays Technology (DArT) datasets \citep{jaccoud2001diversity}.

\section{Modeling convergence on convergence-divergence models}
\label{prelim}

We start with a brief discussion on the development of CDMs in previous articles. We briefly describe what convergence is and how CDMs generalize tree models by incorporating convergence.

On a phylogenetic tree, a Markov model defines rates of substitutions between states of the state space, for example, between nucleotides at each site in a sequence. Here we consider a continuous-time Markov model, with the flexibility of having possibly different rate matrices on distinct edges of the tree. As is standard \citep{felsenstein2004}, the edge lengths, rate matrices and root probabilities define probabilities of combinations of states across all taxa at the time corresponding to the leaves, which we call the leaf taxa.

Under the standard formulation on binary phylogenetic trees, \emph{speciation events} instantaneously split a single ancestral edge in one \emph{epoch} into two descendant edges in the following epoch. Suppose there are $k$ edges present in some epoch before a speciation event. There is some collection of probabilities of combinations of states on the $k$ edges instantaneously before the speciation event. This collection of probabilities is represented by a vector. Then in the epoch after the speciation event there are $k+1$ edges. The vector of probabilities before the speciation event must be modified to represent probabilities of combinations of states on $k+1$ edges after the event.

One of the $k$ edges before the speciation event splits into two edges after the speciation event. The edge that is split is modeled by the Markov model in the epoch before the speciation event, with $n$ independent and identically evolving random variables, for example, a nucleotide sequence of length $n$. The speciation event duplicates each random variable associated with the split edge so there are $n$ pairs of identical random variables instantaneously after the speciation event. For example, for MSAs the two edges after the split correspond to identical sequences. These two edges, with the property of only identical states existing for an arbitrary random variable --- for example, site in an MSA --- are \emph{identical} edges. After the speciation event, the two edges are again modeled by the Markov model. The two edges independently diverge in the epoch after the speciation event and are no longer identical after any time has passed.

The ``splitting operator'' \citep{sumner2005entanglement} is the matrix that converts the vector of probabilities of combinations of states on the $k$ edges instantaneously before the speciation event to the vector on the $k+1$ edges instantaneously after. After marginalizing out the $k-1$ edges not involved in the speciation event, the probabilities of identical states on the identical edges instantaneously after the speciation event equal the probabilities of the states on the single ancestral edge instantaneously before the speciation event. For example, if the probability of state $i$ on the ancestral edge instantaneously before the speciation event is $p_i$, then the probability of state $i$ on the two descendant edges instantaneously after the speciation event is $p_{ii}=p_i$.

\cite{sumner2012algebra} recognized that splitting operators on phylogenetic trees could be ``pushed back'' above the root; equivalent expressions for the probabilities are obtained by assuming that all splitting operators act above the root, with some edges remaining identical from the root until instantaneously after the point where the speciation event was. See Figures~\ref{a}~and~\ref{b} for a graphical depiction of pushing back the splitting operators. Similarly, the ``$N$-taxon process'' of \cite{bryant2009hadamard} accounts for speciation events to determine probabilities of combinations of states at the leaves.

\begin{figure}[!htb]
\centering
\hspace*{\fill}
\begin{subfigure}{0.32\linewidth}
\centering
\begin{tikzpicture}[scale=0.83]
\draw[thick] (0,1.5) -- (0,0);
\draw[thick] (0,0) -- (-2,-2);
\draw[thick] (0,0) -- (2,-2);
\draw[thick] (1,-1) -- (0,-2);
\draw[fill=black](0,0) circle (2 pt) node [right] {$\delta$};
\draw[fill=black](1,-1) circle (2 pt) node [right] {$\delta$};
\end{tikzpicture}
\caption{Before pushing back the splitting operator}
\label{a}
\end{subfigure}
\hfill
\begin{subfigure}{0.32\linewidth}
\centering
\begin{tikzpicture}[scale=0.83]
\draw[thick] (0.09375,1.5) -- (0.09375,1);
\draw[thick] (0.09375,1) -- (0,0.75);
\draw[thick] (0.09375,1) -- (0.1875,0.75);
\draw[thick] (0,0.75) -- (0,0);
\draw[thick] (0.1875,0.75) -- (0.1875,0.5);
\draw[thick] (0.1875,0.5) -- (0.125,0.25);
\draw[thick] (0.1875,0.5) -- (0.25,0.25);
\draw[thick] (0.125,0.25) -- (0.125,0);
\draw[thick] (0.25,0.25) -- (0.25,0);
\draw[fill=black](0.09375,1) circle (2 pt) node [right] {$\delta$};
\draw[fill=black](0.1875,0.5) circle (2 pt) node [right] {$\delta$};
\draw[thick] (0,0) -- (-2,-2);
\draw[thick] (0.125,0) -- (1.125,-1);
\draw[thick] (1.125,-1) -- (0.125,-2);
\draw[thick] (0.25,0) -- (2.25,-2);
\end{tikzpicture}
\caption{After pushing back the splitting operator}
\label{b}
\end{subfigure}
\hfill
\begin{subfigure}{0.32\linewidth}
\centering
\begin{tikzpicture}[scale=0.83]
\draw[thick] (0.09375,1.5) -- (0.09375,1);
\draw[thick] (0.09375,1) -- (0,0.75);
\draw[thick] (0.09375,1) -- (0.1875,0.75);
\draw[thick] (0,0.75) -- (0,0);
\draw[thick] (0.1875,0.75) -- (0.1875,0.5);
\draw[thick] (0.1875,0.5) -- (0.125,0.25);
\draw[thick] (0.1875,0.5) -- (0.25,0.25);
\draw[thick] (0.125,0.25) -- (0.125,0);
\draw[thick] (0.25,0.25) -- (0.25,0);
\draw[fill=black](0.09375,1) circle (2 pt) node [right] {$\delta$};
\draw[fill=black](0.1875,0.5) circle (2 pt) node [right] {$\delta$};
\draw[thick] (0,0) -- (-2,-2);
\draw[thick] (0.125,0) -- (1.125,-1);
\draw[thick] (1.125,-1) -- (0.125,-2);
\draw[thick] (0.25,0) -- (2.75,-2.5);
\draw[thick,red] (-2,-2) to[out=350,in=100] (-1.5,-2.5);
\draw[thick,red] (0.125,-2) to[out=190,in=80] (-0.625,-2.5);
\end{tikzpicture}
\caption{Modeling convergence with the splitting operator}
\label{c}
\end{subfigure}
\hspace*{\fill}
\caption{Two representations (a, b) of a phylogenetic tree with equivalent probability distributions at the leaves. $\delta$ is the splitting operator representing speciation events. (a) Splitting operators have not been pushed back. (b) Splitting operators have been pushed back. Parallel edges separated by small gaps are identical edges. (c) The rate matrix that keeps two identical edges identical models convergence between two diverged edges, represented by the two curved edges}
\label{splittingoperator}
\end{figure}
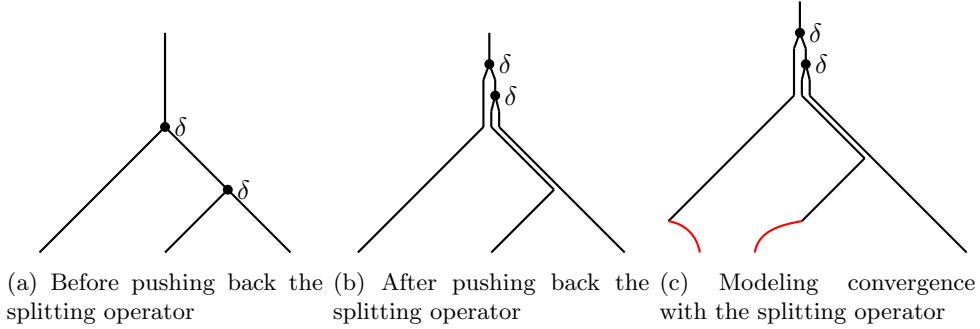

After pushing back all splitting operators, we consider all leaf taxa to be present at all times below the root on the phylogenetic tree. A set of identical edges corresponds to a set of leaf taxa, with each edge having one descendant leaf taxon. We call this set of leaf taxa \emph{identical}. When there is no ambiguity, we simply refer to the leaf taxa as ``taxa''. Taxa that are not identical are diverging at that time.

After pushing back splitting operators, at an arbitrary time below the root each taxon that is diverging from all others has its own rate matrix. In contrast, a single rate matrix models a \emph{set} of identical taxa. For sets of identical taxa, rate matrices model substitutions between combinations of states. Substitutions resulting in different states for identical taxa are not permitted. Thus, after pushing back splitting operators, identical taxa must remain identical until after the speciation event. After the speciation event separate --- but possibly identical --- rate matrices act independently on two subsets of the identical taxa and they diverge. Note that in Section~\ref{defs} and beyond we combine the rate matrices for all taxa at an arbitrary time --- an epoch --- to form a single rate matrix. This rate matrix models all permitted substitutions between combinations of states for all taxa at that time.

\cite{sumner2012algebra} recognized that rate matrices that keep identical taxa identical after pushing back the splitting operator can model ``convergence'' of diverged taxa. The rate matrices only permit substitutions to combinations of states where the converging taxa are all in the same state. See Figure~\ref{c} for a graphical depiction of convergence of diverged taxa. In contrast to independently diverging taxa, converging taxa become more similar over time. If the rate matrix for keeping identical taxa identical is instead applied to diverged taxa, substitutions lead to progressively more random variables that are identical for the taxa, for example, more sites in an MSA that are invariant across the converging taxa. In the limit, converging taxa become identical; see Theorem~\ref{convedges} of Section~\ref{cedges}. Taxa that are converging at the tips --- for example, in Figure~\ref{c} --- have identical MSAs in the limit as the time the taxa are converging increases. These rate matrices form the basis of convergence in our CDMs.

\section{Definitions and assumptions}
\label{defs}

\subsection{CDMs as data generating objects}
\label{cdmsdatagen}

We require several definitions to describe CDMs and our algorithms for inferring them. These assumptions are sufficient for consistent inference of many aspects of a CDM with our algorithms. In this section we introduce the definitions required to describe CDMs as data generating objects representing a combination of divergent and convergent evolutionary processes.

In contrast to a phylogenetic network, a CDM has a single underlying tree describing the ``tree-like'' part of the CDM.

\begin{definition}
\emph{Principal tree $T$} is a binary rooted leaf-labeled phylogenetic tree, ultrametric in time, with all edges having positive lengths.
\end{definition}

The principal tree is defined here before splitting operators are ``pushed back'' above the root. For simplicity, after pushing back splitting operators, we still refer to the resulting object as the principal tree.

We define the root as the node with outdegree $2$ that is the most recent common ancestor of all leaf taxa. It is useful to define the root as having indegree $1$ when considering splitting operators --- see Figure~\ref{splittingoperator} --- and defining the root as having indegree $0$ otherwise --- see Figure~\ref{4taxonCDMs2}. It is of no consequence which of the two ways the root is defined if the root distribution is specified by some model, for example, a Markov model.

We envisage evolution as occurring on the principal tree in discrete \emph{epochs} according to some Markov model. Epochs are separated by \emph{events} where the evolutionary processes change. We say that an epoch or event is \emph{before} another if it is between the root and the other epoch or event along a directed path. The other epoch or event is \emph{after} the epoch or event before it. Similarly, if two edges of the principal tree lie on a directed path from the root to a tip, the edge closest to the root is \emph{ancestral} to the \emph{descendant} edge, which is closest to the tip. The epoch that is after all other epochs is the \emph{tip epoch}. Each event has an associated \emph{event time}; the time along any directed path from the root of the principal tree to the event. Events with the same event time are \emph{concurrent}. Events are \emph{successive} if they occur at different event times, with no events between them. The \emph{epoch interval} $\tau=\left(a,b\right)$, with $b>a$, is the ordered pair of event times for two successive events. The \emph{epoch length} $t=b-a>0$ is the time between the two successive events.

Our notion of epochs on principal trees is not sufficient to describe evolution on CDMs. On CDMs, in some epochs taxa become more similar over time.

\begin{definition}
A set of distinct taxa on principal tree $T$ are \emph{converging} in an epoch if and only if probabilities under the continuous-time Markov model of identical states for all taxa in the set are increasing with time and probabilities of all other combinations of states are decreasing.
\end{definition}

For example, if the dataset is an MSA then a set of distinct taxa are converging if and only if probabilities of site patterns where all taxa in the set have the same state are increasing and all other probabilities are decreasing. Note that a set of identical taxa are not converging since for identical taxa the probabilities of identical states cannot increase.

In light of splitting operators being able to be pushed back above the root, in any epoch there is a one-to-one correspondence between the edges and the leaf taxa. Thus, we define epochs as specific partitions of the leaf taxon set. On a phylogenetic tree, each epoch corresponds to a specific partition of the leaf taxon set. An event corresponding to a speciation splits one part of the partition into two.

For CDMs we require an additional type of partition to standard partitions. A \emph{decorated partition} of the leaf taxon set $X$ is a partition $P$ of non-empty sets and partitions, where each set in $P$ is a strict subset of $X$, each partition in $P$ is a partition of a strict subset of $X$ and each taxon in $X$ appears in exactly one set or partition in $P$. On CDMs, decorated partitions correspond to epochs where there is convergence. For example, if the leaf taxa are labeled $a$, $b$ and $c$ from left to right on the CDM of Figure~\ref{c}, then the three epochs below the root are, in order, $\left\{a\right\}|\left\{b,c\right\}$ (partition), $\left\{a\right\}|\left\{b\right\}|\left\{c\right\}$ (partition) and $\left\{\left\{a\right\},\left\{b\right\}\right\}|\left\{c\right\}$ (decorated partition). Note that partitions in decorated partitions are represented as sets of sets to avoid confusion. For example, the partition in the decorated partition corresponding to the third epoch is represented as $\left\{\left\{a\right\},\left\{b\right\}\right\}$.

A \emph{divergence group} is ``tree-like'', corresponding to a subset of leaf taxon set $X$ in a partition or decorated partition. All taxa in a divergence group are identical and independently diverging from all taxa not in the divergence group. A \emph{convergence group} is ``non-tree-like'', corresponding to a partition in a decorated partition. As with a divergence group, all taxa in a convergence group are independently diverging from all taxa not in the convergence group. Taxa in the convergence group in the same subset of the partition are identical. A pair of taxa in the partition in distinct subsets are not identical, but are converging. Divergence groups and convergence groups collectively form \emph{convergence-divergence groups}.

Note that by defining epochs as corresponding to partitions and decorated partitions, we have already excluded scenarios where leaf taxa belong to multiple convergence-divergence groups in an epoch. For example, if, in some epoch, $a$ and $b$ are converging and $b$ and $c$ are converging, then $a$, $b$ and $c$ must all be in the same convergence group and therefore all converging. Thus, we have excluded the scenario where, in one epoch, $a$ and $b$ are converging and $b$ and $c$ are converging, but $a$ and $c$ are not converging.

The events at epoch boundaries fall into three possible classes. \emph{Speciation events} occur at nodes of the principal tree where splitting operators are before they are pushed back above the root. They take a subset of taxa that occur in the same subset in every partition or decorated partition before an event and split them across two subsets of the partition or decorated partition in the epoch after the event. The remaining two types of events are more broadly \emph{convergence-divergence events}. \emph{Divergence events} are events where at least one subset in a convergence group --- at least one subset of leaf taxa in a partition of a decorated partition --- in the epoch before the event forms a divergence group in the epoch after the event and no new convergence groups are formed. This includes, but is not limited to, scenarios where the epoch after the event corresponds to a partition and the epoch before the event corresponds to a decorated partition. \emph{Convergence events} are events where new convergence groups are formed, possibly with some convergence groups in the epoch before the event not existing in the epoch after the event. Note that there may be multiple convergence groups in an epoch. Note that our definition of convergence-divergence events does not include all possible scenarios of modeling convergence on CDMs. Furthermore, the assumptions that follow in Section~\ref{ass} further restrict convergence scenarios on CDMs. For an example of events on a CDM, consider CDM $5$ in Figure~\ref{4cdms5}. The event at the root and the second and third events are speciation events. The fourth and sixth events are convergence events. The fifth and seventh events are divergence events.

\bigskip

Our algorithms infer an $N$-taxon CDM from inferred $4$-taxon CDMs. Thus, we define CDMs describing the evolutionary history of a strict subset $X'$ of the full leaf taxon set $X$. Suppose $P_{\mathcal{N}}$ is the set of ordered partitions and decorated partitions corresponding to $\mathcal{N}$, ordered from the root to the leaves. For each partition or decorated partition of $P_{\mathcal{N}}$, suppose we delete all taxa in $X\setminus{}X'$ and subsequently delete any empty subsets of the partition or decorated partition. Next, suppose we recursively delete any partitions or decorated partitions identical to the previous one. Then the resulting ordered partitions and decorated partitions $P_{\mathcal{N}'}$ corresponds to \emph{displayed CDM} $\mathcal{N}'$ of $\mathcal{N}$.

Our CDMs have a lot of flexibility in modeling convergence. Some convergence scenarios are challenging to infer, particularly when they involve some of the same converging taxa as other convergence scenarios. In these scenarios a $4$-taxon CDM displayed on an $N$-taxon CDM may have a convergence group that appears in multiple epochs. For example, suppose the partitions and decorated partitions of a $5$-taxon CDM include $\left\{a\right\}|\left\{\left\{b\right\},\left\{c,d\right\}\right\}|\left\{e\right\}$ and $\left\{a\right\}|\left\{\left\{b\right\},\left\{c\right\}\right\}|\left\{d\right\}|\left\{e\right\}$. Then for the $4$-taxon displayed CDM on taxa $\left\{a,b,c,e\right\}$ both decorated partitions become $\left\{a\right\}|\left\{\left\{b\right\},\left\{c\right\}\right\}|\left\{e\right\}$. With the decorated partition repeated on the $4$-taxon displayed CDM, it is difficult to infer both decorated partitions on the $5$-taxon CDM from the $4$-taxon CDMs. With this scenario in mind, a convergence group $C_2$ of CDM $\mathcal{N}$ is \emph{nested} in convergence group $C_1$ of $\mathcal{N}$ if $C_1$ is before $C_2$ and there exists some displayed CDM $\mathcal{N}'$ of $\mathcal{N}$ where the two convergence groups restricted to the taxa displayed on $\mathcal{N}'$ belong to identical decorated partitions.

Furthermore, convergence involving closely related taxa can be challenging to infer. It can be difficult to distinguish from a scenario where the taxa only diverged, but over a shorter time interval. With this in mind, suppose two directed edges of the principal tree are $\left(u,v\right)$ and $\left(w,x\right)$, where $u$ and $w$ are the \emph{parent} nodes, ancestral to $v$ and $x$, the \emph{child} nodes. The two edges are \emph{sister} edges if and only if $u=w$. The corresponding leaf taxa are \emph{sister} taxa. Convergence involving sister taxa is \emph{sister convergence}. A convergence group with at least one pair of converging sister taxa is a \emph{sister convergence group}. Convergence that is not sister convergence is \emph{non-sister convergence} and convergence groups that are not sister convergence groups are \emph{non-sister convergence groups}.

\bigskip

On a phylogenetic tree, rate matrices model the conditionally independent divergence of taxa from common ancestors. Rate matrices are typically assigned to individual edges of the tree. For CDMs, rate matrices are instead assigned to each convergence-divergence group of an epoch. In general, substitution rates differ between each convergence-divergence group in the epoch. These rate matrices are combined across all convergence-divergence groups in an epoch to form a single rate matrix for the epoch, as in \cite{sumner2012algebra}. This rate matrix, defined by a continuous-time Markov model, describes all evolutionary processes in the epoch.

Recall that in each epoch we consider there to be a one-to-one correspondence between the edges and the taxa after pushing back splitting operators. Thus, for $N$ taxa and a state space with $m$ states, all rate matrices are of dimension $m^N\times{}m^N$. Each element of the rate matrix represents a substitution from one of the $m^N$ \emph{combinations of states} across the taxa to another. For example, suppose $N=3$ and $m=2$, with the state space $\left\{0,1\right\}$. We let indices of rate matrices be in binary form. Row $i$ corresponds to $i_1i_2\ldots{}i_N$ and column $j$ corresponds to $j_1j_2\ldots{}j_N$, where $i_a,j_a\in\left\{0,1\right\}$ for all $a\in\left\{1,2,\ldots{},N\right\}$, $i=1+\sum_{a=1}^{N}2^{N-a}i_a$ and $j=1+\sum_{a=1}^{N}2^{N-a}j_a$. Then one of the $2^3\times{}2^3=64$ elements of the rate matrix describes the substitution rate from combination of states $010$ to $011$. That is, the substitution from only the second taxon being in state $1$ to both the second and third taxa being in state $1$. As an example, suppose the rate matrix represents the decorated partition $\left\{\left\{a\right\},\left\{b\right\}\right\}|\left\{c\right\}$ corresponding to the tip epoch of the CDM in Figure~\ref{c}. $\left[\boldsymbol{Q}\right]_{ij}$ is the rate of substitution from $j_1j_2j_3$ to $i_1i_2i_3$, where $i_k,j_k\in\left\{0,1\right\}$, $k\in\left\{1,2,3\right\}$ represent the states of taxa $a$, $b$ and $c$ respectively. Suppose the convergence group $\left\{\left\{a\right\},\left\{b\right\}\right\}$ has rate of substitution from $00$, $01$ or $10$ to $11$ of $\alpha_1>0$ and rate from $01$, $10$ or $11$ to $00$ of $\beta_1>0$. Suppose the divergence group $\left\{c\right\}$ has rate from $0$ to $1$ of $\alpha_2>0$ and rate from $1$ to $0$ of $\beta_2>0$.

Then
\begin{align}
\label{ratematrix}
\boldsymbol{Q}=\left[\begin{array}{c|cccccccc}
 & 000 & 001 & 010 & 011 & 100 & 101 & 110 & 111 \\
\hline
000 & \ast & \beta_2 & \beta_1 & 0 & \beta_1 & 0 & \beta_1 & 0 \\
001 & \alpha_2 & \ast & 0 & \beta_1 & 0 & \beta_1 & 0 & \beta_1 \\
010 & 0 & 0 & \ast & \beta_2 & 0 & 0 & 0 & 0 \\
011 & 0 & 0 & \alpha_2 & \ast & 0 & 0 & 0 & 0 \\
100 & 0 & 0 & 0 & 0 & \ast & \beta_2 & 0 & 0 \\
101 & 0 & 0 & 0 & 0 & \alpha_2 & \ast & 0 & 0 \\
110 & \alpha_1 & 0 & \alpha_1 & 0 & \alpha_1 & 0 & \ast & \beta_2 \\
111 & 0 & \alpha_1 & 0 & \alpha_1 & 0 & \alpha_1 & \alpha_2 & \ast \\
\end{array}
\right],
\end{align}
where each $\ast$ ensures that a column of $\boldsymbol{Q}$ sums to $0$. (Note that many authors use the row sum convention instead.) Column headings represent the initial combination of states, while row headings represent the final combination of states. Observe that since $a$ and $b$ are in the same convergence group, the only permitted substitutions result in a change of state for only $c$ or a change of state for at least one of $a$ and $b$ such that $a$ and $b$ are in the same final state. Note that substitutions involving both the convergence group and the divergence group at the same time --- for example, $000$ to $111$ --- are not permitted. Rate matrices for the $2$-state general Markov model are presented explicitly in \cite{sumner2012algebra}.

\bigskip

Finally, we can define CDMs.

\begin{definition}
\label{CDMa}
A \emph{convergence-divergence model $\mathcal{N}=\left(T,\boldsymbol{\Pi},\mathfrak{E},\mathfrak{Q},\mathfrak{t}\right)$} comprises principal tree $T$, root distribution $\boldsymbol{\Pi}$ and set of partitions and decorated partitions $\mathfrak{E}$ corresponding to epochs ordered from the root to the tip epoch, with $E_i\in\mathfrak{E}$ having associated rate matrix $\boldsymbol{Q}_i\in\mathfrak{Q}$ and epoch interval $\tau_i\in\mathfrak{t}$.
\end{definition}

Root distribution $\boldsymbol{\Pi}$ is the probability vector of states for the single root taxon. Alternatively, after pushing back the splitting operator, it is the probability tensor --- represented as a vector --- of combinations of states across the taxa at the root. Since all taxa must be identical at the root, only the combinations of states where all taxa have the same state have non-zero probabilities.

For the algorithms that follow in Sections~\ref{topprincipaltree}-\ref{metricCDM}, we consider a special type of CDM.

\begin{definition}
\label{CDM}
A \emph{$2$-state general convergence-divergence model} is a convergence-divergence model with rate matrices from the $2$-state general Markov model, equal ratios of substitution rates for all convergence-divergence groups ($\frac{\alpha_l}{\beta_l}=\frac{\alpha}{\beta}$, $\alpha_l,\beta_l>0$ for the $l^{th}$ convergence-divergence group) and $\boldsymbol{\Pi}$ the stationary distribution.
\end{definition}

It is straightforward to show that before pushing back splitting operators the stationary distribution is $\boldsymbol{\Pi}=\left[\frac{\beta}{\alpha+\beta},\frac{\alpha}{\alpha+\beta}\right]^T$ and this is omitted. (Recall that we define rate matrices such that columns sum to $0$. Thus, the root distribution is a column vector.) After pushing back splitting operators, indices of the stationary distribution are in binary form, with $\left[\boldsymbol{\Pi}\right]_0=\frac{\beta}{\alpha+\beta}$, $\left[\boldsymbol{\Pi}\right]_{2^N}=\frac{\alpha}{\alpha+\beta}$ and $\left[\boldsymbol{\Pi}\right]_i=0$ for all $i\in\left\{2,3,\ldots{},2^N-1\right\}$.

Note that with $\frac{\alpha_{l}}{\beta_{l}}=\frac{\alpha}{\beta}$ and the stationary distribution at the root the Markov model is equivalent to the $2$-state general time-reversible (GTR) model. In fact, this is the only such $2$-state $2$-parameter model with $\frac{\alpha_{l}}{\beta_{l}}=\frac{\alpha}{\beta}$ and the stationary distribution at the root. From here onwards we assume all CDMs are $2$-state general CDMs, which we simply refer to as CDMs.

As is common in phylogenetic Markov model based inference, substitution rates and epoch lengths are not identifiable individually in our models. Instead, we can only identify some products of epoch lengths and the substitution rates in the epochs between the events. Roughly, these products represent the ``amount of evolution''. Furthermore, we cannot always identify changes in the amount of evolution between epochs. Instead, we can only identify the ``average'' amount of evolution across these epochs. Thus, to obtain identifiable parameters we consider contiguous sections of edges of the principal tree that potentially span multiple epochs. A \emph{converging section} of an edge is a section of an edge restricted to a single epoch where that edge corresponds to converging taxa in the epoch. \emph{Diverging sections} of an edge are the disjoint sections of an edge that remain after deleting the converging sections or entire edges if there are no converging sections. We assign parameters to each converging and diverging section, which we call \emph{convergence parameters} and \emph{divergence parameters}, respectively. All edges corresponding to taxa in a convergence group in an epoch correspond to the same convergence parameter. We refer to the sum of convergence and divergence parameters along the shortest path between two leaf taxa as the \emph{distance} between the leaf taxa. Similarly, the sum of convergence and divergence parameters along an edge of the principal tree is the \emph{edge length}. Finally, since the Markov model has two states, the root distribution has a single identifiable parameter called the \emph{root parameter}. For example, the divergence parameters of CDM $5$ are parameters $1-5$, $7$, $8$, $10$ and $11$ of Figure~\ref{4cdms5}, while the convergence parameters are parameters $6$ and $9$. We discuss the parameters in more detail in Appendix~\ref{pars}.

Finally, the collection of probabilities of combinations of states at the leaves of the principal tree is called the \emph{phylogenetic tensor}. It is a vector representation of a tensor.

The phylogenetic tensor $\boldsymbol{P}$ is
\begin{align*}
\boldsymbol{P}=\prod_{a=1}^{r}\exp\left(\boldsymbol{Q}_at_a\right)\cdot{}\boldsymbol{\Pi},
\end{align*}
where, for epoch $a$, $\boldsymbol{Q}_a$ is the rate matrix, $t_a$ is the epoch length and the product is over the epochs, whose indices are ordered according to Definition~\ref{CDMa}. Note that the phylogenetic tensor can also be expressed in terms of root, convergence and divergence parameters.

For convenience, the phylogenetic tensor is transformed into the Hadamard basis by multiplying it by a Hadamard matrix \citep{hendy1989framework,hendy1989relationship,bryant2009hadamard}. The Hadamard basis permits a simple parameterization of the phylogenetic tensor, making it easier to establish identifiability of a CDM. From here onwards we deal mostly with phylogenetic tensors in the Hadamard basis. We refer to these as \emph{transformed phylogenetic tensors}.

\subsection{Assumptions}
\label{ass}

Here we introduce some simplifying assumptions for our CDMs. Any assumptions already made, such as no taxa appearing in multiple subsets in a partition or multiple subsets or partitions in a decorated partition, apply to our CDMs in the following sections. In addition, the following assumptions on $N$-taxon CDMs are designed to avoid overparameterization and simplify inference. Further assumptions sufficient for consistent inference of many aspects of the $N$-taxon CDM appear later.

The \emph{generating CDM} and \emph{generating parameter} are the CDM and parameter vector corresponding to the data generating process. The \emph{parameter space} of a CDM is the set of all possible hypothetical generating parameters. Parameter spaces of distinct CDMs may not be disjoint. For example, the parameter space of the first CDM has the parameter space of the second CDM \emph{nested} in it if the second CDM can be obtained from the first CDM by fixing some parameters. (Note that the fixed parameters of the first CDM are not considered parameters in the second CDM.) Likewise, if the parameter space of one CDM has the parameter space of a second CDM nested in it, we say the second CDM is \emph{nested} in the first CDM. Thus, the generating parameter may be in parameter spaces of multiple CDMs. In this case, we consider the generating CDM to be the lowest parameter CDM whose parameter space contains the generating parameter.

Our assumptions on the CDMs are as follows.

\begin{enumerate}
    \item The generating CDM is a $2$-state general CDM. \label{2stategenCDM}
    \item The principal tree is rooted by a single outgroup taxon $o$.
    \item There are no convergence groups that include $o$. \label{nooutcon}
    \item No events are both speciation events and convergence-divergence events. \label{assevents}
    \item In each epoch there is at most one convergence group. \label{onecongroup}
    \item All convergence groups correspond to partitions of exactly two subsets.
    \item Each convergence group appears in at most one epoch.
    \item There are no consecutive epochs both with convergence groups.
    \item No convergence groups are nested in other convergence groups. \label{nonest}
    \item There are no sister convergence groups. \label{nosist}
    \item The generating parameter corresponds to a generic point in the generating CDM parameter space. \label{genparam}
\end{enumerate}

For most algorithms, propositions and theorems of the following sections --- excluding Section~\ref{cedges} --- the assumptions of Section~\ref{ass} hold, as well as some other assumptions that we describe later, sufficient for consistent inference of many aspects of the CDM.

In addition to the assumptions on the CDMs, all random variables are independent and generated by the data generating process corresponding to the generating parameter --- the iid assumption. Throughout this article $n$ is the sample size. Random variables may be sites in an MSA for a gene or genomic region or genes for gene presence/absence datasets. Our random variables are multinomially distributed --- a distribution from the exponential family --- with, for each random variable, each of the $2^N$ combinations of states having some probability that is a function of the generating parameter.

Generating parameters on the boundary of the generating CDM parameter space correspond to non-generic points in the parameter space. Thus, these parameters are not permitted by Assumption~\ref{genparam}, ensuring a \emph{regular} exponential family; see Definition~2.1 of \cite{drton2009likelihood} for a definition of regular models. These parameters are problematic as they correspond to edges of the principal tree of length zero, epoch lengths of zero or substitution rates of zero. These parameters can correspond to polytomies on the principal tree or convergence parameters of value zero.

The iid assumption may not be biologically realistic. For example, sites in an MSA are typically not independent. Nonetheless, it is common in phylogenetic inference to use all sites in an MSA, often restricted to a gene or genomic window. Our algorithms involve comparing maximum likelihoods for various CDMs. However, ignoring any violations of independence and assuming all random variables are independent is essentially a composite likelihood (or pseudolikelihood) approach. The composite likelihood is used in many phylogenetic software, including MP-EST \citep{liu2010maximum}, PhyloNet \citep{yu2015maximum} and SNaQ \citep{solis2016inferring}.

The composite likelihood has desirable statistical properties. Under standard regularity conditions, the maximum composite likelihood estimator is a consistent estimator of the generating parameter \citep{lindsay1988composite}. Furthermore, central limit theorems exist for dependent random variables under some specific weak conditions, for example, for univariate random variables \citep{bradley2007introduction}, extended to multivariate random variables \citep{tone2010central}. Roughly, if our variables are sites in an MSA, then --- along with standard assumptions on the mean and variance --- sufficient conditions for the central limit theorem of \cite{tone2010central} are: 1) dependence between sites decreases to zero as distance between them increases, 2) the joint distribution of an arbitrary $k$ adjacent sites is invariant across the MSA and 3) no two sites in an MSA are perfectly correlated.

To (roughly) satisfy 2), we can restrict MSAs to individual genes or genomic windows, inferring CDMs on each gene or genomic window independently, which appropriately accounts for incomplete lineage sorting. We could instead retain only sites that are approximately independent. However, with incomplete lineage sorting in mind, if we restrict an MSA to a gene or genomic window, then discarding most sites likely gives poor statistical power. With this in mind and the desirable statistical properties of the composite likelihood, we retain the iid assumption without discarding any random variables.

\bigskip

Correctly discovering sister convergence is challenging, typically requiring larger sample sizes than non-sister convergence. However, sister convergence is unlikely to lead to incorrect inference of the topology of the principal tree or false discovery of non-sister convergence groups. Thus, we assume there is no sister convergence and do not attempt to infer it in our algorithms.

We prove that if all convergence parameters are sufficiently ``small'', then the topology of the principal tree can be inferred consistently. ``Small'' convergence parameters could correspond to slow substitution rates and/or short epoch intervals. With some further assumptions that we describe, we can consistently infer all convergence groups of a CDM and some of its parameters.

\section{Limiting behavior of converging taxa}
\label{cedges}

Recall that taxa are converging if probabilities of identical states for all taxa increase with time and all other probabilities decrease. Recall also that for a given random variable all identical taxa have the same state. For example, if all taxa are identical then all sites in an MSA are invariant. It is not immediately clear what the limiting behavior of convergence groups is and whether converging taxa become identical in the limit or not. In this section we establish that converging taxa correspond to identical taxa in the limit as the epoch length increases. For this section we do not need any of the assumptions of Section~\ref{ass} except for Assumption~\ref{2stategenCDM} and also Assumption~\ref{genparam} for Theorem~\ref{convedges}.

We start by establishing some properties of rate matrices of convergence groups. These properties are then used to establish the limiting behavior of elements of the phylogenetic tensor corresponding to probabilities of combinations of states for converging taxa.

For the following proposition and theorem, we assume $i$ and $j$ are the indices corresponding to the final and initial combination of states, respectively, associated with a state transition, with binary forms $i_1i_2\ldots{}i_N$ and $j_1j_2\ldots{}j_N$ for $N$ taxa.

Recall that our rate matrices that model convergence of diverged taxa are the same as the rate matrices that model identical taxa remaining identical. Thus, for our rate matrices no distinction needs to be made between a pair of taxa that are identical and a pair of taxa that are converging. Thus, the only information required for composing the rate matrix is the sets of (identical) taxa in divergence groups --- possibly a single taxon --- and the sets of taxa in convergence groups --- possibly with some identical taxa. We assume that if there are $k$ convergence-divergence groups in some epoch, then $\mathcal{C}=\left\{C_1,C_2,\ldots{}C_k\right\}$ is the set of sets of taxa in each convergence-divergence group. For a partition, $\mathcal{C}$ is the set of subsets in the partition. For a decorated partition, each subset of $\mathcal{C}$ is either a subset in the decorated partition or the set of taxa in a partition in the decorated partition.

\begin{proposition}
\label{p1}
Suppose a tip epoch of CDM $\mathcal{N}$ with leaf taxon set $X$ and $\left|X\right|=N$ corresponds to a set of sets of taxa in each convergence-divergence group $\mathcal{C}=\left\{C_1,C_2,\ldots{}C_k\right\}$. Suppose $\boldsymbol{Q}^{\left[\mathcal{C}\right]}$ is the $2^N\times2^N$ rate matrix representing $\mathcal{C}$. Then
\begin{align*}
\left[\boldsymbol{Q}^{\left[\mathcal{C}\right]}\right]_{ij}=
\begin{cases}
\alpha_r \quad &\text{if for some } C_r\in{}\mathcal{C}, \quad \prod_{a\in{}C_r}j_a=0, \\
&\prod_{a\in{}C_r}i_a=1 \text{ and } i_a=j_a \text{ for all } a\in{}X\setminus{}C_r, \\
\beta_r \quad &\text{if for some } C_r\in{}\mathcal{C}, \quad \prod_{a\in{}C_r}\left(1-j_a\right)=0, \\
&\prod_{a\in{}C_r}\left(1-i_a\right)=1 \text{ and } i_a=j_a \text{ for all } a\in{}X\setminus{}C_r, \\
0 \quad &\text{otherwise if } i\neq{}j, \\
-\sum_{s=1,s\neq{}j}^{2^N}\left[\boldsymbol{Q}^{\left[\mathcal{C}\right]}\right]_{sj} \quad &\text{if } i=j,
\end{cases}
\end{align*}
where $\alpha_r,\beta_r>0$.
\end{proposition}

See Appendix~\ref{prop1} for the proof.

For example, suppose the epoch is $\left\{\left\{a\right\},\left\{b\right\}\right\}|\left\{c\right\}$. Then $\mathcal{C}=\left\{\left\{a,b\right\},\left\{c\right\}\right\}$ and it is straightforward to verify that $\boldsymbol{Q}^{\left[\mathcal{C}\right]}$ is the rate matrix of Equation~\ref{ratematrix} corresponding to the tip epoch in Figure~\ref{c}. Note that if the epoch is $\left\{a,b\right\}|\left\{c\right\}$, then again $\mathcal{C}=\left\{\left\{a,b\right\},\left\{c\right\}\right\}$ and $\boldsymbol{Q}^{\left[\mathcal{C}\right]}$ is also the rate matrix of Equation~\ref{ratematrix}, this time corresponding to an epoch where taxa $a$ and $b$ are identical and diverging from taxon $c$.

\begin{theorem}
\label{convedges}
Suppose an arbitrary epoch of CDM $\mathcal{N}$ corresponds to set of sets of taxa in each convergence-divergence group $\mathcal{C}=\left\{C_1,C_2,\ldots{}C_k\right\}$. Then if $a,b\in{}C_i$, as tip epoch length $t\to\infty$, $a$ and $b$ become identical.
\end{theorem}

See Appendix~\ref{thm1} for the proof.

As an example of the consequence of Theorem~\ref{convedges}, suppose our dataset is an MSA, the CDM is that of Figure~\ref{c} and $\boldsymbol{P}=\left[\left[\boldsymbol{P}\right]_{000},\left[\boldsymbol{P}\right]_{001},\ldots{},\left[\boldsymbol{P}\right]_{111}\right]^T$ is the phylogenetic tensor. Then in the limit as the epoch length of the tip epoch increases, $\left[\boldsymbol{P}\right]_{010}=\left[\boldsymbol{P}\right]_{011}=\left[\boldsymbol{P}\right]_{100}=\left[\boldsymbol{P}\right]_{101}=0$, while $\left[\boldsymbol{P}\right]_{000},\left[\boldsymbol{P}\right]_{001},\left[\boldsymbol{P}\right]_{110},\left[\boldsymbol{P}\right]_{111}>0$.

\section{CDM identifiability and distinguishability}
\label{ident}

\subsection{Constructing \texorpdfstring{$N$}{N}-taxon CDMs from \texorpdfstring{$4$}{4}-taxon CDMs}

In model inference and selection we are often concerned with whether generating models can be recovered given a sufficiently large amount of data. Roughly, models whose parameters can be recovered given enough data have the statistical property of identifiability.

The more complex a model is, the more challenging it is to establish identifiability theoretically. Even with the assumptions of Section~\ref{ass}, CDMs potentially have a lot of complexity, with many possible principal tree topologies and convergence groups, particularly if there are many leaf taxa. Furthermore, even if identifiability can be established, model selection may require substantial computational resources or heuristic methods to search over the parameter spaces of all possible CDMs.

These challenges can be avoided by considering $4$-taxon CDMs, performing model selection on these $4$-taxon CDMs and inferring an $N$-taxon CDM from the $4$-taxon CDMs. We consider all $\binom{N-1}{3}$ $4$-taxon sets that include an outgroup. We include the outgroup, which is defined \emph{a priori}, since the edge to place the root of the principal tree on is typically not identifiable.

Combinatorial methods are used to construct an $N$-taxon CDM from the inferred $4$-taxon CDMs. This ``divide and conquer'' approach of reconstructing phylogenetic trees and phylogenetic networks from triplets or quartets is well established in phylogenetic inference \citep{semple2003phylogenetics,huson1998disk,huson1999disk}.

\subsection{Identifiability of \texorpdfstring{$4$}{4}-taxon CDMs}
\label{4taxaidentifiability}

Here we introduce the $4$-taxon CDMs meeting the assumptions of Section~\ref{ass}. We establish whether these CDMs are \emph{generically identifiable} --- the set of points in the parameter space where identifiability does not hold is of measure zero. From here onwards, generic identifiability is called identifiability.

The assumption that one taxon in each quartet is the outgroup reduces the number of possible topologies of $4$-taxon principal trees; up to leaf labeling it can be assumed to be $\left(o,\left(a,\left(b,c\right)\right)\right)$, with outgroup $o$ and leaf taxa $a$, $b$ and $c$. Assumption~\ref{nooutcon} of Section~\ref{ass} of no convergence involving leaf $o$ also reduces the number of possible convergence groups.

Recall our argument that sister convergence is challenging to discover and is unlikely to lead to incorrect inference of the topology of the principal tree or false discovery of non-sister convergence groups. Thus, we assume the $N$-taxon CDM has no sister convergence (Assumption~\ref{nosist}). However, $N$-taxon CDMs with no sister convergence may still display $4$-taxon CDMs with sister convergence. Again, ignoring a sister convergence group on a $4$-taxon CDM is unlikely to lead to incorrect inference, other than the sister convergence group failing to be detected. Thus, we only perform model selection on $4$-taxon CDMs with no sister convergence groups. Furthermore, since we assume there are no sister convergence groups on the $N$-taxon CDM, for any arbitrary convergence group on the $N$-taxon CDM there must be at least one displayed $4$-taxon CDM with the convergence group where it is a non-sister convergence group. The algorithms that follow appropriately consider the fact that some convergence groups on the $N$-taxon CDM are sister convergence groups on some $4$-taxon CDMs and thus not detected.

Assuming arbitrary leaf taxon labels and arbitrary non-generic parameter values, there are five $4$-taxon CDMs satisfying the assumptions of Section~\ref{ass}, called CDMs $1-5$ and displayed in Figure~\ref{4taxonCDMs2}. To determine whether the five CDMs are identifiable we must establish whether there is a parameter set on each CDM that is identifiable. Note that this set may not be the root, convergence and divergence parameters, but some combinations of these parameters. Although we assume a specific outgroup taxon, the exact location of the root is not identifiable. With the root location not identifiable, the only parameters of the principal tree that we attempt to infer are the parameters of the \emph{unrooted} principal tree. We then root the principal tree at an arbitrary position on the outgroup edge.

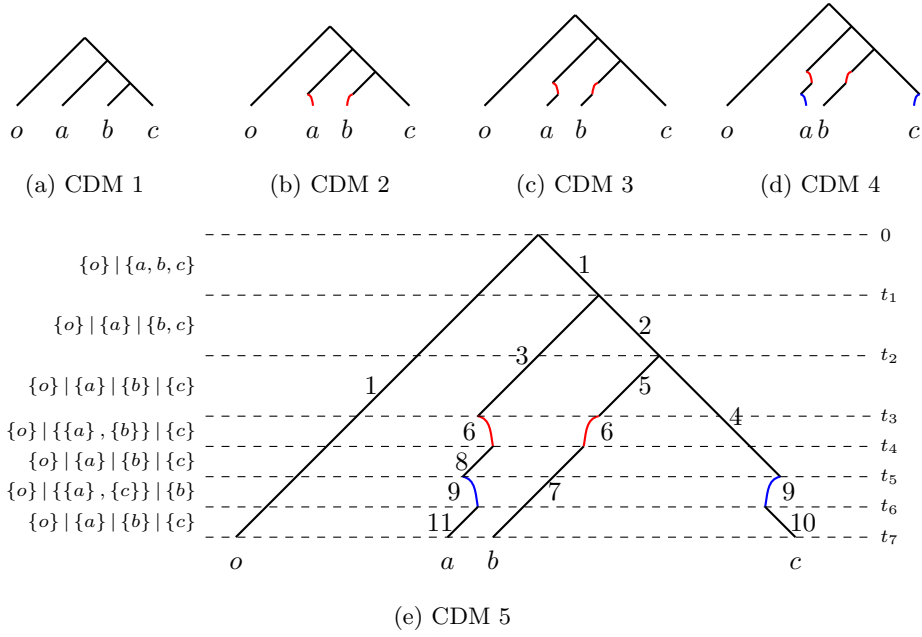
\begin{figure}[!htb]
\hspace*{\fill}
\begin{subfigure}{0.24\linewidth}
\centering
\begin{tikzpicture}[scale=0.3]]
\draw[thick] (0,0) -- (-3,-3);
\draw[thick] (0,0) -- (3,-3);
\draw[thick] (1,-1) -- (-1,-3);
\draw[thick] (2,-2) -- (1,-3);
\node[below] at (-3,-3) {\strut{$o$}};
\node[below] at (-1,-3) {\strut{$a$}};
\node[below] at (1,-3) {\strut{$b$}};
\node[below] at (3,-3) {\strut{$c$}};
\end{tikzpicture}
\caption{CDM $1$}
\label{4cdms1}
\end{subfigure}
\hfill
\begin{subfigure}{0.24\linewidth}
\centering
\begin{tikzpicture}[scale=0.3]
\draw[thick] (0,0) -- (-3.5,-3.5);
\draw[thick] (0,0) -- (3.5,-3.5);
\draw[thick] (1,-1) -- (-1,-3);
\draw[thick] (2,-2) -- (1,-3);
\draw[thick,red] (-1,-3) to[out=350,in=100] (-0.75,-3.5);
\draw[thick,red] (1,-3) to[out=190,in=80] (0.75,-3.5);
\node[below] at (-3.5,-3.5) {\strut{$o$}};
\node[below] at (-0.75,-3.5) {\strut{$a$}};
\node[below] at (0.75,-3.5) {\strut{$b$}};
\node[below] at (3.5,-3.5) {\strut{$c$}};
\end{tikzpicture}
\caption{CDM $2$}
\label{4cdms2}
\end{subfigure}
\hfill
\begin{subfigure}{0.24\linewidth}
\centering
\begin{tikzpicture}[scale=0.3]
\draw[thick] (0,0) -- (-4,-4);
\draw[thick] (0,0) -- (4,-4);
\draw[thick] (1,-1) -- (-1,-3);
\draw[thick] (2,-2) -- (1,-3);
\draw[thick,red] (-1,-3) to[out=350,in=100] (-0.75,-3.5);
\draw[thick,red] (1,-3) to[out=190,in=80] (0.75,-3.5);
\draw[thick] (-0.75,-3.5) -- (-1.25,-4);
\draw[thick] (0.75,-3.5) -- (0.25,-4);
\node[below] at (-4,-4) {\strut{$o$}};
\node[below] at (-1.25,-4) {\strut{$a$}};
\node[below] at (0.25,-4) {\strut{$b$}};
\node[below] at (4,-4) {\strut{$c$}};
\end{tikzpicture}
\caption{CDM $3$}
\label{4cdms3}
\end{subfigure}
\hfill
\begin{subfigure}{0.24\linewidth}
\centering
\begin{tikzpicture}[scale=0.3]
\draw[thick] (0,0) -- (-4.5,-4.5);
\draw[thick] (0,0) -- (4,-4);
\draw[thick] (1,-1) -- (-1,-3);
\draw[thick] (2,-2) -- (1,-3);
\draw[thick,red] (-1,-3) to[out=350,in=100] (-0.75,-3.5);
\draw[thick,red] (1,-3) to[out=190,in=80] (0.75,-3.5);
\draw[thick] (-0.75,-3.5) -- (-1.25,-4);
\draw[thick] (0.75,-3.5) -- (-0.25,-4.5);
\draw[thick,blue] (-1.25,-4) to[out=350,in=100] (-1,-4.5);
\draw[thick,blue] (4,-4) to[out=190,in=80] (3.75,-4.5);
\node[below] at (-4.5,-4.5) {\strut{$o$}};
\node[below] at (-1,-4.5) {\strut{$a$}};
\node[below] at (-0.25,-4.5) {\strut{$b$}};
\node[below] at (3.75,-4.5) {\strut{$c$}};
\end{tikzpicture}
\caption{CDM $4$}
\label{4cdms4}
\end{subfigure}
\hspace*{\fill}
\\
\centering
\begin{subfigure}{\linewidth}
\centering
\begin{tikzpicture}[scale=0.8]
\draw[thick] (0,0) -- (-5,-5);
\draw[thick] (0,0) -- (4,-4);
\draw[thick] (1,-1) -- (-1,-3);
\draw[thick] (2,-2) -- (1,-3);
\draw[thick,red] (-1,-3) to[out=350,in=100] (-0.75,-3.5);
\draw[thick,red] (1,-3) to[out=190,in=80] (0.75,-3.5);
\draw[thick] (-0.75,-3.5) -- (-1.25,-4);
\draw[thick] (0.75,-3.5) -- (-0.75,-5);
\draw[thick,blue] (-1.25,-4) to[out=350,in=100] (-1,-4.5);
\draw[thick,blue] (4,-4) to[out=190,in=80] (3.75,-4.5);
\draw[thick] (-1,-4.5) -- (-1.5,-5);
\draw[thick] (3.75,-4.5) -- (4.25,-5);
\node[below] at (-5,-5) {\strut{$o$}};
\node[below] at (-1.5,-5) {\strut{$a$}};
\node[below] at (-0.75,-5) {\strut{$b$}};
\node[below] at (4.25,-5) {\strut{$c$}};
\draw[dashed] (-5.5,0) -- (5.5,0);
\draw[dashed] (-5.5,-1) -- (5.5,-1);
\draw[dashed] (-5.5,-2) -- (-0.25,-2);
\draw[dashed] (0,-2) -- (5.5,-2);
\draw[dashed] (-5.5,-3) -- (3,-3);
\draw[dashed] (3.25,-3) -- (5.5,-3);
\draw[dashed] (-5.5,-3.5) -- (5.5,-3.5);
\draw[dashed] (-5.5,-4) -- (5.5,-4);
\draw[dashed] (-5.5,-4.5) -- (5.5,-4.5);
\draw[dashed] (-5.5,-5) -- (5.5,-5);
\node[right] at (5.5,0) {\footnotesize{$0$}};
\node[right] at (5.5,-1) {\footnotesize{$t_1$}};
\node[right] at (5.5,-2) {\footnotesize{$t_2$}};
\node[right] at (5.5,-3) {\footnotesize{$t_3$}};
\node[right] at (5.5,-3.5) {\footnotesize{$t_4$}};
\node[right] at (5.5,-4) {\footnotesize{$t_5$}};
\node[right] at (5.5,-4.5) {\footnotesize{$t_6$}};
\node[right] at (5.5,-5) {\footnotesize{$t_7$}};
\node[left] at (-5.5,-0.5) {\footnotesize{$\left\{o\right\}|\left\{a,b,c\right\}$}};
\node[left] at (-5.5,-1.5) {\footnotesize{$\left\{o\right\}|\left\{a\right\}|\left\{b,c\right\}$}};
\node[left] at (-5.5,-2.5) {\footnotesize{$\left\{o\right\}|\left\{a\right\}|\left\{b\right\}|\left\{c\right\}$}};
\node[left] at (-5.5,-3.25) {\footnotesize{$\left\{o\right\}|\left\{\left\{a\right\},\left\{b\right\}\right\}|\left\{c\right\}$}};
\node[left] at (-5.5,-3.75) {\footnotesize{$\left\{o\right\}|\left\{a\right\}|\left\{b\right\}|\left\{c\right\}$}};
\node[left] at (-5.5,-4.25) {\footnotesize{$\left\{o\right\}|\left\{\left\{a\right\},\left\{c\right\}\right\}|\left\{b\right\}$}};
\node[left] at (-5.5,-4.75) {\footnotesize{$\left\{o\right\}|\left\{a\right\}|\left\{b\right\}|\left\{c\right\}$}};
\node[left] at (-2.5,-2.5) {$1$};
\node[right] at (0.5,-0.5) {$1$};
\node[right] at (1.5,-1.5) {$2$};
\node[left] at (0,-2) {$3$};
\node[right] at (3,-3) {$4$};
\node[right] at (1.5,-2.5) {$5$};
\node[left] at (-0.875,-3.25) {$6$};
\node[right] at (0.875,-3.25) {$6$};
\node[right] at (0,-4.25) {$7$};
\node[left] at (-1,-3.75) {$8$};
\node[right] at (3.875,-4.25) {$9$};
\node[left] at (-1.125,-4.25) {$9$};
\node[right] at (4,-4.75) {$10$};
\node[left] at (-1.25,-4.75) {$11$};
\end{tikzpicture}
\caption{CDM $5$}
\label{4cdms5}
\end{subfigure}
\caption{The five $4$-taxon CDMs meeting assumptions of Section~\ref{ass} before considering leaf labeling and parameter values. Convergence is drawn as curves. Epochs are separated by events represented by dashed lines on CDM $5$. For each epoch the corresponding partition or decorated partition is on the left. Epoch intervals are on the right. Parameters are labeled on sections of the edges of CDM $5$}
\label{4taxonCDMs2}
\end{figure}

The following proposition establishes identifiability of the specific parameter set.

\begin{proposition}
\label{CDMs}
The parameter set for CDM $5$ is identifiable.
\end{proposition}

See Appendix~\ref{CDM5appendix} for the proof. Since CDMs $1-4$ are nested in CDM $5$, the parameter spaces of CDMs $1-4$ can be obtained from the parameter space of CDM $5$ by fixing some parameters of CDM $5$. Thus, the proofs of identifiability of CDMs $1-4$ follow directly and are thus omitted.

\subsection{Distinguishability of \texorpdfstring{$4$}{4}-taxon CDMs}
\label{distinguish}

Accounting for all possible leaf labelings, without considering parameter values there are $27$ $4$-taxon CDMs: $3$ leaf labelings for CDM $1$ and $6$ each for CDMs $2-5$. Next we establish that the intersection of parameter spaces of any two of these distinct CDMs must be ``small''. Assuming the generating parameter is a generic point in the CDM parameter space, this property along with the nested property of our CDMs guarantees that if one of the $4$-taxon CDMs is displayed on the generating CDM then it is consistently inferred.

We establish whether the CDMs of Section~\ref{4taxaidentifiability} can be consistently inferred for generic parameters. Roughly, for two arbitrary CDMs the intersection of their parameter spaces must be ``small'' relative to the larger dimension of the two parameter spaces. For example, for MSAs two models must have sets of possible site pattern probabilities that are either disjoint or have a ``small'' intersection.

\begin{definition}
\label{distinguishability}
CDMs $\mathcal{N}_1$ and $\mathcal{N}_2$ with parameter spaces $\Theta_1$ and $\Theta_2$ are \emph{distinguishable} if $\dim\left(\Theta_1\cap{}\Theta_2\right)<\dim\left(\Theta_1\cup{}\Theta_2\right)$.
\end{definition}

\begin{proposition}
\label{distinguishabilitydiffdim}
CDMs $\mathcal{N}_1$ and $\mathcal{N}_2$ with parameter spaces $\Theta_1$ and $\Theta_2$, such that $\dim\left(\Theta_1\right)<\dim\left(\Theta_2\right)$, are distinguishable.

\begin{proof}
$\dim\left(\Theta_1\cap{}\Theta_2\right)\leq\dim\left(\Theta_1\right)<\dim\left(\Theta_2\right)\leq\dim\left(\Theta_1\cup{}\Theta_2\right)$.
\end{proof}
\end{proposition}

\begin{theorem}
\label{disttheorem}
All pairs of $4$-taxon leaf-labeled CDMs of Section~\ref{4taxaidentifiability} are distinguishable.
\end{theorem}

See Appendix~\ref{theoremdist} for the proof.

We have established the theoretical framework to accurately infer $4$-taxon CDMs. In the following sections we present algorithms to infer $N$-taxon CDMs from inferred $4$-taxon CDMs using combinatorial methods. The algorithms infer CDMs meeting the assumptions of Section~\ref{ass}, with the exception that the user can choose for polytomies of the principal tree to remain unresolved. We establish sufficient conditions for inference with these algorithms to be consistent.

In the first algorithm the topology of the $N$-taxon principal tree is inferred. Next, we infer the convergence groups. Recall that in order to be identifiable, parameters correspond to sections of edges of the principal tree, potentially spanning multiple epochs. Thus, although the convergence groups can be inferred, the specific epochs they belong to cannot. Next, since the epochs cannot be inferred, the partial order of convergence groups below the root is inferred. Note that in general only a partial order and not a total order can be inferred; some pairs of convergence groups involving disjoint converging taxon sets do not have an identifiable order. For convergence groups involving taxa corresponding with terminal edges of the principal tree, it is inferred whether the convergence group is in the tip epoch or not. Lastly, parameters are inferred to construct an $N$-taxon CDM. For all algorithms, any ties are settled at random with equal probabilities for all possible outcomes.

See Figure~\ref{schematic} for a schematic describing the process of inferring an $N$-taxon CDM from a dataset. Note that the schematic is simplified; the inference of convergence group partial orders and whether tip epochs involve convergence groups is not included. Furthermore, convergence parameters need to be adjusted if the sum of convergence parameters corresponding to an edge exceeds the edge length. We assume the input is a binary matrix, with an outgroup taxon and possibly with missing data. This could be an MSA restricted to a gene or genomic region, with the state space replaced by a binary state space, for example, the Watson-Crick base pairs $\left\{A,T\right\}\to\left\{0\right\}$ and $\left\{C,G\right\}\to\left\{1\right\}$. For MSAs, columns of the matrix are sites and rows are taxa. Alternatively, it could be ancestral/derived states or a gene presence/absence dataset. The algorithms use criteria, including a multiple comparisons correction, when inferring convergence groups on $4$-taxon CDMs to avoid overfitting. If convergence groups are falsely inferred on the $4$-taxon CDMs, then the combinatorial methods that follow can falsely infer convergence groups on the $N$-taxon CDM.

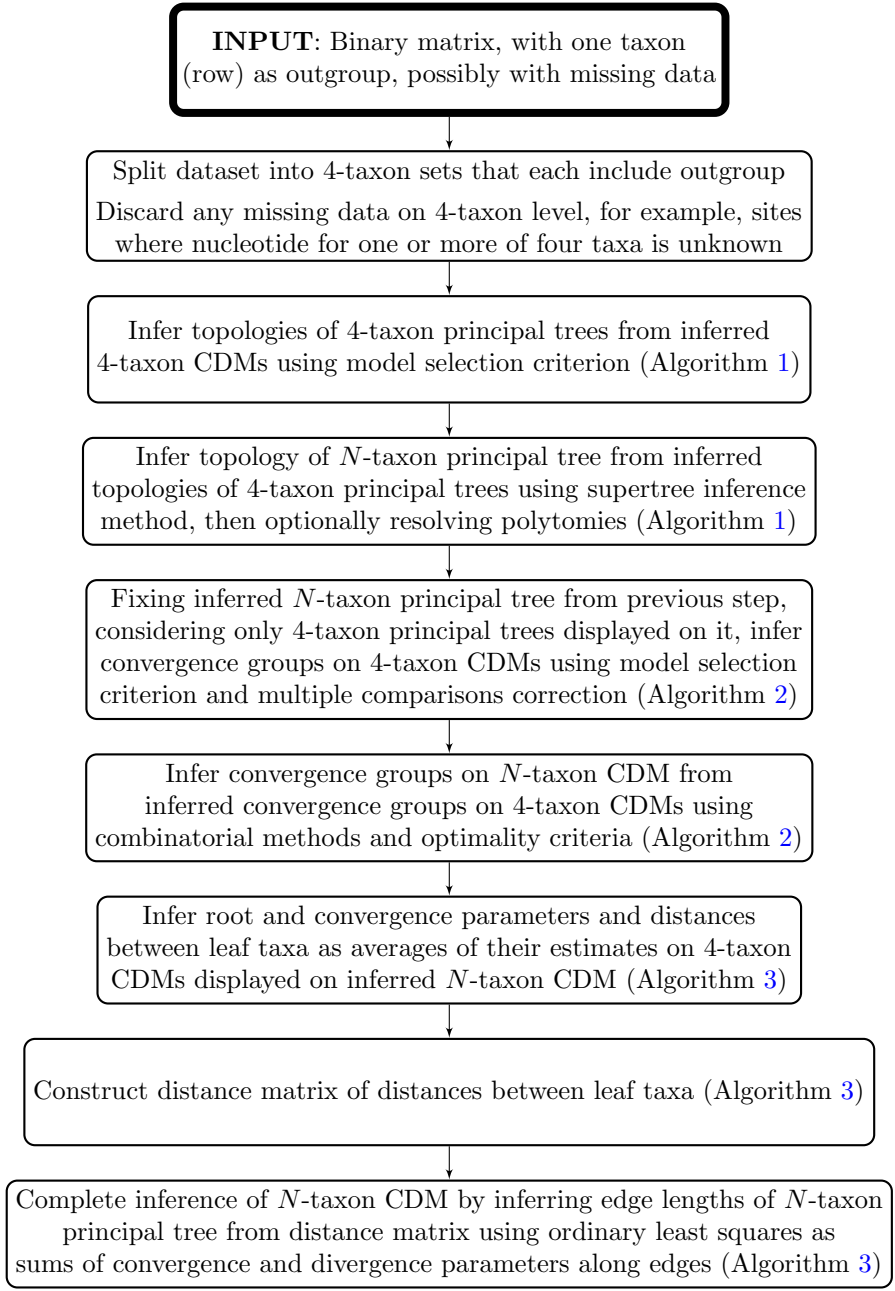
\begin{figure}[!htb]
\centering
\begin{tikzpicture}
[> = latex', auto,
block/.style ={rectangle,draw=black, line width=3pt,
align=flush center, rounded corners,
minimum height=4em},
block2/.style ={rectangle,draw=black, thick,
align=flush center, rounded corners,
minimum height=4em},
]
\node [block] at (0,0) (box1) {\textbf{INPUT}: Binary matrix, with one taxon \\
(row) as outgroup, possibly with missing data};
\node [block2,below=0.44 of box1] (box2) {Split dataset into $4$-taxon sets that each include outgroup \\
\vspace{-0.3cm} \\
Discard any missing data on $4$-taxon level, for example, sites \\
where nucleotide for one or more of four taxa is unknown};
\node [block2,below=0.44 of box2] (box3) {Infer topologies of $4$-taxon principal trees from inferred \\
$4$-taxon CDMs using model selection criterion (Algorithm~\ref{algorithmprincipaltree2})};
\node [block2,below=0.44 of box3] (box4) {Infer topology of $N$-taxon principal tree from inferred \\
topologies of $4$-taxon principal trees using supertree inference \\
method, then optionally resolving polytomies (Algorithm~\ref{algorithmprincipaltree2})};
\node [block2,below=0.44 of box4] (box5) {Fixing inferred $N$-taxon principal tree from previous step, \\
considering only $4$-taxon principal trees displayed on it, infer \\
convergence groups on $4$-taxon CDMs using model selection \\
criterion and multiple comparisons correction (Algorithm~\ref{algorithmtopCDM})};
\node [block2,below=0.44 of box5] (box6) {Infer convergence groups on $N$-taxon CDM from \\
inferred convergence groups on $4$-taxon CDMs using \\
combinatorial methods and optimality criteria (Algorithm~\ref{algorithmtopCDM})};
\node [block2,below=0.44 of box6] (box7)  {Infer root and convergence parameters and distances \\
between leaf taxa as averages of their estimates on $4$-taxon \\
CDMs displayed on inferred $N$-taxon CDM (Algorithm~\ref{algorithmmetCDM})};
\node [block2,below=0.44 of box7] (box8) {Construct distance matrix of distances between leaf taxa (Algorithm~\ref{algorithmmetCDM})};
\node [block2,below=0.44 of box8] (box9)  {Complete inference of $N$-taxon CDM by inferring edge lengths of $N$-taxon \\
principal tree from distance matrix using ordinary least squares as \\
sums of convergence and divergence parameters along edges (Algorithm~\ref{algorithmmetCDM})};
\draw[->] (box1) -- (box2);
\draw[->] (box2) -- (box3);
\draw[->] (box3) -- (box4);
\draw[->] (box4) -- (box5);
\draw[->] (box5) -- (box6);
\draw[->] (box6) -- (box7);
\draw[->] (box7) -- (box8);
\draw[->] (box8) -- (box9);
\end{tikzpicture}
\caption{The process of inferring an $N$-taxon CDM from an empirical dataset. All $4$-taxon trees and CDMs that include the outgroup are considered}
\label{schematic}
\end{figure}

\section{Inferring topologies of \texorpdfstring{$N$}{N}-taxon principal trees}
\label{topprincipaltree}

The first algorithm that we present infers the topology of the $N$-taxon principal tree. Algorithm~\ref{algorithmprincipaltree2} achieves this by inferring the topologies of all the $4$-taxon principal trees that include the outgroup, then building the the $N$-taxon principal tree from them.

We decompose a set of $N$ taxa into the $\binom{N-1}{3}$ $4$-taxon sets that include the outgroup taxon. For each $4$-taxon set, a CDM is selected from the $27$ candidates. From the set of $\binom{N-1}{3}$ $4$-taxon CDMs the topologies of the $4$-taxon principal trees are inferred, from which the topology of the $N$-taxon principal tree is inferred.

\bigskip

Inference of topologies of $4$-taxon principal trees is complicated by non-sister convergence. If not properly accounted for, non-sister convergence groups with large convergence parameters can lead to taxa descended from non-sister convergence groups being erroneously inferred as sister taxa.

Algorithm~\ref{algorithmprincipaltree2} addresses this issue in two ways. Firstly, consider a hypothetical generating CDM with principal tree $\left(o,\left(a,\left(b,c\right)\right)\right)$ and convergence group $\left\{\left\{a\right\},\left\{b\right\}\right\}$. Suppose that model selection is only performed on the three unrooted trees. Then if the convergence parameter is sufficiently large, the tree $\left(o,\left(c,\left(a,b\right)\right)\right)$ is incorrectly inferred with high probability; see Proposition~\ref{nonsistervtree} for a formal statement of this property. However, if model selection is also performed on the CDMs with non-sister convergence, then from the identifiability and distinguishability results of Section~\ref{ident}, the generating CDM is consistently inferred.

Secondly, in Algorithm~\ref{algorithmprincipaltree2} $4$-taxon sets where a single CDM fits much better than all others --- for example, according to the AIC or BIC --- are segregated from those sets where other CDMs have similar goodness of fit to the best fitting CDM. For a given $4$-taxon set, when a single CDM easily fits best, a single topology of the $4$-taxon principal tree is inferred. Otherwise, for a given $4$-taxon set, topologies of $4$-taxon principal trees on the best fitting $4$-taxon CDM and $4$-taxon CDMs with similar goodness of fit are retained for latter parts of the algorithm. A set of topologies of $N$-taxon supertrees is inferred from the retained topologies of $4$-taxon principal trees. From the set of supertrees, a consensus topology of the $N$-taxon principal tree is inferred.

\bigskip

The algorithm typically performs well when there are enough $4$-taxon sets with confidence in a single topology of the principal tree to accurately infer the topology of the $N$-taxon principal tree. However, if there are not enough such sets, the topology of the $N$-taxon principal tree may not be fully resolved, with soft polytomies. Although soft polytomies violate the first assumption of Section~\ref{ass}, it may be useful to a practitioner to choose to permit them.

Soft polytomies may be the result of non-sister convergence groups. They can arise when closely related taxa are converging on the CDM, further obscuring the already close relationships. Soft polytomies may also result from an inability to resolve phylogenetic relationships between closely related taxa despite there being no convergence. A principal tree with polytomies describes similar phylogenetic relationships to a tree of blobs \citep{allman2023tree,allman2024tinnik}. The tree of blobs represents the ``tree-like'' parts of a phylogenetic network, with the ``network-like'' parts represented by ``blobs'', which are contracted to nodes.

Finally, soft polytomies can be resolved if desired. To resolve them, we consider all topologies of $4$-taxon principal trees of $4$-taxon CDMs with the outgroup displayed on the inferred $N$-taxon principal tree. If a $4$-taxon principal tree has no polytomies, it is retained for the next part of the algorithm. Otherwise, the topology of the $4$-taxon principal tree is resolved by inferring a $4$-taxon CDM using the model selection criterion. Pairwise distances --- discussed in the next paragraph --- are then assigned to each pair of leaf taxa according to the topology of the $4$-taxon principal tree. Note that these distances are different to the definition in Section~\ref{defs}. The topology of a resolved $N$-taxon principal tree is then inferred by a clustering method --- for example, neighbor joining \citep{saitou1987neighbor} --- constrained to be a topology that can be obtained from the unresolved topology by resolving polytomies.

\bigskip

To assign pairwise distances between leaf taxa we require a tree metric. We apply the rooted triple metrization of \cite{rhodes2019topological} to the (rooted) $N$-taxon principal tree, making it ultrametric; the $N$-taxon principal tree need not have defined parameters for this step. For directed edge $e=\left(u,v\right)$, where $u$ is the parent of child $v$, the edge length is the number of descendants of $u$ minus the number of descendants of $v$; if $v$ is a tip node then it has one descendant. This parameter transformation does not influence inference of any other parts of the CDM, including parameters. The metrization is only used to infer the \emph{topology} of the $N$-taxon principal tree. Parameters are inferred in later algorithms.

The resulting tree metric is a slight modification of the tree metric $d_{RT}$ of \cite{rhodes2019topological} for rooted triples to $4$-taxon trees that all include the outgroup. For any two non-outgroup taxa, the distance is a simple linear function of the number of principal trees of displayed $4$-taxon CDMs with the outgroup where the two taxa are non-sisters --- twice this number plus two. If one of the two taxa is the outgroup, then the distance is twice the distance from the root to any leaf, $2\left(N-1\right)=2N-2$. The tree metric is described more formally in the following theorem.

\begin{theorem}[Distance on the topology of an $N$-taxon principal tree]
\label{dist}
Let $\mathcal{T}$ be a principal tree, with outgroup $o$. Suppose $\mathcal{T}$ is given the rooted triple metrization. Then the distance $d_\mathcal{T}\left(x,y\right)$ between leaf taxa $x$ and $y$ is
\begin{align*}
d_\mathcal{T}\left(x,y\right)=\begin{cases}
0 & \text{ if } $x=y$, \\
2N-2 & \text{ if $x\neq{}y$ and one of $x=o$, $y=o$}, \\
2\left|R_{x,y}\right|+2 & \text{ otherwise},
\end{cases}
\end{align*}
where $R_{x,y}$ is the set of rooted $4$-taxon principal trees displayed on $\mathcal{T}$ with outgroup $o$ displaying both $x$ and $y$, where $x$ and $y$ are non-sisters.

\end{theorem}

See Appendix~\ref{theorem6} for the proof.

\bigskip

Algorithm~\ref{algorithmprincipaltree2} for inferring the topology of the $N$-taxon principal tree follows after inputting the data. The data is a vector of counts of the distinct random variables, for example, the counts in an MSA of the $2^N$ combinations of states across the $N$ leaf taxa. The algorithm computes maximum likelihood values for the $27$ CDMs for each $4$-taxon set that includes the outgroup. We output the matrix of model selection values $\boldsymbol{M}$ from Algorithm~\ref{algorithmprincipaltree2} as it is also used in algorithms that follow.

\begin{algorithm}
\caption{Inference of topology of $N$-taxon principal tree}
\textbf{Input: }Vector $\boldsymbol{F}$ of counts of $2^N$ state combinations across $N$ leaf taxa and tolerance $\tau>0$.
\begin{enumerate}[label*=\arabic*.]
\item Initialize empty list of topologies of inferred $4$-taxon principal trees $T_Q$. Initialize empty vector of model selection criterion values $\boldsymbol{M}$.
\item For each $4$-taxon set that includes outgroup $o$:
\begin{enumerate}[label*=\arabic*.]
\item From $\boldsymbol{F}$, tally counts of $2^4$ state combinations $ijkl$, $i,j,k,l\in\left\{0,1\right\}$.
\item Compute model selection criterion values --- for example, AIC or BIC --- for all $27$ leaf-labeled CDMs, subtracting $c$ from each value so minimum is $0$. Append vector of values to $\boldsymbol{M}$ as bottom row.
\item Append topologies of $4$-taxon principal trees of CDMs with model selection criterion values below $\tau$ to $T_Q$.
\end{enumerate}
\item Use consistent supertree inference method to infer set of topologies of $N$-taxon principal trees $S$ from $T_Q$.
\item Infer consensus tree $\widehat{\mathcal{T}}'$ from $S$, rooting with $o$.
\item If $\widehat{\mathcal{T}}'$ is not resolved, either set $\widehat{\mathcal{T}}=\widehat{\mathcal{T}}'$ and terminate algorithm or resolve:
\begin{enumerate}[label*=\arabic*.]
\item Initialize empty list of topologies of $4$-taxon principal trees $T_Q'$.
\item For each $4$-taxon set that includes $o$:
\begin{enumerate}[label*=\arabic*.]
\item If topology of $4$-taxon principal tree displayed on $\widehat{\mathcal{T}}'$ is resolved, append to $T_Q'$. Otherwise, resolve by selecting CDM via model selection criterion, then append to $T_Q'$.
\end{enumerate}
\item Compute distance matrix $\boldsymbol{D}$ from $T_Q'$ using tree metric $d_\mathcal{T}$.
\item Infer $\widehat{\mathcal{T}}$ from $\boldsymbol{D}$ using consistent clustering method, constraining $\widehat{\mathcal{T}}$ to be binary tree displayed on $\widehat{\mathcal{T}}'$ after a resolution of polytomies, rooting with $o$.
\end{enumerate}
\end{enumerate}
\textbf{Output: }Topology of $N$-taxon principal tree $\widehat{\mathcal{T}}$ and $\binom{N-1}{3}\times{}27$ matrix of model selection criterion values $\boldsymbol{M}$.
\label{algorithmprincipaltree2}
\end{algorithm}

A model selection criterion --- for example, the AIC \citep{akaike1974new} or BIC \citep{schwarz1978estimating} --- is used to select a $4$-taxon CDM for each $4$-taxon set. For iid random variables from a regular exponential family, if the BIC is used and the generating model is among the candidate models then in the limit as the sample size increases the probability of selecting the generating model converges to $1$ \citep{haughton1988choice}.

We cannot assume the $4$-taxon CDM displayed on the generating $N$-taxon CDM is among the candidate models since the $4$-taxon CDM may have sister convergence groups. However, for an arbitrary $4$-taxon set that includes the outgroup, we can establish a consistency result if we assume all convergence parameters on the $N$-taxon CDM are sufficiently ``small'' and the BIC is used for model selection. We prove that as the sample size increases the probability of the topology of the principal tree of the inferred $4$-taxon CDM being identical to the topology of the principal tree of the $4$-taxon CDM displayed on the generating $N$-taxon CDM converges to $1$. Thus, it follows that as the sample size increases the probability of the inferred topology of the $N$-taxon principal tree being identical to the topology of the principal tree of the generating $N$-taxon CDM converges to $1$.

\begin{theorem}
\label{principaltreeconsistent}

Suppose CDM $\mathcal{N}$ has topology of principal tree $\mathcal{T}$. Suppose the BIC is used for model selection in step $2$ of Algorithm~\ref{algorithmprincipaltree2}. Suppose $\widehat{\mathcal{T}}$ is the estimate of $\mathcal{T}$ inferred by Algorithm~\ref{algorithmprincipaltree2}. Then there exists some constant $c>0$ such that if the largest convergence parameter of $\mathcal{N}$ is less than $c$,
\begin{align*}
\lim_{n\to\infty}\mathbb{P}\left(\widehat{\mathcal{T}}=\mathcal{T}\right)=1.
\end{align*}

\end{theorem}

See Appendix~\ref{toptreeconsistency} for a discussion on preliminary results required for the proof and Appendix~\ref{theorem7} for a proof.

Roughly, Theorem~\ref{principaltreeconsistent} ensures that if all convergence parameters are sufficiently ``small'', then the generating parameter is a small perturbation from a generic point in the parameter space of the $4$-taxon CDM that results from suppressing any sister convergence groups of the $4$-taxon CDM displayed on the generating CDM. Thus, the topology of the principal tree is inferred consistently. Such an assumption is reasonable if convergence only happens over short time periods.

\bigskip

Algorithm~\ref{algorithmprincipaltree2} allows inference of the topology of the $N$-taxon principal tree to be informed more heavily by the more confidently inferred $4$-taxon CDMs. These CDMs tend to have no convergence or convergence between taxa not closely related on the $N$-taxon CDM.

Since the topology of the $N$-taxon principal tree describes the tree-like part of the CDM, one may be tempted to ignore non-tree CDMs entirely when inferring it. That is, to perform model selection on only the three trees for each $4$-taxon set that includes the outgroup taxon. However, if there is substantial non-sister convergence, it is likely some topologies of $4$-taxon principal trees will be inferred incorrectly. In turn, leaf taxa descended from the same non-sister convergence group may be erroneously inferred as a clade on the inferred principal tree.

Even using Algorithm~\ref{algorithmprincipaltree2}, inference errors may be unavoidable if convergence parameters are sufficiently large relative to sample size. See Figure~\ref{equivalent4taxonCDMs} for an example of three CDMs with identical sets of possible phylogenetic tensors in the limit as some epoch lengths converge to $0$ or diverge to $\infty$. That is, the phylogenetic tensor as a function of the parameters of the CDM is restricted by taking the limit of some of the parameters.

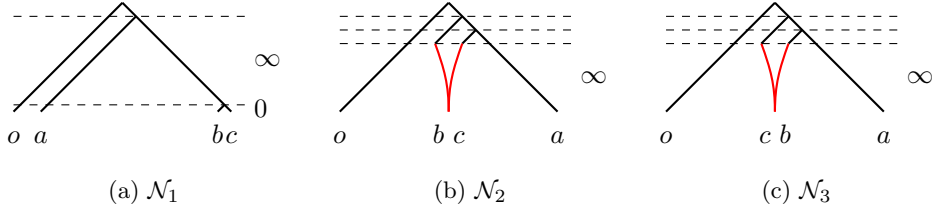
\begin{figure}[!htb]
\centering
\hspace*{\fill}
\begin{subfigure}{0.32\linewidth}
\centering
\begin{tikzpicture}[scale=0.18]
\draw[thick] (0,0) -- (-8,-8);
\draw[thick] (0,0) -- (8,-8);
\draw[thick] (1,-1) -- (-6,-8);
\draw[thick] (7.5,-7.5) -- (7,-8);
\node[below] at (-8,-8) {\strut{$o$}};
\node[below] at (-6,-8) {\strut{$a$}};
\node[below] at (7,-8) {\strut{$b$}};
\node[below] at (8,-8) {\strut{$c$}};
\draw[dashed] (-8,-1) -- (9,-1);
\draw[dashed] (-8,-7.5) -- (9,-7.5);
\node[right] at (9,-4.25) {$\infty$};
\node[right] at (9,-7.75) {$0$};
\end{tikzpicture}
\caption{$\mathcal{N}_1$}
\label{con1}
\end{subfigure}
\hfill
\begin{subfigure}{0.32\linewidth}
\centering
\begin{tikzpicture}[scale=0.18]
\draw[thick] (0,0) -- (-8,-8);
\draw[thick] (0,0) -- (8,-8);
\draw[thick] (1,-1) -- (-1,-3);
\draw[thick] (2,-2) -- (1,-3);
\draw[thick,red] (-1,-3) to[out=290,in=90] (0,-8);
\draw[thick,red] (1,-3) to[out=250,in=90] (0,-8);
\node[below] at (-8,-8) {\strut{$o$}};
\node[below] at (0,-8) {\strut{$b$ $c$}};
\node[below] at (8,-8) {\strut{$a$}};
\draw[dashed] (-8,-1) -- (9,-1);
\draw[dashed] (-8,-2) -- (9,-2);
\draw[dashed] (-8,-3) -- (9,-3);
\node[right] at (9,-5.5) {$\infty$};
\end{tikzpicture}
\caption{$\mathcal{N}_2$}
\label{con2}
\end{subfigure}
\hfill
\begin{subfigure}{0.32\linewidth}
\centering
\begin{tikzpicture}[scale=0.18]
\draw[thick] (0,0) -- (-8,-8);
\draw[thick] (0,0) -- (8,-8);
\draw[thick] (1,-1) -- (-1,-3);
\draw[thick] (2,-2) -- (1,-3);
\draw[thick,red] (-1,-3) to[out=290,in=90] (0,-8);
\draw[thick,red] (1,-3) to[out=250,in=90] (0,-8);
\node[below] at (-8,-8) {\strut{$o$}};
\node[below] at (0,-8) {\strut{$c$ $b$}};
\node[below] at (8,-8) {\strut{$a$}};
\draw[dashed] (-8,-1) -- (9,-1);
\draw[dashed] (-8,-2) -- (9,-2);
\draw[dashed] (-8,-3) -- (9,-3);
\node[right] at (9,-5.5) {$\infty$};
\end{tikzpicture}
\caption{$\mathcal{N}_3$}
\label{con3}
\end{subfigure}
\hspace*{\fill}
\caption{$4$-taxon CDMs $\mathcal{N}_1$, $\mathcal{N}_2$ and $\mathcal{N}_3$, with identical sets of possible phylogenetic tensors in the limit that epoch lengths labeled $0$ and $\infty$ converge or diverge to $0$ or $\infty$}
\label{equivalent4taxonCDMs}
\end{figure}

\begin{proposition}
\label{nonsistervtree}
Let CDM $\mathcal{N}_i$ have topology of principal tree $\mathcal{T}_i$ and $t_{i,j}$ be the epoch length of epoch $j$. Let $\mathcal{N}_1$ be the tree, with $\mathcal{T}_1=\left(o,\left(a,\left(b,c\right)\right)\right)$, $t_{1,2}\to\infty$ and $t_{1,3}\to{}0$. Let $\mathcal{N}_2$ and $\mathcal{N}_3$ be CDMs with $\mathcal{T}_2=\left(o,\left(b,\left(a,c\right)\right)\right)$ and $\mathcal{T}_3=\left(o,\left(c,\left(a,b\right)\right)\right)$, each with a non-sister convergence group $\left\{\left\{b\right\},\left\{c\right\}\right\}$ in the tip epoch, with $t_{2,4},t_{3,4}\to\infty$. Then the sets of possible phylogenetic tensors of $\mathcal{N}_1$, $\mathcal{N}_2$ and $\mathcal{N}_3$ converge to the same set.
\end{proposition}

See Appendix~\ref{equivalentCDMs} for the proof.

Proposition~\ref{nonsistervtree} is intuitive; from Theorem~\ref{convedges}, converging taxa are identical in the limit. Thus, in the limit, for all three CDMs the phylogenetic tensors are those of the tree $\left(o,\left(a,\left(b,c\right)\right)\right)$, where the epoch length of the tip epoch converges to $0$ and the epoch length of the epoch before it diverges to infinity. A similar proposition results if we consider the parametrization in terms of convergence and divergence parameters in Section~\ref{cdmsdatagen}. The proposition involves the convergence parameters of the non-tree CDMs diverging to infinity and the appropriate divergence parameters of the tree either converging to $0$ or diverging to infinity.

Although large convergence parameters leading to incorrect inference of the topology of the $4$-taxon principal tree may seem suboptimal, the limiting property is intuitive. For a finite sample size, a sufficiently large epoch length of the tip epoch on $\mathcal{N}_2$ or $\mathcal{N}_3$ results in those converging taxa having identical states --- for example, identical sequence alignments --- with probability arbitrarily close to $1$. Thus, the converging taxa appear identical and tree $\mathcal{N}_1$ fits well. In cases of similar likelihoods, $\mathcal{N}_1$ is supported more than $\mathcal{N}_2$ or $\mathcal{N}_3$ by model selection procedures since it has fewer parameters. Thus, $\mathcal{N}_1$ is erroneously inferred. The larger the epoch length of the tip epoch of $\mathcal{N}_2$ or $\mathcal{N}_3$, the larger the sample size needs to be for correct inference of the topology of the $4$-taxon principal tree with high probability.

\section{Inferring convergence groups on \texorpdfstring{$N$}{N}-taxon CDMs}
\label{topCDM}

We cannot generally identify all aspects of the $N$-taxon CDM. With epoch lengths not being identifiable, we cannot identify which epochs convergence groups belong to. However, each edge of the $N$-taxon principal tree uniquely defines a set of leaf taxa descended from it and each inferred $4$-taxon CDM uniquely defines sets of converging leaf taxa. Thus, we can identify the convergence groups of the $N$-taxon principal tree by matching the sets of converging leaf taxa on inferred $4$-taxon CDMs to edges of the $N$-taxon principal tree they descend from. Inference of convergence groups is achieved by tallying leaf taxa descended from convergence groups of inferred $4$-taxon CDMs and finding an $N$-taxon CDM with similar counts of converging leaf taxa in its convergence groups. Furthermore, we can infer relative orders of some convergence groups. Suppose $C_1=\left\{c_{1,a},c_{1,b}\right\}$ and $C_2=\left\{c_{2,a},c_{2,b}\right\}$ are convergence groups, where $c_{1,a}$, $c_{1,b}$, $c_{2,a}$ and $c_{2,b}$ are sets of taxa. If $c_{2,a}\subset{}c_{1,a}$, then $C_1$ must be in an epoch before $C_2$. Furthermore, since CDMs $4$ and $5$ of Section~\ref{4taxaidentifiability} both have two convergence groups in separate epochs, it is possible to infer relative orders of some convergence groups from the inferred $4$-taxon CDMs. Thus, we can infer a partial order on the convergence groups.

The next algorithm infers the convergence groups of the $N$-taxon CDM from the inferred convergence groups of $4$-taxon CDMs that include the outgroup and have a principal tree that is displayed on the $N$-taxon principal tree. The $N$-taxon principal tree may have polytomies. Recall that we do not permit any sister convergence groups. Thus, we do not consider any convergence groups on the $N$-taxon CDM involving edges that have a polytomous node as their shared parent node. This corresponds to no convergence groups on displayed $4$-taxon CDMs whose non-outgroup leaf taxa are all sisters --- $4$-taxon CDMs whose topology of principal tree is $\left(o,\left(a,b,c\right)\right)$, where $o$ is the outgroup. Thus, since for this algorithm we are only inferring convergence groups, these $4$-taxon sets are not considered in the algorithm.

For each $4$-taxon set to consider, we compare the nine leaf-labeled CDMs which have the required topology of the $4$-taxon principal tree (one from CDM~$1$ and two each from CDMs~$2$-$5$). We then select a $4$-taxon CDM with a model selection criterion.

We construct a matrix of ``observed'' proportions of converging quartets. Each element corresponds to a pair of leaf taxa. For each pair, we tally the inferred $4$-taxon CDMs where the pair are converging in the same non-sister convergence group (``converging quartets'') and divide by the number of $4$-taxon CDMs (``quartets'') displaying both taxa ($N-3$).

We compare the matrix of observed proportions of converging quartets to corresponding ``expected'' matrices for proposed CDMs. We select a proposed CDM with a similar expected matrix to the observed matrix. To do this, we introduce convergence groups one at a time that minimize the sum of squared differences between the observed and expected matrices. We ensure tolerance criteria are met to avoid ``overfitting'' the CDM with too many convergence groups. See Appendix~\ref{overfitting} for a discussion on avoiding overfitting the CDM.

As previously discussed, each introduced convergence group potentially induces a partial order on principal tree nodes and edges. We update partial orders when a convergence group is appended and do not consider convergence groups corresponding to edges where one is after the other. The algorithm terminates when no non-sister convergence groups not already on the CDM satisfy the assumptions of Section~\ref{ass}, the tolerance criteria and the partial orders and decrease the sum of squared differences.

\bigskip

To compute the expected proportions of converging quartets, we use the following proposition.

\begin{proposition}
\label{convratios}
For convergence group $C=\left\{c_1,c_2\right\}$ on CDM $\mathcal{N}$, let $a\in{}c_1$ and $b\in{}c_2$. Let $v$ be the MRCA node of $a$ and $b$, $X_v$ be the set of leaf taxa descending from $v$ and $X_C=c_1\cup{}c_2$. Then the expected proportion of converging quartets for $\left\{a,b\right\}$ is
\begin{align*}
\frac{\left|X_v\setminus{}X_C\right|}{N-3}=\frac{\left|X_v\right|-\left|X_C\right|}{N-3},
\end{align*}
where $\left|X_v\right|$ and $\left|X_C\right|$ are the cardinalities of sets $X_v$ and $X_C$.
\end{proposition}

See Appendix~\ref{convratiosproof} for the proof.

Although we do not attempt to infer sister convergence, it is constructive to consider a scenario with sister convergence groups.

\begin{corollary}
\label{cor}
If $C=\left\{c_1,c_2\right\}$ is a sister convergence group on CDM $\mathcal{N}$, with $a\in{}c_1$ and $b\in{}c_2$, then the expected proportion of converging quartets for $\left\{a,b\right\}$ is $0$.
\end{corollary}

The proof follows directly from Proposition~\ref{convratios} and is omitted.

Thus, non-zero proportions can be attributed to non-sister convergence groups.

The following proposition follows from Assumption~\ref{nonest} of Section~\ref{ass}, that no convergence groups can be nested in other convergence groups, since nested convergence groups share at least one pair of converging leaf taxa.

\begin{proposition}
\label{nosharedpair}
An arbitrary pair of distinct convergence groups on CDM $\mathcal{N}$ share no pair of converging leaf taxa.
\end{proposition}

See Appendix~\ref{nosharedpairproof} for the proof.

Thus, it follows that every non-zero element of the matrix of proportions of converging quartets is determined by a single convergence group on the CDM. However, given a topology of a principal tree, it is possible that a given matrix of proportions of converging quartets does not correspond to a unique set of convergence groups. For an example, see Figure~\ref{cdmssamematrix}.

\begin{figure}[!htb]
\centering
\hspace*{\fill}
\begin{subfigure}{0.32\linewidth}
\centering
\begin{tikzpicture}[scale=0.21]
\draw[thick] (0,0) -- (-9,-9);
\draw[thick] (0,0) -- (5,-5);
\draw[thick] (1,-1) -- (-3,-5);
\draw[thick] (-2,-4) -- (0.5,-6.5);
\draw[thick,red] (-3,-5) to[out=350,in=100] (-2.75,-5.5);
\draw[thick,red] (5,-5) to[out=190,in=80] (4.75,-5.5);
\draw[thick] (-2.75,-5.5) -- (-6.25,-9);
\draw[thick] (4.75,-5.5) -- (5.75,-6.5);
\draw[thick,blue] (0.5,-6.5) to[out=350,in=100] (0.75,-7);
\draw[thick,blue] (5.75,-6.5) to[out=190,in=80] (5.5,-7);
\draw[thick] (0.75,-7) -- (2.75,-9);
\draw[thick] (5.5,-7) -- (7.5,-9);
\draw[thick] (6.5,-8) -- (5.5,-9);
\node[below] at (-9,-9) {\strut{$o$}};
\node[below] at (-6.25,-9) {\strut{$a$}};
\node[below] at (2.75,-9) {\strut{$b$}};
\node[below] at (5.5,-9) {\strut{$c$}};
\node[below] at (7.5,-9) {\strut{$d$}};
\end{tikzpicture}
\caption{$\mathcal{N}_1$}
\label{cdm1}
\end{subfigure}
\hfill
\begin{subfigure}{0.32\linewidth}
\centering
\begin{tikzpicture}[scale=0.21]
\draw[thick] (0,0) -- (-9,-9);
\draw[thick] (0,0) -- (5,-5);
\draw[thick] (1,-1) -- (-3,-5);
\draw[thick] (4,-4) -- (1.5,-6.5);
\draw[thick,red] (-3,-5) to[out=350,in=100] (-2.75,-5.5);
\draw[thick,red] (5,-5) to[out=190,in=80] (4.75,-5.5);
\draw[thick] (4.75,-5.5) -- (8.25,-9);
\draw[thick] (-2.75,-5.5) -- (-3.75,-6.5);
\draw[thick,blue] (-3.75,-6.5) to[out=350,in=100] (-3.5,-7);
\draw[thick,blue] (1.5,-6.5) to[out=190,in=80] (1.25,-7);
\draw[thick] (-3.5,-7) -- (-5.5,-9);
\draw[thick] (-4.5,-8) -- (-3.5,-9);
\draw[thick] (1.25,-7) -- (-0.75,-9);
\node[below] at (-9,-9) {\strut{$o$}};
\node[below] at (-5.5,-9) {\strut{$a$}};
\node[below] at (-3.5,-9) {\strut{$b$}};
\node[below] at (-0.75,-9) {\strut{$c$}};
\node[below] at (8.25,-9) {\strut{$d$}};
\end{tikzpicture}
\caption{$\mathcal{N}_2$}
\label{cdm2}
\end{subfigure}
\hfill
\begin{subfigure}{0.32\linewidth}
\centering
\begin{align*}
\left[
\begin{array}{c|ccccc}
 & o & a & b & c & d \\
\hline
o & 0 & 0 & 0 & 0 & 0 \\
a & 0 & 0 & 0 & \frac{1}{2} & \frac{1}{2} \\
b & 0 & 0 & 0 & \frac{1}{2} & \frac{1}{2} \\
c & 0 & \frac{1}{2} & \frac{1}{2} & 0 & 0 \\
d & 0 & \frac{1}{2} & \frac{1}{2} & 0 & 0
\end{array}
\right]
\end{align*}
\caption{Matrix of proportions of converging quartets}
\label{matrix}
\end{subfigure}
\hspace*{\fill}
\caption{Distinct $5$-taxon CDMs $\mathcal{N}_1$ and $\mathcal{N}_2$ with the same topology of the principal tree and matrix of proportions of converging quartets}
\label{cdmssamematrix}
\end{figure}
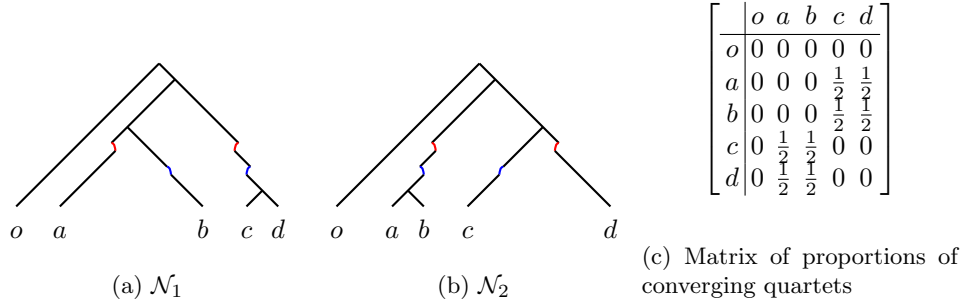

We desire to identify the set of all convergence groups. However, the matrix of proportions of converging quartets may not be sufficient to identify this set. To identify the set of all convergence groups, we consider the set of displayed $4$-taxon CDMs after suppressing sister convergence groups on the $N$-taxon CDMs, which is unique to the $N$-taxon CDM.

\begin{proposition}
\label{ident4taxoncdms}
The set of all convergence groups on CDM $\mathcal{N}$ can be identified from the set of displayed $4$-taxon CDMs after suppressing sister convergence groups.
\end{proposition}

See Appendix~\ref{ident4taxoncdmsa} for the proof.

Proposition~\ref{ident4taxoncdms} establishes that convergence groups of $\mathcal{N}$ can be identified from the inferred $4$-taxon CDMs. However, it does not guarantee that the $4$-taxon CDMs are inferred correctly up to sister convergence groups.

If there is a matrix of proportions of converging quartets that corresponds to multiple sets of convergence groups --- such as in Figure~\ref{matrix} --- the set of convergence groups of $\mathcal{N}$ can be identified from the displayed $4$-taxon CDMs. For example, consider the $4$-taxon CDMs displayed on $\mathcal{N}_1$ of Figure~\ref{cdm1} and $\mathcal{N}_2$ of Figure~\ref{cdm2} after suppressing sister convergence groups. For $\mathcal{N}_1$, the displayed $4$-taxon CDMs on taxon sets $\left\{o,a,b,c\right\}$ and $\left\{o,a,b,d\right\}$ both have two non-sister convergence groups, while the displayed $4$-taxon CDMs on $\left\{o,a,c,d\right\}$ and $\left\{o,b,c,d\right\}$ are trees. For $\mathcal{N}_2$, the displayed $4$-taxon CDMs on $\left\{o,a,b,c\right\}$ and $\left\{o,a,b,d\right\}$ are trees, while the displayed $4$-taxon CDMs on $\left\{o,a,c,d\right\}$ and $\left\{o,b,c,d\right\}$ both have two non-sister convergence groups.

Algorithm~\ref{algorithmtopCDM} for inferring the convergence groups of the $N$-taxon CDM then follows. Starting with the principal tree, convergence groups are inferred one at a time. Convergence groups are inferred with the aid of an indicator variable and two criteria at each step. We only allow a convergence group on the inferred $N$-taxon CDM if the set of inferred $4$-taxon CDMs is similar to the set of $4$-taxon CDMs after suppressing sister convergence groups displayed on the inferred $N$-taxon CDM. When inferring a convergence group on the CDM, we choose the convergence group that minimizes the sum of squared differences between the observed and expected matrices of proportions of converging quartets and satisfies our criteria.

Indicator $\delta_{ijk}$ describes whether, for the $k^{th}$ $4$-taxon set, taxon $i\in{}c_{1,a}$ and taxon $j\in{}c_{2,a}$ are converging or not converging on both an inferred $4$-taxon CDM and the $4$-taxon CDM displayed on an $N$-taxon CDM after suppressing sister convergence groups. If the taxa are converging or not converging on both $4$-taxon CDMs, $\delta_{ijk}=1$. Otherwise, $\delta_{ijk}=0$.

The two criteria are
\begin{align*}
\begin{cases}
r^{\left(a\right)}=&\frac{1}{\left|c_{1,a}\right|\left|c_{2,a}\right|\binom{N-1}{3}}\sum_{i\in{}c_{1,a}}\sum_{j\in{}c_{2,a}}\sum_{k=1}^{\binom{N-1}{3}}\delta_{ijk}, \\
s^{\left(a\right)}=&\sum_{i=1}^{N}\sum_{j=1}^{N}\left(O_{ij}^{\left(a\right)}-E_{ij}^{\left(a\right)}\right)^2.
\end{cases}
\end{align*}

$r^{\left(a\right)}$ is a measure of the average similarity of convergence groups on inferred $4$-taxon CDMs to convergence groups on $4$-taxon CDMs displayed on the proposed $N$-taxon CDM after suppressing sister convergence groups, which includes proposed convergence group $C_a$ and all other convergence groups included from previous steps. $r^{\left(a\right)}$ must be close to $1$ for $C_a$ to be included on the inferred $N$-taxon CDM.

$s^{\left(a\right)}$ is a measure of how similar the matrix of ``observed'' proportions of converging quartets is to the matrix of ``expected'' proportions of converging quartets for the CDM including $C_a$ and all other convergence groups included from previous steps. $s^{\left(a\right)}$ must be lower than the equivalent value for all other proposed convergence groups that could be included at that step --- the convergence groups that meet the assumptions of Section~\ref{ass}, meet criteria of Appendix~\ref{overfitting} to avoid overfitting and satisfy $r^{\left(a\right)}\geq{}u\in\left[0,1\right]$. Furthermore, $s^{\left(a\right)}$ must be lower than the previous value obtained from including the previous convergence group.

\begin{algorithm}
\caption{Convergence group inference}
\textbf{Input: }$N$-taxon topology of principal tree $\widehat{\mathcal{T}}$, $\binom{N-1}{3}\times{}27$ matrix of model selection criterion values $\boldsymbol{M}$ and tolerance $u\in\left[0,1\right]$.
\begin{enumerate}[label*=\arabic*.]
\item Initialize empty list of inferred $4$-taxon CDMs $L_Q$. Initialize $\left(2N-2\right)\times{}\left(2N-2\right)$ matrix $\boldsymbol{P}$ of edge partial orders on $\widehat{\mathcal{T}}$, with $\left[\boldsymbol{P}\right]_{ij}=1$ if edge $i$ ancestral to $j$ and $0$ otherwise. Initialize empty list $\widehat{\mathcal{G}}$ of convergence groups. Initialize $N\times{}N$ matrices of observed and expected proportions of converging quartets $\boldsymbol{O}$ and $\boldsymbol{E}$ as zero matrices. Initialize $N$-taxon CDM $\widehat{\mathcal{N}}$ as comprising $N$-taxon topology of principal tree $\widehat{\mathcal{T}}$ and list of convergence groups $\widehat{\mathcal{G}}$.
\item For each $4$-taxon set that includes outgroup $o$:
\begin{enumerate}[label*=\arabic*.]
\item If $4$-taxon principal tree displayed on $\widehat{\mathcal{T}}$ has no polytomies:
\begin{enumerate}[label*=\arabic*.]
\item Select CDM with $4$-taxon principal tree displayed on $\widehat{\mathcal{T}}$ with model selection criterion, using multiple comparisons correction, such as in Appendix~\ref{overfitting}, and append to $L_Q$.
\end{enumerate}
\end{enumerate}
\item Compute $\left[\boldsymbol{O}\right]_{ij}$ for all pairs of taxa $i,j$.
\item Compute initial sum of squared differences between elements of $\boldsymbol{O}$ and $\boldsymbol{E}$, $s=\sum_{i=1}^{N}\sum_{j=1}^{N}\left(\left[\boldsymbol{O}\right]_{ij}-\left[\boldsymbol{E}\right]_{ij}\right)^2$. If $s=0$, terminate algorithm.
\item For each convergence group not in $G$, for example, $C_a=\left\{c_{1,a},c_{2,a}\right\}$, with converging taxa corresponding to edges $x$ and $y$, if:
\begin{itemize}
\item $C_a$ meets assumptions of Section~\ref{ass}, and
\item $\left[\boldsymbol{P}\right]_{xy}=\left[\boldsymbol{P}\right]_{yx}=0$, and
\item constraints of Section~\ref{overfitting} to avoid overfitting are met,
\end{itemize}
then compute $r^{\left(a\right)}$, $s^{\left(a\right)}$. If no such convergence groups, terminate algorithm.
\item Of convergence groups with $r^{\left(a\right)}\geq{}u$, find minimum $s^{\left(a\right)}$. If $\min_{a}s^{\left(a\right)}<s$, include convergence group in $\widehat{\mathcal{G}}$ and set $s$ to $\min_{a}s^{\left(a\right)}$. Else, terminate algorithm.
\item Suppose edges $x$ and $y$ correspond to convergence group last included in $\widehat{\mathcal{G}}$. Update $\boldsymbol{P}$ so all edges ancestral to $x$ are ancestral to all edges descendant from $y$ and the same when swapping $x$ and $y$. If $s=0$, terminate algorithm.
\item Return to Step~$5$.
\end{enumerate}
\textbf{Output: }$N$-taxon CDM $\widehat{\mathcal{N}}$ comprising $N$-taxon topology of principal tree $\widehat{\mathcal{T}}$ and list of convergence groups $\widehat{\mathcal{G}}$, as well as $\binom{N-1}{3}\times{}27$ matrix of model selection criterion values $\boldsymbol{M}$ and matrix of edge partial orders $\boldsymbol{P}$.
\label{algorithmtopCDM}
\end{algorithm}

We cannot establish consistency of inference of the convergence groups under the assumptions of Section~\ref{ass}. This is because $4$-taxon CDMs displayed on the generating $N$-taxon CDM may have sister convergence groups. We cannot discount the possibility that sister convergence groups may not lead to consistent inference of the non-sister convergence groups. We can, however, establish a set of sufficient conditions for inference of the convergence groups to be consistent.

\begin{theorem}
\label{conscongroups}
Suppose CDM $\mathcal{N}$ has topology of principal tree $\mathcal{T}$ and convergence groups $\mathcal{G}$. Suppose for all $l$, $\alpha_l=\beta_l$. Suppose for convergence group $\mathcal{C}_i=\left\{c_{1,i},c_{2,i}\right\}$ that if $a\in{}c_{1,i}\cup{}c_{2,i}$, then $a\notin{}c_{1,j}\cup{}c_{2,j}$ for any $j\neq{}i$. Suppose $\mathcal{T}$ is input into Algorithm~\ref{algorithmtopCDM}, the BIC is used for model selection in step~2, there are no multiple comparisons corrections and the tolerance criterion is $u=1$. Suppose $\widehat{\mathcal{G}}$ is the estimate of $\mathcal{G}$ inferred by Algorithm~\ref{algorithmtopCDM}. Then there exists some constant $c>0$ such that if the largest convergence parameter of $\mathcal{N}$ is less than $c$,
\begin{align*}
\lim_{n\to\infty}\mathbb{P}\left(\widehat{\mathcal{G}}=\mathcal{G}\right)=1.
\end{align*}
\end{theorem}

See Appendix~\ref{theorem13} for the proof.

Theorem~\ref{conscongroups} ensures consistent inference of the set of convergence groups if all convergence parameters are ``small'', the Markov model of the generating CDM is the $2$-state binary symmetric model and there are no taxa in multiple convergence groups. Such assumptions are reasonable if substitutions between the two states are expected to be approximately equal and convergence is ``rare'', with a small number of convergence groups on the CDM. Note that we do not generally assume that $\alpha_l=\beta_l$ for each convergence-divergence group $l$ when inferring the $N$-taxon CDM.

\section{Inferring parameters of \texorpdfstring{$N$}{N}-taxon CDMs}
\label{metricCDM}

With the topology of the principal tree and convergence groups of the $N$-taxon CDM inferred, all that is left to infer is the positions of the convergence groups and the parameters. Inference of positions of the convergence groups involves inferring partial orders of convergence groups and determining which convergence groups are in the tip epoch. For more on inferring partial orders of convergence groups and tip epoch convergence groups, see Algorithms~\ref{cdmconorder}~and~\ref{cdmdiv} of Appendix~\ref{grouporders}.

We do not attempt to infer all parameters of the $N$-taxon CDM as they are not all identifiable. Instead, we infer all convergence parameters, the root parameter and all edge lengths of the (unrooted) $N$-taxon principal tree.

Taxa converging increase the probabilities of combinations of states where those taxa have the same state. We may reasonably expect that convergence should \emph{decrease} a distance between taxa if that distance was to reflect how similar the random variables --- for example, sequences --- are to each other. However, recall that the distance between taxa is the sum of convergence and divergence parameters along the shortest path between the two taxa. Thus, the distances on the $N$-taxon principal tree do not necessarily reflect how similar the random variables are to each other.

\begin{proposition}
\label{distident}
All edge lengths of the principal tree of each of CDM $1-5$ are identifiable.
\end{proposition}

See Appendix~\ref{prop14} for the proof. Note that in the proof of Proposition~\ref{distident} we establish that all pairwise distances between leaf taxa are also identifiable.

\begin{proposition}
\label{conident}
All convergence parameters of each of CDM $2-5$ are identifiable.
\end{proposition}

See Appendix~\ref{prop15} for the proof.

\begin{proposition}
\label{rootident}
The root parameter $\gamma=\left[\boldsymbol{\Pi}\right]_0-\left[\boldsymbol{\Pi}\right]_1$, where $\left[\boldsymbol{\Pi}\right]_0$ and $\left[\boldsymbol{\Pi}\right]_1$ are the probabilities of states $0$ and $1$ at the root, respectively, is identifiable on each of CDM $1-5$.
\end{proposition}

See Appendix~\ref{prop16} for the proof.

\bigskip

For $4$-taxon sets with the outgroup, we compute maximum likelihood estimates of the $4$-taxon CDM convergence and root parameters and all six of the pairwise distances between taxa. We average root parameter estimates over all $4$-taxon sets and average convergence parameter estimates over all $4$-taxon sets displaying the convergence parameters. For each pair of taxa we average distances over all $4$-taxon sets displaying the two taxa to form a vector of pairwise distances. From the distances, a consistent method --- ordinary least squares --- is used to infer the parameters of the $N$-taxon principal tree, fixing the topology of the principal tree to be that already known or inferred by Algorithm~\ref{algorithmtopCDM} and rooting with the outgroup. If the $N$-taxon principal tree has no polytomies, inference of the principal tree is complete up to the precise root location.

If the $N$-taxon principal tree has polytomies, we resolve each one, as in Algorithm~\ref{algorithmprincipaltree2}, before computing pairwise distances between leaf taxa. A resolved $N$-taxon principal tree with edge lengths is then inferred, as described above. Paths on the resolved tree are identified that correspond to edges below polytomies of the unresolved tree. The polytomies are then reintroduced, with lengths of the identified edges on the unresolved tree estimated as the sums of parameters on the corresponding paths of the resolved tree. Inference of the $N$-taxon principal tree is then complete.

Algorithm~\ref{algorithmmetCDM} for inferring the $N$-taxon CDM parameters then follows. Note that Algorithm~\ref{algorithmmetCDM} takes the matrix of expected convergence group orders $\boldsymbol{E}$ and vector of tip epoch convergence groups $\boldsymbol{D}$ computed in Algorithms~\ref{cdmconorder}~and~\ref{cdmdiv} of Appendix~\ref{grouporders} as input. $\boldsymbol{E}$ is a binary matrix, with $\left[\boldsymbol{E}\right]_{ij}=1$ if the epoch of convergence group $i$ is before that of convergence group $j$ and $0$ otherwise. $\boldsymbol{D}$ is a binary vector, with $\left[\boldsymbol{D}\right]_i=1$ if convergence group $i$ is in the tip epoch and $\left[\boldsymbol{D}\right]_i=0$ if it is not. Note that since there can be at most one convergence group in the tip epoch, there is at most one non-zero element of $\boldsymbol{D}$. Note that after applying Algorithms~\ref{cdmconorder}~and~\ref{cdmdiv} of Appendix~\ref{grouporders}, at step $6$ of Algorithm~\ref{algorithmmetCDM} there is only one out of the $27$ possible topologies of principal trees for each $4$-taxon CDM that includes the outgroup.

Algorithm~\ref{algorithmmetCDM} takes as input either all $4$-taxon sets with the outgroup or only the $4$-taxon sets with the outgroup corresponding to the $4$-taxon CDMs displayed on the generating $N$-taxon CDM that have no sister convergence. For the latter, a consistency result can be obtained since, with no sister convergence, the probability of the inferred $4$-taxon CDM being the $4$-taxon CDM displayed on the generating $N$-taxon CDM converges to $1$. However, for the latter, some pairwise distances between leaf taxa may not be estimated since some $4$-taxon sets are not considered. Matrix $\boldsymbol{X}$ describes the edges of principal tree $\mathcal{T}$ that are traversed to compute pairwise distances between taxa. $\left[\boldsymbol{X}\right]_{ij}=1$ if the distance between the $i^{th}$ taxon pair --- ordered arbitrarily --- is computed by traversing edge $j$ --- ordered arbitrarily --- of $\mathcal{T}$. Otherwise, $\left[\boldsymbol{X}\right]_{ij}=0$. Row $i$ of $\boldsymbol{X}$ is removed in step $6$ of Algorithm~\ref{algorithmmetCDM} if the pairwise distance between the $i^{th}$ taxon pair is not computed on any of the $4$-taxon CDMs; this pairwise distance is also removed from the computation. Thus, $\boldsymbol{X}^T\boldsymbol{X}$ may not be invertible and there may not be a unique solution for the principal tree edge lengths in the ordinary least squares computation. For an example, consider the $5$-taxon CDM with topology of principal tree $\left(o,\left(a,\left(b,\left(c,d\right)\right)\right)\right)$ and convergence groups $C_1=\left\{\left\{a\right\},\left\{c,d\right\}\right\}$, $C_2=\left\{\left\{b\right\},\left\{c\right\}\right\}$ and $C_3=\left\{\left\{b\right\},\left\{d\right\}\right\}$. See Figure~\ref{5taxonexample} for the CDM and its displayed CDMs. Only one displayed $4$-taxon CDM does not have sister convergence, the CDM with topology of principal tree $\left(o,\left(b,\left(c,d\right)\right)\right)$ and convergence groups $C_2'=\left\{\left\{b\right\},\left\{c\right\}\right\}$ and $C_3'=\left\{\left\{b\right\},\left\{d\right\}\right\}$. Thus, only six pairwise distances can be estimated, despite the principal tree having seven edges. Thus, there is no unique solution for the principal tree edge lengths. In this scenario, we estimate pairwise distances from all displayed $4$-taxon CDMs after suppressing any sister convergence groups. However, this is unlikely to be problematic when $N$ is large as the $\binom{N}{2}$ pairs of taxa vastly outnumbers the $2N-3$ edges.

\begin{figure}[!htb]
\hspace*{\fill}
\begin{subfigure}{0.19\linewidth}
\centering
\begin{tikzpicture}[scale=0.154]
\draw[thick] (0,0) -- (-6.5,-6.5);
\draw[thick] (0,0) -- (3,-3);
\draw[thick] (1,-1) -- (-1,-3);
\draw[thick] (2,-2) -- (-0.5,-4.5);
\draw[thick,red] (-1,-3) to[out=350,in=100] (-0.75,-3.5);
\draw[thick,red] (3,-3) to[out=190,in=80] (2.75,-3.5);
\draw[thick] (-0.75,-3.5) -- (-3.75,-6.5);
\draw[thick] (2.75,-3.5) -- (4.75,-5.5);
\draw[thick] (3.25,-4) -- (2.75,-4.5);
\draw[thick,blue] (-0.5,-4.5) to[out=350,in=100] (-0.25,-5);
\draw[thick,blue] (2.75,-4.5) to[out=190,in=80] (2.5,-5);
\draw[thick] (-0.25,-5) -- (-0.75,-5.5);
\draw[thick] (2.5,-5) -- (1,-6.5);
\draw[thick,teal] (-0.75,-5.5) to[out=350,in=100] (-0.5,-6);
\draw[thick,teal] (4.75,-5.5) to[out=190,in=80] (4.5,-6);
\draw[thick] (-0.5,-6) -- (-1,-6.5);
\draw[thick] (4.5,-6) -- (5,-6.5);
\node[below] at (-6.5,-6.5) {\strut{$o$}};
\node[below] at (-3.75,-6.5) {\strut{$a$}};
\node[below] at (-1,-6.5) {\strut{$b$}};
\node[below] at (1,-6.5) {\strut{$c$}};
\node[below] at (5,-6.5) {\strut{$d$}};
\end{tikzpicture}
\end{subfigure}
\hfill
\begin{subfigure}{0.19\linewidth}
\centering
\begin{tikzpicture}[scale=0.2]]
\draw[thick] (0,0) -- (-5,-5);
\draw[thick] (0,0) -- (3,-3);
\draw[thick] (1,-1) -- (-1,-3);
\draw[thick] (2,-2) -- (0,-4);
\draw[thick,red] (-1,-3) to[out=350,in=100] (-0.75,-3.5);
\draw[thick,red] (3,-3) to[out=190,in=80] (2.75,-3.5);
\draw[thick] (-0.75,-3.5) -- (-2.25,-5);
\draw[thick] (2.75,-3.5) -- (3.25,-4);
\draw[thick,blue] (0,-4) to[out=350,in=100] (0.25,-4.5);
\draw[thick,blue] (3.25,-4) to[out=190,in=80] (3,-4.5);
\draw[thick] (0.25,-4.5) -- (-0.25,-5);
\draw[thick] (3,-4.5) -- (3.5,-5);
\node[below] at (-5,-5) {\strut{$o$}};
\node[below] at (-2.25,-5) {\strut{$a$}};
\node[below] at (-0.25,-5) {\strut{$b$}};
\node[below] at (3.5,-5) {\strut{$c$}};
\end{tikzpicture}
\end{subfigure}
\hfill
\begin{subfigure}{0.19\linewidth}
\centering
\begin{tikzpicture}[scale=0.2]]
\draw[thick] (0,0) -- (-5,-5);
\draw[thick] (0,0) -- (3,-3);
\draw[thick] (1,-1) -- (-1,-3);
\draw[thick] (2,-2) -- (0,-4);
\draw[thick,red] (-1,-3) to[out=350,in=100] (-0.75,-3.5);
\draw[thick,red] (3,-3) to[out=190,in=80] (2.75,-3.5);
\draw[thick] (-0.75,-3.5) -- (-2.25,-5);
\draw[thick] (2.75,-3.5) -- (3.25,-4);
\draw[thick,teal] (0,-4) to[out=350,in=100] (0.25,-4.5);
\draw[thick,teal] (3.25,-4) to[out=190,in=80] (3,-4.5);
\draw[thick] (0.25,-4.5) -- (-0.25,-5);
\draw[thick] (3,-4.5) -- (3.5,-5);
\node[below] at (-5,-5) {\strut{$o$}};
\node[below] at (-2.25,-5) {\strut{$a$}};
\node[below] at (-0.25,-5) {\strut{$b$}};
\node[below] at (3.5,-5) {\strut{$d$}};
\end{tikzpicture}
\end{subfigure}
\hfill
\begin{subfigure}{0.19\linewidth}
\centering
\begin{tikzpicture}[scale=0.2]
\draw[thick] (0,0) -- (-5,-5);
\draw[thick] (0,0) -- (2,-2);
\draw[thick] (1,-1) -- (0,-2);
\draw[thick,red] (0,-2) to[out=350,in=100] (0.25,-2.5);
\draw[thick,red] (2,-2) to[out=190,in=80] (1.75,-2.5);
\draw[thick] (0.25,-2.5) -- (-2.25,-5);
\draw[thick] (1.75,-2.5) -- (4.25,-5);
\draw[thick] (2.25,-3) -- (0.25,-5);
\node[below] at (-5,-5) {\strut{$o$}};
\node[below] at (-2.25,-5) {\strut{$a$}};
\node[below] at (0.25,-5) {\strut{$c$}};
\node[below] at (4.25,-5) {\strut{$d$}};
\end{tikzpicture}
\end{subfigure}
\hfill
\begin{subfigure}{0.19\linewidth}
\centering
\begin{tikzpicture}[scale=0.2]
\draw[thick] (0,0) -- (-5,-5);
\draw[thick] (0,0) -- (4,-4);
\draw[thick] (1,-1) -- (-1,-3);
\draw[thick] (2,-2) -- (1,-3);
\draw[thick,blue] (-1,-3) to[out=350,in=100] (-0.75,-3.5);
\draw[thick,blue] (1,-3) to[out=190,in=80] (0.75,-3.5);
\draw[thick] (-0.75,-3.5) -- (-1.25,-4);
\draw[thick] (0.75,-3.5) -- (-0.75,-5);
\draw[thick,teal] (-1.25,-4) to[out=350,in=100] (-1,-4.5);
\draw[thick,teal] (4,-4) to[out=190,in=80] (3.75,-4.5);
\draw[thick] (-1,-4.5) -- (-1.5,-5);
\draw[thick] (3.75,-4.5) -- (4.25,-5);
\node[below] at (-5,-5) {\strut{$o$}};
\node[below] at (-1.5,-5) {\strut{$b$}};
\node[below] at (-0.75,-5) {\strut{$c$}};
\node[below] at (4.25,-5) {\strut{$d$}};
\end{tikzpicture}
\end{subfigure}
\hspace*{\fill}
\caption{A $5$-taxon CDM (left) and the displayed $4$-taxon CDMs with the outgroup}
\label{5taxonexample}
\end{figure}
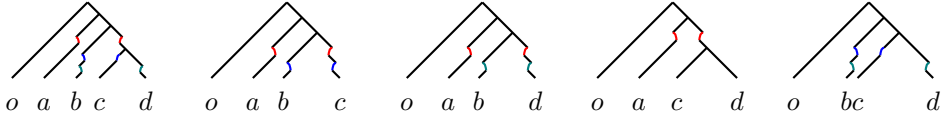

\begin{algorithm}
\small
\caption{$N$-taxon CDM inference}
\textbf{Input: }$N$-taxon CDM $\widehat{\mathcal{N}}$ comprising $N$-taxon topology of principal tree $\widehat{\mathcal{T}}$ and list of convergence groups $\widehat{\mathcal{G}}$, $\binom{N-1}{3}\times{}27$ matrix of model selection criterion values $\boldsymbol{M}$, matrix of expected convergence group orders $\boldsymbol{E}$ and vector of tip epoch convergence groups $\boldsymbol{D}$.
\begin{enumerate}[label*=\arabic*.]
\item Initialize column vector $\widehat{\boldsymbol{d}}$ of length $\binom{N}{2}$ of pairwise distances between leaf taxa as zero vector, with arbitrary order. Initialize binary $\binom{N}{2}\times{}\left(2N-3\right)$ matrix $\boldsymbol{X}$ as matrix of zeros, with row and column orders corresponding to orders of $\widehat{\boldsymbol{d}}$ and arbitrary edge labelings of (unrooted) $\widehat{\mathcal{T}}$.
\item Set $\left[\boldsymbol{X}\right]_{ij}=1$ if distance between taxon pair corresponding with $\left[\widehat{\boldsymbol{d}}\right]_i$ is computed by traversing edge $j$ of $\widehat{\mathcal{T}}$.
\item If $\widehat{\mathcal{T}}$ has polytomies, form $\widehat{\mathcal{T}}'$ by resolving them using step $5$ of Algorithm~\ref{algorithmprincipaltree2}. Otherwise, let $\widehat{\mathcal{T}}'=\widehat{\mathcal{T}}$. Form $\widehat{\mathcal{N}}'$ by replacing $\widehat{\mathcal{T}}$ with $\widehat{\mathcal{T}}'$.
\item Consider either: 1) each $4$-taxon set that includes $o$, suppressing any sister convergence groups of $4$-taxon CDMs displayed on $\mathcal{N}$, or 2) only $4$-taxon sets that include $o$ and for which $4$-taxon CDMs displayed on $\mathcal{N}$ have no sister convergence. Say arbitrary $4$-taxon set is $\left\{o,a,b,c\right\}$. For the $4$-taxon CDM:
\begin{enumerate}[label*=\arabic*.]
    \item Compute maximum likelihood estimate of $\gamma$.
    \item Compute maximum likelihood estimates of convergence parameters.
    \item Compute $d_{o,a}$, $d_{o,b}$, $d_{o,c}$, $d_{a,b}$, $d_{a,c}$ and $d_{b,c}$ as sums of parameter maximum likelihood estimates on shortest paths.
\end{enumerate}
\item Estimate $\widehat{\gamma}$, vector of convergence parameters $\widehat{\boldsymbol{v}}$ and elements of $\widehat{\boldsymbol{d}}$ as means of their estimates.
\item If element $\left[\widehat{\boldsymbol{d}}\right]_i$ has not been estimated, remove that element of $\widehat{\boldsymbol{d}}$ and row $i$ of $\boldsymbol{X}$.
\item Fixing topology of $\widehat{\mathcal{T}}'$, infer edges lengths $\widehat{\boldsymbol{l}}$ using ordinary least squares, $\widehat{\boldsymbol{l}}=\left(\boldsymbol{X}^T\boldsymbol{X}\right)^{-1}\boldsymbol{X}^T\widehat{\boldsymbol{d}}$, setting any negative elements of $\widehat{\boldsymbol{l}}$ to $0$. This step may not be possible if 2) is chosen in step $4$. In this case, return to step $4$ and choose 1).
\item If $\widehat{\mathcal{T}}$ has no polytomies:
\begin{itemize}
\item[] Let $\widehat{\mathcal{T}}=\widehat{\mathcal{T}}'$ be $N$-taxon principal tree.
\end{itemize}
Else:
\begin{itemize}
\item[] Form $\widehat{\mathcal{T}}''$ by reintroducing polytomies to $\widehat{\mathcal{T}}'$. For each edge $e_i$ of $\widehat{\mathcal{T}}''$ whose parent node is a polytomy, set length of $e_i$ to sum of parameters along path on $\widehat{\mathcal{T}}'$ between two nodes with same sets of descendant leaf taxa as two nodes $e_i$ is incident on. Let $\widehat{\mathcal{T}}=\widehat{\mathcal{T}}''$.
\end{itemize}
\item If there are edges of $\widehat{\mathcal{T}}$ shorter than sum of convergence parameters corresponding to sections of the edges, choose one such edge arbitrarily and rescale all such convergence parameters so they sum to edge length. Noting that convergence parameters corresponding to sections of other edges --- there are two edges corresponding to each convergence group --- have decreased, repeat process on remaining edges.
\item Let $\widehat{\mathcal{N}}$ be $N$-taxon CDM, with principal tree $\widehat{\mathcal{T}}$, root parameter $\widehat{\gamma}$, convergence groups $\widehat{\mathcal{G}}$, vector of convergence parameters $\widehat{\boldsymbol{v}}$, matrix of expected convergence group orders $\boldsymbol{E}$ and vector of tip epoch convergence groups $\boldsymbol{D}$.
\end{enumerate}
\textbf{Output: }$N$-taxon CDM $\widehat{\mathcal{N}}$.
\label{algorithmmetCDM}
\end{algorithm}

\begin{theorem}
\label{consparams}

Suppose CDM $\mathcal{N}$ has topology of principal tree $\mathcal{T}$, convergence groups $\mathcal{G}$, principal tree edge lengths $\boldsymbol{l}$, root parameter $\gamma$ and convergence parameters $\boldsymbol{v}$. Suppose $\mathcal{T}$, $\mathcal{G}$, convergence group partial orders and tip epoch convergence groups of $\mathcal{N}$ are input into Algorithm~\ref{algorithmmetCDM}. Suppose in step~4 of Algorithm~\ref{algorithmmetCDM} only $4$-taxon sets for which $4$-taxon CDMs displayed on $\mathcal{N}$ have no sister convergence are considered. Suppose that for each convergence group of $\mathcal{G}$ --- say $C_a=\left\{c_{1,a},c_{2,a}\right\}$ --- there is at least one $4$-taxon CDM displayed on $\mathcal{N}$ with no sister convergence where $x\in{}c_{1,a}$, $y\in{}c_{2,a}$ are non-sister leaf taxa on the displayed CDM. Suppose further that matrix $\boldsymbol{X}$ in step $6$ of Algorithm~\ref{algorithmmetCDM} has rank $2N-3$. Suppose $\widehat{\boldsymbol{l}}$, $\widehat{\gamma}$ and $\widehat{\boldsymbol{v}}$ are the estimates of $\boldsymbol{l}$, $\gamma$ and $\boldsymbol{v}$, respectively, inferred by Algorithm~\ref{algorithmmetCDM}. Then for any $\epsilon_>0$,
\begin{align*}
\lim_{n\to\infty}\mathbb{P}\left(\left|\widehat{\boldsymbol{l}}-\boldsymbol{l}\right|>\epsilon\right)=0,\quad{}\lim_{n\to\infty}\mathbb{P}\left(\left|\widehat{\gamma}-\gamma\right|>\epsilon\right)=0,\quad{}\lim_{n\to\infty}\mathbb{P}\left(\left|\widehat{\boldsymbol{v}}-\boldsymbol{v}\right|>\epsilon\right)=0,
\end{align*}
where $\left|\widehat{\boldsymbol{l}}-\boldsymbol{l}\right|$ and $\left|\widehat{\boldsymbol{v}}-\boldsymbol{v}\right|$ involve $l^1$ norms.

\end{theorem}

See Appendix~\ref{theorem17} for the proof.

\section{Discussion}

Phylogenetic networks usually assume instantaneous gene flow between taxa. However, many biological processes can lead to gradual gene flow over time, for example, introgressive hybridization. Isolation with migration models \citep{hey2010isolation} model migration of individuals between populations over a period of time. ABBA-BABA tests \citep{green2010draft} test for gene flow between populations over time. However, both isolation with migration models and ABBA-BABA tests are limited to small numbers of taxa. We have developed phylogenetic models of gene flow between taxa over time that can be applied to large datasets. Convergence-divergence models are generalizations of phylogenetic trees for many-taxon datasets. In contrast to phylogenetic networks, they have a single ``principal tree''. A Markov model describes independent divergence of taxa on the principal tree from common ancestors. However, the Markov model also describes convergence of some previously diverged taxa over a time interval.

Genotypic replicated evolution is the independent evolution in distinct taxa of genotypic similarities --- for example, at nucleotide sites --- from similar selective pressures. This violates the assumptions of phylogenetic trees of independent divergence of taxa from common ancestors. Replicated evolution can lead to the gradual convergence of taxa. This process is not appropriately modeled by phylogenetic networks, but can be modeled by our convergence-divergence models. A gene ``principal tree'' can be inferred for a specific gene or genomic window. On the gene or genomic window, replicated evolution can be modeled by convergence of taxa over a time interval.

There are many types of datasets that a researcher may have access to --- for example, gene presence/absence datasets. Genes being formed and lost on independently diverging taxa could be modeled by a Markov model on a phylogenetic tree. However, some taxa may show similar patterns of presence/absence of some genes due to gene flow over a time interval. This gene flow over time leading to similar presence/absence patterns of genes in some taxa can be modeled by convergence-divergence models.

\backmatter

\bmhead{Supplementary Information} Supplementary material is available at \url{https://github.com/jonathanmitchell88/CDMsSI}.

\bmhead{Acknowledgements}

We thank Jeremy Sumner for helpful feedback on the manuscript.



\section*{Declarations}

\bmhead{Funding} This work was funded by The Australian Research Council Centre of Excellence for Plant Success in Nature and Agriculture (CE200100015).
\bmhead{Conflict of interest} The authors declare no conflict of interest.

\putbib
\end{bibunit}

\newpage

\newcommand{\nphantom}[1]{\sbox0{#1}\hspace{-\the\wd0}}

\setcounter{theorem}{0}

\renewcommand{\thetheorem}{\arabic{theorem}A}

\setcounter{section}{0}

\setcounter{proposition}{0}

\renewcommand\theproposition{\arabic{proposition}A}

\setcounter{algorithm}{0}

\renewcommand\thealgorithm{\arabic{algorithm}A}

\setcounter{equation}{0}

\renewcommand\theequation{\arabic{equation}A}

\setcounter{figure}{0}

\renewcommand\thefigure{\arabic{figure}A}

\begin{appendices}

\begin{bibunit}

\section{Parameter identifiability}
\label{pars}

Recall that we decompose edges of the principal tree into ``diverging sections'' and ``converging sections''. Converging sections span only a single epoch, while diverging sections may span multiple epochs.

Recall that rates and epoch times cannot be identified individually; only their products can be identified. A parameter associated with a diverging section of an edge than spans a single epoch cannot be identified. Instead, an ``average'' over the maximum number of epochs a contiguous diverging section can span can be identified. For example, suppose rate matrix $\boldsymbol{Q}_1$ applies over epoch time $t_1$ to a diverging section of an edge immediately before an event and rate matrix $\boldsymbol{Q}_2$ applies over epoch time $t_2$ to a diverging section of the edge immediately after the event. Then for the $2$-state general Markov model, $\exp\left(\boldsymbol{\widehat{Q}}\left(t_1+t_2\right)\right)=\exp\left(\boldsymbol{Q}_2t_2\right)\exp\left(\boldsymbol{Q}_1t_1\right)$, where $\boldsymbol{\widehat{Q}}$ is again a rate matrix from the $2$-state general Markov model. Thus, we apply rate matrix $\boldsymbol{\widehat{Q}}$ to both diverging sections of the edge.

This lack of identifiability result follows from the $2$-state general Markov model forming a Lie algebra, sufficient for multiplicative closure of the model class \citep{sumner2012lie}. Suppose $\alpha_1$ and $\beta_1$ and $\alpha_2$ and $\beta_2$ correspond with rate matrices $\boldsymbol{Q}_1$ and $\boldsymbol{Q}_2$, respectively. Then by Definition~\ref{CDM}, $\frac{\alpha_1}{\beta_1}=\frac{\alpha_2}{\beta_2}$. It is straightforward to show that if $\widehat{\alpha}$ and $\widehat{\beta}$ are associated with $\boldsymbol{\widehat{Q}}$, then $\frac{\widehat{\alpha}}{\widehat{\beta}}=\frac{\alpha_1}{\beta_1}=\frac{\alpha_2}{\beta_2}$. Thus, the product of the two transition matrices is replaced by a single ``average'' transition matrix.

All parameters except for those corresponding to the root distribution are of the form $l_i=a_i+b_i=\alpha_{i}t_i+\beta_{i}t_i=a_i\left(1+\frac{b_j}{a_j}\right)$, where $i$ and $j$ are arbitrary parameter indices, $a_i=\alpha_it_i$ and $b_i=\beta_it_i$. (Note that these parameters are scalars, whereas $\exp\left(\boldsymbol{\widehat{Q}}\left(t_1+t_2\right)\right)$ is a matrix.) Since $\frac{\alpha_i}{\beta_i}=\frac{\alpha_j}{\beta_j}$, it follows that $\frac{a_i}{b_i}=\frac{a_j}{b_j}$.

Contiguous diverging sections of an edge --- not separated by a converging section --- each have a single associated parameter $l_i$. Furthermore, each convergence group has an associated parameter $l_j$, in common for all converging sections of edges in the convergence group. In addition to parameters describing the convergence groups and contiguous diverging sections, there is a parameter $\gamma=\left[\boldsymbol{\Pi}\right]_0-\left[\boldsymbol{\Pi}\right]_1=\frac{-a_i+b_i}{a_i+b_i}$ describing the difference in probabilities of states $0$ and $1$ on the root taxon.

To form the set of parameters of a CDM, we consider a particular unique set of diverging and converging sections. Since differences in parameters between contiguous diverging sections cannot be identified, the diverging sections we consider are those sections on the principal tree between a node or converging section and another node or converging section. Furthermore, since the exact root location on the outgroup edge is not identifiable, we consider one diverging section to be the entire outgroup edge when the principal tree of the CDM is unrooted. The converging sections correspond to individual epochs where there is a convergence group. Converging sections correspond to convergence parameters and diverging sections correspond to divergence parameters.

Note that although this is the general formulation of the parameter space, on a given CDM not all parameters are necessarily identifiable. To obtain an identifiable set of parameters some combinations of the divergence parameters may be required, which we describe in Appendix~\ref{CDM5appendix}. For the following sections, the parameters $x_i=\exp\left(-l_i\right)\in\left(0,1\right)$ and the variants $y_i$ and $z_i$ are used for establishing identifiability and distinguishability of CDMs.

\section{Limiting behavior of converging taxa}
\label{lim}

\subsection{Proof of Proposition~\ref{p1}}
\label{prop1}

Before proving the claim, we introduce some notation. Using similar notation to \cite{sumner2012algebra}, for some arbitrary integer $l\geq{}1$ and $\boldsymbol{X}\in\left\{\boldsymbol{L}_\alpha,\boldsymbol{L}_\beta\right\}$,
\begin{align*}
\boldsymbol{X}^{\left(A\right)}=\prod_{i\in{}A}\boldsymbol{X}^{\left(i\right)},
\end{align*}
where $\boldsymbol{X}^{\left(i\right)}=\boldsymbol{I}\otimes{}\boldsymbol{I}\otimes{}\ldots{}\otimes{}\boldsymbol{X}\otimes{}\boldsymbol{I}\otimes{}\ldots{}\otimes{}\boldsymbol{I}$ has $\boldsymbol{X}$ in the $i^{th}$ position and $\boldsymbol{I}$ in all $l-1$ other positions, $\otimes$ is the Kronecker product, $A\subseteq{}\left[l\right]=\left\{1,2,\ldots{},l\right\}$ and
\begin{align*}
\boldsymbol{L}_\alpha=\left[\begin{array}{cc}
-1 & 0 \\
1 & 0
\end{array}\right],\quad{}
\boldsymbol{L}_\beta=\left[\begin{array}{cc}
0 & 1 \\
0 & -1
\end{array}\right]
\quad\text{ and }
\boldsymbol{I}=\left[\begin{array}{cc}
1 & 0 \\
0 & 1
\end{array}\right].
\end{align*}

Then we define
\begin{align*}
\boldsymbol{\mathcal{L}}_\alpha^{\left[l\right]}=&\sum_{B\subseteq{}\left[l\right]}\boldsymbol{L}_\alpha^{\left(B\right)}=\left(\boldsymbol{L}_\alpha+\boldsymbol{I}\right)^{\otimes{}l}, \\
\boldsymbol{\mathcal{L}}_\beta^{\left[l\right]}=&\sum_{B\subseteq{}\left[l\right]}\boldsymbol{L}_\beta^{\left(B\right)}=\left(\boldsymbol{L}_\beta+\boldsymbol{I}\right)^{\otimes{}l},
\end{align*}
where $\left(\boldsymbol{X}+\boldsymbol{I}\right)^{\otimes{}l}=\left(\boldsymbol{X}+\boldsymbol{I}\right)\otimes{}\left(\boldsymbol{X}+\boldsymbol{I}\right)\otimes{}\ldots{}\otimes{}\left(\boldsymbol{X}+\boldsymbol{I}\right)$ involves $\boldsymbol{X}+\boldsymbol{I}$ a total of $l$ times.

Note that this definition is very similar to the rate matrix of \cite{sumner2012algebra} for all $l$ taxa present in an epoch also being in a convergence group,
\begin{align*}
\boldsymbol{Q}^{\left[l\right]}=\alpha\boldsymbol{\mathfrak{L}}_\alpha^{\left[l\right]}+\beta\boldsymbol{\mathfrak{L}}_\beta^{\left[l\right]},
\end{align*}
where
\begin{align*}
\boldsymbol{\mathfrak{L}}_\alpha^{\left[l\right]}=\sum_{B\subseteq{}\left[l\right],B\neq{}\emptyset}\boldsymbol{L}_\alpha^{\left(B\right)},\quad{}
\boldsymbol{\mathfrak{L}}_\beta^{\left[l\right]}=\sum_{B\subseteq{}\left[l\right],B\neq{}\emptyset}\boldsymbol{L}_\beta^{\left(B\right)}.
\end{align*}

Then
\begin{align*}
\boldsymbol{\mathcal{L}}_\alpha^{\left[l\right]}=\boldsymbol{\mathfrak{L}}_\alpha^{\left[l\right]}+\boldsymbol{I}^{\otimes{}l},\quad{}\boldsymbol{\mathcal{L}}_\beta^{\left[l\right]}=\boldsymbol{\mathfrak{L}}_\beta^{\left[l\right]}+\boldsymbol{I}^{\otimes{}l},
\end{align*}
where $\boldsymbol{I}^{\otimes{}l}=\boldsymbol{I}\otimes{}\boldsymbol{I}\otimes{}\ldots{}\otimes{}\boldsymbol{I}$ involves $\boldsymbol{I}$ a total of $l$ times.

\bigskip

The proof is split into four parts. We determine the rate matrix for an arbitrary epoch in each part of the proof. 1) Instead of having $N$ taxa, we assume that $\mathcal{N}$ has $l$ taxa, where $l\in\left\{1,2,\ldots{},N\right\}$, all in the same convergence-divergence group in some arbitrary epoch. 2) We assume that $\mathcal{N}$ has $N$ taxa, with the first $l$ --- according to the indices $i=i_1i_2\ldots{}i_N$ and $j=j_1j_2\ldots{}j_N$ --- in the same convergence-divergence group in the epoch. 3) We determine the rate matrix corresponding to an arbitrary convergence-divergence group with $l$ taxa --- still assuming $\mathcal{N}$ has $N$ taxa --- by permuting the taxon labels, which corresponds to permuting the indices $i$ and $j$. 4) The rate matrix for the epoch is determined by summing the rate matrices corresponding to all convergence-divergence groups in arbitrary epoch --- all $N$ taxa are in exactly one convergence-divergence group.

\setcounter{prop}{0}

\begin{prop}
Suppose a tip epoch of CDM $\mathcal{N}$ with leaf taxon set $X$ and $\left|X\right|=N$ corresponds to a set of sets of taxa in each convergence-divergence group $\mathcal{C}=\left\{C_1,C_2,\ldots{}C_k\right\}$. Suppose $\boldsymbol{Q}^{\left[\mathcal{C}\right]}$ is the $2^N\times2^N$ rate matrix representing $\mathcal{C}$. Then
\begin{align*}
\left[\boldsymbol{Q}^{\left[\mathcal{C}\right]}\right]_{ij}=
\begin{cases}
\alpha_r \quad &\text{if for some } C_r\in{}\mathcal{C}, \quad \prod_{a\in{}C_r}j_a=0, \\
&\prod_{a\in{}C_r}i_a=1 \text{ and } i_a=j_a \text{ for all } a\in{}X\setminus{}C_r, \\
\beta_r \quad &\text{if for some } C_r\in{}\mathcal{C}, \quad \prod_{a\in{}C_r}\left(1-j_a\right)=0, \\
&\prod_{a\in{}C_r}\left(1-i_a\right)=1 \text{ and } i_a=j_a \text{ for all } a\in{}X\setminus{}C_r, \\
0 \quad &\text{otherwise if } i\neq{}j, \\
-\sum_{s=1,s\neq{}j}^{2^N}\left[\boldsymbol{Q}^{\left[\mathcal{C}\right]}\right]_{sj} \quad &\text{if } i=j,
\end{cases}
\end{align*}
where $\alpha_r,\beta_r>0$.
\end{prop}

\begin{proof}

1) Suppose $\mathcal{N}$ has only $l$ taxa, where $l\in\left\{1,2,\ldots{},N\right\}$, all in the same convergence-divergence group in some epoch. Then from \cite{sumner2012algebra}, the rate matrix for the epoch is
\begin{align*}
\boldsymbol{Q}^{\left[l\right]}=&\alpha\boldsymbol{\mathfrak{L}}_\alpha^{\left[l\right]}+\beta\boldsymbol{\mathfrak{L}}_\beta^{\left[l\right]},
\end{align*}
where $\alpha,\beta>0$.

We first prove that
\begin{align*}
\left[\boldsymbol{Q}^{\left[l\right]}\right]_{ij}=
\begin{cases}
\alpha \quad &\text{if } \prod_{a=1}^{l}i_a=1 \text{ and }\prod_{a=1}^{l}j_a=0, \\
\beta \quad &\text{if } \prod_{a=1}^{l}\left(1-i_a\right)=1 \text{ and }\prod_{a=1}^{l}\left(1-j_a\right)=0, \\
0 \quad &\text{otherwise if } i\neq{}j.
\end{cases}
\end{align*}

Note that we have not declared the diagonal elements of $\boldsymbol{Q}^{\left[l\right]}$, which are determined in the next part of the proof.

We first define $\boldsymbol{\widetilde{Q}}^{\left[l\right]}=\alpha\boldsymbol{\mathcal{L}}_{\alpha}^{\left[l\right]}+\beta\boldsymbol{\mathcal{L}}_{\beta}^{\left[l\right]}$ and determine its elements by mathematical induction. It is straightforward to show that
\begin{align*}
\boldsymbol{L}_\alpha+\boldsymbol{I}=\left[\begin{array}{cc}
0 & 0 \\
1 & 1
\end{array}\right],\quad
\boldsymbol{L}_\beta+\boldsymbol{I}=\left[\begin{array}{cc}
1 & 1 \\
0 & 0
\end{array}\right]
\end{align*}
and in turn,
\begin{align*}
\boldsymbol{\widetilde{Q}}^{\left[1\right]}=\alpha\left(\boldsymbol{L}_\alpha+\boldsymbol{I}\right)+\beta\left(\boldsymbol{L}_\beta+\boldsymbol{I}\right)=&\left[\begin{array}{cc}
\beta & \beta \\
\alpha & \alpha
\end{array}\right].
\end{align*}

We claim that for some integer $m\geq{}1$,
\begin{align*}
\left[\boldsymbol{\widetilde{Q}}^{\left[m\right]}\right]_{ij}=
\begin{cases}
\beta \quad &\text{if } i=1, \\
\alpha \quad &\text{if } i=2^{m}, \\
0 \quad &\text{otherwise}.
\end{cases}
\end{align*}

We establish that if the claim is true, then
\begin{align*}
\left[\boldsymbol{\widetilde{Q}}^{\left[m+1\right]}\right]_{ij}=
\begin{cases}
\beta \quad &\text{if } i=1, \\
\alpha \quad &\text{if } i=2^{m+1}, \\
0 \quad &\text{otherwise}.
\end{cases}
\end{align*}

Since
\begin{align*}
\boldsymbol{\widetilde{Q}}^{\left[m\right]}=&\alpha\boldsymbol{\mathcal{L}}_{\alpha}^{\left[m\right]}+\beta\boldsymbol{\mathcal{L}}_{\beta}^{\left[m\right]} \\
=&\alpha\left(\boldsymbol{L}_\alpha+\boldsymbol{I}\right)^{\otimes{}m}+\beta\left(\boldsymbol{L}_\beta+\boldsymbol{I}\right)^{\otimes{}m},
\end{align*}
it follows that
\begin{align*}
\boldsymbol{\widetilde{Q}}^{\left[m+1\right]}=\alpha\left(\boldsymbol{L}_\alpha+\boldsymbol{I}\right)\otimes\boldsymbol{\mathcal{L}}_\alpha^{\left[m\right]}+\beta\left(\boldsymbol{L}_\beta+\boldsymbol{I}\right)\otimes\boldsymbol{\mathcal{L}}_\beta^{\left[m\right]}
\end{align*}
and thus the claim is established.

Then since
\begin{align*}
\boldsymbol{Q}^{\left[l\right]}=&\alpha\boldsymbol{\mathfrak{L}}_\alpha^{\left[l\right]}+\beta\boldsymbol{\mathfrak{L}}_\beta^{\left[l\right]} \\
=&\boldsymbol{\widetilde{Q}}^{\left[l\right]}-\left(\alpha+\beta\right)\boldsymbol{I}^{\otimes{}l},
\end{align*}
it follows that
\begin{align*}
\left[\boldsymbol{Q}^{\left[l\right]}\right]_{ij}=
\begin{cases}
\alpha \quad &\text{if } \prod_{a=1}^{l}i_a=1 \text{ and }\prod_{a=1}^{l}j_a=0, \\
\beta \quad &\text{if } \prod_{a=1}^{l}\left(1-i_a\right)=1 \text{ and }\prod_{a=1}^{l}\left(1-j_a\right)=0, \\
0 \quad &\text{otherwise if } i\neq{}j.
\end{cases}
\end{align*}

Note that $\sum_{i=1}^{2^l}\left[\boldsymbol{\widetilde{Q}}^{\left[l\right]}\right]_{ij}=\alpha+\beta$. Thus, $\sum_{i=1}^{2^l}\left[\boldsymbol{Q}^{\left[l\right]}\right]_{ij}=0$. Summarizing,
\begin{align*}
\left[\boldsymbol{Q}^{\left[l\right]}\right]_{ij}=
\begin{cases}
\alpha \quad &\text{if } \prod_{a=1}^{l}i_a=1 \text{ and }\prod_{a=1}^{l}j_a=0, \\
\beta \quad &\text{if } \prod_{a=1}^{l}\left(1-i_a\right)=1 \text{ and }\prod_{a=1}^{l}\left(1-j_a\right)=0, \\
0 \quad &\text{otherwise if } i\neq{}j, \\
-\sum_{s=1,s\neq{}j}^{2^l}\left[\boldsymbol{Q}^{\left[l\right]}\right]_{sj} \quad &\text{if } i=j.
\end{cases}
\end{align*}

From here onwards, we let $\ast$ represent the negative of the sum of all non-diagonal elements of column $j$ of the rate matrix.

\bigskip

2) Suppose that $\mathcal{N}$ has $N$ taxa, with the first $l$ --- according to the indices $i=i_1i_2\ldots{}i_N$ and $j=j_1j_2\ldots{}j_N$ --- in a convergence-divergence group. Assume that this is the only convergence-divergence group in the epoch. That is, in the epoch the last $N-l$ taxa are not in any convergence-divergence group; they will be assigned to convergence-divergence groups at later steps in the proof. Then by \cite{sumner2012algebra}, the rate matrix for the last $N-l$ taxa is $I^{\otimes{}N-l}$; that is, the last $N-l$ taxa are not evolving. Note that since the first $l$ taxa are diverging independently from the last $N-l$ taxa, the rate matrix for all $N$ taxa is the Kronecker product of the rate matrix for the first $l$ taxa and the rate matrix for the last $N-l$ taxa. Then it follows directly from 1) that the rate matrix has elements
\begin{align*}
\left[\boldsymbol{Q}^{\left[l\right]}\otimes{}\boldsymbol{I}^{\otimes{}N-l}\right]_{ij}=
\begin{cases}
\alpha \quad &\text{if } \prod_{a=1}^{l}i_a=1, \text{ }\prod_{a=1}^{l}j_a=0 \\
&\text{and } i_a=j_a \text{ for all } a\in\left\{l+1,l+2,\ldots{},N\right\}, \\
\beta \quad &\text{if } \prod_{a=1}^{l}\left(1-i_a\right)=1, \text{ }\prod_{a=1}^{l}\left(1-j_a\right)=0 \\
&\text{and } i_a=j_a \text{ for all } a\in\left\{l+1,l+2,\ldots{},N\right\}, \\
0 \quad &\text{otherwise if } i\neq{}j, \\
\ast \quad &\text{if } i=j.
\end{cases}
\end{align*}

\bigskip

3) Consider a permutation of leaf taxon $X$. This is an element of the symmetric group $\mathfrak{S}_N$ on $X$. An arbitrary such permutation corresponds to an arbitrary re-order of the leaf taxa. We allow the permutation to act on $V^{\otimes{}N}$ (the tensor product space of \cite{sumner2012algebra}). For some arbitrary convergence-divergence group $C_r\in\mathcal{C}$ involving $l$ taxa, there exists a permutation $\sigma\in\mathfrak{S}_N$ such that $\sigma\left(\boldsymbol{Q}^{\left[l\right]}\otimes{}\boldsymbol{I}^{\otimes{}N-l}\right)=\boldsymbol{Q}^{\left[C_r\right]}$, where $\boldsymbol{Q}^{\left[C_r\right]}$ is the rate matrix for the convergence-divergence group $C_r$ and all other leaf taxa belonging to no convergence-divergence group. Then by \cite{sumner2012algebra},
\begin{align*}
\boldsymbol{Q}^{\left[C_r\right]}=\alpha\boldsymbol{\mathfrak{L}}_\alpha^{\left[C_r\right]}+\beta\boldsymbol{\mathfrak{L}}_\beta^{\left[C_r\right]},
\end{align*}
which is obtained from $\boldsymbol{Q}^{\left[l\right]}\otimes{}\boldsymbol{I}^{\otimes{}N-l}$ by the permutation $\sigma$ on the slots of the Kronecker products of each term of $\boldsymbol{\mathfrak{L}}_\alpha^{\left[l\right]}\otimes{}\boldsymbol{I}^{\otimes{}N-l}$ and $\boldsymbol{\mathfrak{L}}_\beta^{\left[l\right]}\otimes{}\boldsymbol{I}^{\otimes{}N-l}$.

Then it follows directly from 2) that
\begin{align*}
\left[\boldsymbol{Q}^{\left[C_r\right]}\right]_{ij}=
\begin{cases}
\alpha \quad &\text{if } \prod_{a\in{}C_r}i_a=1, \text{ }\prod_{a\in{}C_r}j_a=0 \\
&\text{and } i_a=j_a \text{ for all } a\in{}\left[N\right]\setminus{}C_r, \\
\beta \quad &\text{if } \prod_{a\in{}C_r}\left(1-i_a\right)=1, \text{ }\prod_{a\in{}C_r}\left(1-j_a\right)=0 \\
&\text{and } i_a=j_a \text{ for all } a\in{}\left[N\right]\setminus{}C_r, \\
0 \quad &\text{otherwise if } i\neq{}j, \\
\ast \quad &\text{if } i=j.
\end{cases}
\end{align*}

\bigskip

4) Now suppose the substitution rates for convergence-divergence group $C_r$ are $\alpha_r,\beta_r>0$. Then since $\boldsymbol{Q}^{\left[\mathcal{C}\right]}=\sum_{C_r\in\mathcal{C}}\boldsymbol{Q}^{\left[C_r\right]}$, it follows directly from 3) that
\begin{align*}
\left[\boldsymbol{Q}^{\left[\mathcal{C}\right]}\right]_{ij}=
\begin{cases}
\alpha_r \quad &\text{if for some } C_r\in\mathcal{C} \text{, } \prod_{a\in{}C_r}i_a=1, \text{ }\prod_{a\in{}C_r}j_a=0 \\
&\text{and } i_a=j_a \text{ for all } a\in{}\left[N\right]\setminus{}C_r, \\
\beta_r \quad &\text{if for some } C_r\in\mathcal{C} \text{, } \prod_{a\in{}C_r}\left(1-i_a\right)=1, \text{ }\prod_{a\in{}C_r}\left(1-j_a\right)=0 \\
&\text{and } i_a=j_a \text{ for all } a\in{}\left[N\right]\setminus{}C_r, \\
0 \quad &\text{otherwise if } i\neq{}j, \\
\ast \quad &\text{if } i=j.
\end{cases}
\end{align*}

\end{proof}

\subsection{Proof of Theorem~\ref{convedges}}
\label{thm1}

For the proof, we assume the phylogenetic tensor at the beginning of the tip epoch takes an arbitrary form. We then consider the effect of convergence-divergence groups in the tip epoch on the phylogenetic tensor. As in the proof of Proposition~\ref{p1}, the proof is split into several parts. 1) We assume $\mathcal{N}$ has $N$ taxa, but $\mathcal{C}$ has only one convergence-divergence group $C_a$ involving the first $l$ taxa in the tip epoch. We determine the elements of the transition matrix $\boldsymbol{M}^{\left[C_a\right]}=\exp\left(\boldsymbol{Q}^{\left[l\right]}t\right)\otimes{}\boldsymbol{I}^{N-l}$. 2) We determine the phylogenetic tensor in the limit as the epoch time of the tip epoch diverges. 3) We assume $\mathcal{C}$ has $k\geq{}1$ convergence-divergence groups, $\mathcal{C}=\left\{C_1,C_2,\ldots{},C_k\right\}$, and determine the phylogenetic tensor in the limit as the epoch time of the tip epoch diverges.

\setcounter{thm}{1}

\begin{thm}
Suppose an arbitrary epoch of CDM $\mathcal{N}$ corresponds to set of sets of taxa in each convergence-divergence group $\mathcal{C}=\left\{C_1,C_2,\ldots{}C_k\right\}$. Then if $a,b\in{}C_i$, as tip epoch length $t\to\infty$, $a$ and $b$ become identical.
\end{thm}

\begin{proof}
1) As in the proof of Proposition~\ref{p1}, assume that the convergence-divergence group $C_a$ involves the first $l$ taxa. We assume $\mathcal{N}$ has $N$ taxa, unlike in 1) of the proof of Proposition~\ref{p1}. Suppose $\boldsymbol{P}'$ is the phylogenetic tensor representing the probabilities of combinations of states immediately before the tip epoch. Then let
\begin{align*}
\boldsymbol{\widetilde{P}}=\exp\left(\boldsymbol{Q}^{\left[C_a\right]}t\right)\cdot{}\boldsymbol{P}',
\end{align*}
where
\begin{align*}
\boldsymbol{Q}^{\left[C_a\right]}=\boldsymbol{Q}^{\left[l\right]}\otimes{}\boldsymbol{I}^{N-l}.
\end{align*}

To find an expression for $\exp\left(\boldsymbol{Q}^{\left[C_a\right]}t\right)$, we use the Taylor series,
\begin{align*}
\exp\left(\boldsymbol{Q}^{\left[C_a\right]}t\right)=&\exp\left(\left(\boldsymbol{Q}^{\left[l\right]}\otimes{}\boldsymbol{I}^{N-l}\right)t\right) \\
=&\boldsymbol{I}^{\otimes{}N}+\left(\boldsymbol{Q}^{\left[l\right]}\otimes{}\boldsymbol{I}^{N-l}\right)t+\frac{1}{2!}\left(\left(\boldsymbol{Q}^{\left[l\right]}\otimes{}\boldsymbol{I}^{N-l}\right)t\right)^2+\ldots \\
=&\boldsymbol{I}^{\otimes{}N}+\left(\boldsymbol{Q}^{\left[l\right]}\otimes{}\boldsymbol{I}^{N-l}\right)t+\frac{t^2}{2!}\left(\boldsymbol{Q}^{\left[l\right]}\otimes{}\boldsymbol{I}^{N-l}\right)\cdot{}\left(\boldsymbol{Q}^{\left[l\right]}\otimes{}\boldsymbol{I}^{N-l}\right)+\ldots{} \\
=&\boldsymbol{I}^{\otimes{}N}+\left(\boldsymbol{Q}^{\left[l\right]}\otimes{}\boldsymbol{I}^{N-l}\right)t+\frac{t^2}{2!}\left(\boldsymbol{Q}^{\left[l\right]}\cdot{}\boldsymbol{Q}^{\left[l\right]}\right)\otimes{}\left(\boldsymbol{I}^{N-l}\cdot{}\boldsymbol{I}^{N-l}\right)+\ldots{} \\
=&\boldsymbol{I}^{\otimes{}N}+\left(\boldsymbol{Q}^{\left[l\right]}\otimes{}\boldsymbol{I}^{N-l}\right)t+\frac{t^2}{2!}\left(\boldsymbol{Q}^{\left[l\right]}\right)^2\otimes{}\boldsymbol{I}^{N-l}+\ldots{} \\
=&\left(\boldsymbol{I}^{\otimes{}l}+\boldsymbol{Q}^{\left[l\right]}t+\frac{t^2}{2!}\left(\boldsymbol{Q}^{\left[l\right]}\right)^2+\ldots{}\right)\otimes{}\boldsymbol{I}^{N-l} \\
=&\exp\left(\boldsymbol{Q}^{\left[l\right]}t\right)\otimes{}\boldsymbol{I}^{N-l}.
\end{align*}

Then
\begin{align*}
\boldsymbol{\widetilde{P}}=\left(\exp\left(\boldsymbol{Q}^{\left[l\right]}t\right)\otimes{}\boldsymbol{I}^{N-l}\right)\cdot{}\boldsymbol{P}'.
\end{align*}

Now focusing on $\exp\left(\boldsymbol{Q}^{\left[l\right]}t\right)$, again using a Taylor series,
\begin{align*}
\exp\left(\boldsymbol{Q}^{\left[l\right]}t\right)=\boldsymbol{I}^{\otimes{}l}+\boldsymbol{Q}^{\left[l\right]}t+\frac{1}{2!}\left(\boldsymbol{Q}^{\left[l\right]}\right)^2t^2+\ldots{}.
\end{align*}

Focusing on $\left(\boldsymbol{Q}^{\left[l\right]}\right)^2$,
\begin{align*}
\left(\boldsymbol{Q}^{\left[l\right]}\right)^2=&\left(\alpha\boldsymbol{\mathfrak{L}}_\alpha^{\left[l\right]}+\beta\boldsymbol{\mathfrak{L}}_\beta^{\left[l\right]}\right)\cdot{}\left(\alpha\boldsymbol{\mathfrak{L}}_\alpha^{\left[l\right]}+\beta\boldsymbol{\mathfrak{L}}_\beta^{\left[l\right]}\right) \\
=&\alpha^2\left(\boldsymbol{\mathfrak{L}}_\alpha^{\left[l\right]}\right)^2+\alpha\beta\left(\boldsymbol{\mathfrak{L}}_\alpha^{\left[l\right]}\boldsymbol{\mathfrak{L}}_\beta^{\left[l\right]}+\boldsymbol{\mathfrak{L}}_\beta^{\left[l\right]}\boldsymbol{\mathfrak{L}}_\alpha^{\left[l\right]}\right)+\beta^2\left(\boldsymbol{\mathfrak{L}}_\beta^{\left[l\right]}\right)^2,
\end{align*}
where $\alpha,\beta>0$.

From \cite{sumner2012algebra},
\begin{align*}
\left(\boldsymbol{\mathfrak{L}}_\alpha^{\left[l\right]}\right)^2=-\boldsymbol{\mathfrak{L}}_\alpha^{\left[l\right]},\quad{}\boldsymbol{\mathfrak{L}}_\alpha^{\left[l\right]}\boldsymbol{\mathfrak{L}}_\beta^{\left[l\right]}=-\boldsymbol{\mathfrak{L}}_\beta^{\left[l\right]},\quad{}\boldsymbol{\mathfrak{L}}_\beta^{\left[l\right]}\boldsymbol{\mathfrak{L}}_\alpha^{\left[l\right]}=-\boldsymbol{\mathfrak{L}}_\alpha^{\left[l\right]},\quad{}\left(\boldsymbol{\mathfrak{L}}_\beta^{\left[l\right]}\right)^2=-\boldsymbol{\mathfrak{L}}_\beta^{\left[l\right]}.
\end{align*}

Then
\begin{align*}
\left(\boldsymbol{Q}^{\left[l\right]}\right)^2=&-\alpha^2\boldsymbol{\mathfrak{L}}_\alpha^{\left[l\right]}-\alpha\beta\left(\boldsymbol{\mathfrak{L}}_\alpha^{\left[l\right]}+\boldsymbol{\mathfrak{L}}_\beta^{\left[l\right]}\right)-\beta^2\boldsymbol{\mathfrak{L}}_\beta^{\left[l\right]} \\
=&-\left(\alpha+\beta\right)\left(\alpha\boldsymbol{\mathfrak{L}}_\alpha^{\left[l\right]}+\beta\boldsymbol{\mathfrak{L}}_\beta^{\left[l\right]}\right) \\
=&-\left(\alpha+\beta\right)\boldsymbol{Q}^{\left[l\right]}.
\end{align*}

It follows that
\begin{align*}
\left(\boldsymbol{Q}^{\left[l\right]}\right)^u=&\left(-1\right)^{u-1}\left(\alpha+\beta\right)\boldsymbol{Q}^{\left[l\right]}
\end{align*}
for any positive integer $u\geq{}2$.

Returning to the Taylor series,
\begin{align*}
\exp\left(\boldsymbol{Q}^{\left[l\right]}t\right)=&\boldsymbol{I}^{\otimes{}l}+\boldsymbol{Q}^{\left[l\right]}t-\frac{1}{2!}\left(\alpha+\beta\right)\boldsymbol{Q}^{\left[l\right]}t^2+\ldots{} \\
=&\boldsymbol{I}^{\otimes{}l}+\left(t-\frac{\left(\alpha+\beta\right)t^2}{2}+\ldots{}\right)\boldsymbol{Q}^{\left[l\right]} \\
=&\boldsymbol{I}^{\otimes{}l}+\frac{1}{\alpha+\beta}\left(\left(\alpha+\beta\right)t-\frac{\left(\alpha+\beta\right)^2t^2}{2}+\ldots{}\right)\boldsymbol{Q}^{\left[l\right]} \\
=&\boldsymbol{I}^{\otimes{}l}-\frac{1}{\alpha+\beta}\left(-\left(\alpha+\beta\right)t+\frac{\left(\alpha+\beta\right)^2t^2}{2}-\ldots{}\right)\boldsymbol{Q}^{\left[l\right]} \\
=&\boldsymbol{I}^{\otimes{}l}-\frac{1}{\alpha+\beta}\left(1-\left(\alpha+\beta\right)t+\frac{\left(\alpha+\beta\right)^2t^2}{2}-\ldots{}\right)\boldsymbol{Q}^{\left[l\right]}+\frac{1}{\alpha+\beta}\boldsymbol{Q}^{\left[l\right]} \\
=&\boldsymbol{I}^{\otimes{}l}+\frac{1}{\alpha+\beta}\left(1-\exp\left(-\left(\alpha+\beta\right)t\right)\right)\boldsymbol{Q}^{\left[l\right]}.
\end{align*}

Now recall from the proof of Proposition~\ref{p1} that
\begin{align*}
\left[\boldsymbol{Q}^{\left[l\right]}\right]_{ij}=
\begin{cases}
\alpha \quad &\text{if } \prod_{a=1}^{l}i_a=1 \text{ and }\prod_{a=1}^{l}j_a=0, \\
\beta \quad &\text{if } \prod_{a=1}^{l}\left(1-i_a\right)=1 \text{ and }\prod_{a=1}^{l}\left(1-j_a\right)=0, \\
0 \quad &\text{otherwise if } i\neq{}j, \\
\ast \quad &\text{otherwise if } i=j.
\end{cases}
\end{align*}

Thus, for columns to sum to zero,
\begin{align*}
\left[\boldsymbol{Q}^{\left[l\right]}\right]_{ij}=
\begin{cases}
\alpha \quad &\text{if } \prod_{a=1}^{l}i_a=1 \text{ and }\prod_{a=1}^{l}j_a=0, \\
\beta \quad &\text{if } \prod_{a=1}^{l}\left(1-i_a\right)=1 \text{ and }\prod_{a=1}^{l}\left(1-j_a\right)=0, \\
0 \quad &\text{otherwise if } i\neq{}j, \\
-\alpha \quad{} &\text{if } \prod_{a=1}^{l}\left(1-i_a\right)=\prod_{a=1}^{l}\left(1-j_a\right)=1, \\
-\beta \quad{} &\text{if } \prod_{a=1}^{l}i_a=\prod_{a=1}^{l}j_a=1, \\
-\left(\alpha+\beta\right) \quad{} &\text{otherwise}.
\end{cases}
\end{align*}

Letting $\boldsymbol{M}^{\left[l\right]}=\exp\left(\boldsymbol{Q}^{\left[l\right]}t\right)$,
\begin{align*}
\left[\boldsymbol{M}^{\left[l\right]}\right]_{ij}=
\begin{cases}
\frac{\alpha}{\alpha+\beta}\left(1-\exp\left(-\left(\alpha+\beta\right)t\right)\right) \quad &\text{if } \prod_{a=1}^{l}i_a=1 \text{ and }\prod_{a=1}^{l}j_a=0, \\
\frac{\beta}{\alpha+\beta}\left(1-\exp\left(-\left(\alpha+\beta\right)t\right)\right) \quad &\text{if } \prod_{a=1}^{l}\left(1-i_a\right)=1 \text{ and }\prod_{a=1}^{l}\left(1-j_a\right)=0, \\
0 \quad &\text{otherwise if } i\neq{}j, \\
1-\frac{\alpha}{\alpha+\beta}\left(1-\exp\left(-\left(\alpha+\beta\right)t\right)\right) \quad{} &\text{if } \prod_{a=1}^{l}\left(1-i_a\right)=\prod_{a=1}^{l}\left(1-j_a\right)=1, \\
1-\frac{\beta}{\alpha+\beta}\left(1-\exp\left(-\left(\alpha+\beta\right)t\right)\right) \quad{} &\text{if } \prod_{a=1}^{l}i_a=\prod_{a=1}^{l}j_a=1, \\
\exp\left(-\left(\alpha+\beta\right)t\right) \quad{} &\text{otherwise}.
\end{cases}
\end{align*}

Next, take the limit as $t\to\infty$. Then
\begin{align*}
\lim_{t\to\infty}\left[\boldsymbol{M}^{\left[l\right]}\right]_{ij}=
\begin{cases}
\frac{\alpha}{\alpha+\beta} \quad{} &\text{if } \prod_{a=1}^{l}i_a=1, \\
\frac{\beta}{\alpha+\beta} \quad{} &\text{if } \prod_{a=1}^{l}\left(1-i_a\right)=1, \\
0 \quad{} &\text{otherwise}.
\end{cases}
\end{align*}

Now let $\boldsymbol{M}^{\left[C_a\right]}=\exp\left(\boldsymbol{Q}^{\left[l\right]}t\right)\otimes{}\boldsymbol{I}^{N-l}$. Then
\begin{align*}
\lim_{t\to\infty}\left[\boldsymbol{M}^{\left[C_a\right]}\right]_{ij}=
\begin{cases}
\frac{\alpha}{\alpha+\beta} \quad{} &\text{if } \prod_{a=1}^{l}i_a=1 \\
&\text{and } i_a=j_a \text{ for all } a\in\left\{l+1,l+2,\ldots{},N\right\}, \\
\frac{\beta}{\alpha+\beta} \quad{} &\text{if } \prod_{a=1}^{l}\left(1-i_a\right)=1 \\
&\text{and } i_a=j_a \text{ for all } a\in\left\{l+1,l+2,\ldots{},N\right\}, \\
0 \quad{} &\text{otherwise}.
\end{cases}
\end{align*}

In summary, in the limit, the only rows of $\boldsymbol{M}^{\left[C_a\right]}$ with non-zero elements have the first $l$ indices being either all $0$ or all $1$. Then in the limit, the only non-zero elements of $\boldsymbol{\widetilde{P}}$ also have the first $l$ indices being either all $0$ or all $1$.

\bigskip

2) We let the substitution rates for convergence group $C_k$ be $\alpha_k,\beta_k>0$ and recognize that since $\frac{\alpha_k}{\beta_k}=\frac{\alpha}{\beta}$, then $\frac{\alpha_k}{\alpha_k+\beta_k}=\frac{\alpha}{\alpha+\beta}$ and $\frac{\beta_k}{\alpha_k+\beta_k}=\frac{\beta}{\alpha+\beta}$. Then using the same arguments as in Proposition~\ref{p1}, $\sigma\left(\boldsymbol{M}^{\left[C_a\right]}\right)=\boldsymbol{M}^{\left[C_k\right]}$ and
\begin{align*}
\lim_{t\to\infty}\left[\boldsymbol{M}^{\left[C_k\right]}\right]_{ij}=
\begin{cases}
\frac{\alpha}{\alpha+\beta} \quad{} &\text{if } \prod_{a\in{}C_k}i_a=1 \\
&\text{and } i_a=j_a \text{ for all } a\in{}\left[N\right]\setminus{}C_k, \\
\frac{\beta}{\alpha+\beta} \quad{} &\text{if } \prod_{a\in{}C_k}\left(1-i_a\right)=1 \\
&\text{and } i_a=j_a \text{ for all } a\in{}\left[N\right]\setminus{}C_k, \\
0 \quad{} &\text{otherwise}
\end{cases}
\end{align*}
and
\begin{align*}
\lim_{t\to\infty}\sum_{j=1}^{2^N}\left[\boldsymbol{M}^{\left[C_k\right]}\right]_{ij}\left[\boldsymbol{P}'\right]_j=
\begin{cases}
c_i^{\left[C_k\right]}>0 \quad &\text{if } \prod_{a\in{}C_k}i_a=1 \text{ or } \prod_{a\in{}C_k}\left(1-i_a\right)=1, \\
0 \quad{} &\text{otherwise}.
\end{cases}
\end{align*}

Note that $c_i^{\left[C_k\right]}>0$ being strictly positive follows from Assumption~\ref{genparam} of Section~\ref{ass}.

\bigskip

3) $\mathcal{N}$ has a tip epoch with epoch time $t$ and set of convergence groups $\mathcal{C}=\left\{C_1,C_2,\ldots,C_k\right\}$. Then the phylogenetic tensor $\boldsymbol{P}$ representing the probabilities of combinations of states at the leaves of the principal tree can be expressed as
\begin{align*}
\boldsymbol{P}=&\exp\left(\boldsymbol{Q}^{\left[\mathcal{C}\right]}t\right)\cdot{}\boldsymbol{P}' \\
=&\prod_{r=1}^{k}\boldsymbol{M}^{\left[C_r\right]}\cdot{}\boldsymbol{P}'.
\end{align*}

We prove that all elements of $\boldsymbol{P}$ converge to $0$ except those where, for each $C_r\in\mathcal{C}$, all taxa in $C_r$ are in the same state.

In the limit as the epoch time of the tip epoch diverges, the phylogenetic tensor is
\begin{align*}
\lim_{t\to\infty}\boldsymbol{P}=\lim_{t\to\infty}\prod_{r=1}^{k}\boldsymbol{M}^{\left[C_r\right]}\cdot{}\boldsymbol{P}'=&\prod_{r=1}^{k}\lim_{t\to\infty}\boldsymbol{M}^{\left[C_r\right]}\cdot{}\boldsymbol{P}'.
\end{align*}

We prove that
\begin{align*}
\lim_{t\to\infty}\left[\boldsymbol{P}\right]_i
=&\begin{cases}
c_i^{\left[\mathcal{C}\right]}>0 \quad &\text{if for all } C_r\in\mathcal{C} \text{, } \prod_{a\in{}C_r}i_a=1 \text{ or } \prod_{a\in{}C_r}\left(1-i_a\right)=1, \\
0 \quad{} &\text{otherwise}.
\end{cases}
\end{align*}

We prove this claim by induction on the $k$ convergence groups. Note that in 2) we have already proven the claim for the first convergence-divergence group $C_k$ applied to $\boldsymbol{P}'$. Thus, all that remains it to prove that given the claim is true for $\boldsymbol{P}^{k-v+2}=\boldsymbol{M}^{\left[C_v\right]}\cdots{}\boldsymbol{M}^{\left[C_{v+1}\right]}\ldots{}\boldsymbol{M}^{\left[C_k\right]}\cdot{}\boldsymbol{P}'$ for some $v\in\left\{2,\ldots{},k\right\}$, it must be true for $\boldsymbol{P}^{k-v+3}=\boldsymbol{M}^{\left[C_{v-1}\right]}\cdots{}\boldsymbol{M}^{\left[C_v\right]}\ldots{}\boldsymbol{M}^{\left[C_k\right]}\cdot{}\boldsymbol{P}'$.

We assume that
\begin{align*}
\lim_{t\to\infty}\left[\boldsymbol{P}^{k-v+2}\right]_i=
\begin{cases}
c_i^{\left[\cup_{r=v}^{k}C_r\right]}>0 \quad &\text{if for all } r\in\left\{v,v+1,\ldots{},k\right\} \text{, } \prod_{a\in{}C_r}i_a=1 \\
&\text{or } \prod_{a\in{}C_r}\left(1-i_a\right)=1, \\
0 \quad &\text{otherwise}.
\end{cases}
\end{align*}

Then in the limit as $t\to\infty$, all elements of $\boldsymbol{P}^{k-v+2}$ are $0$ except those where for all $C_r\in\mathcal{C}$, $r\in\left\{v,v+1,\ldots{},k\right\}$, all taxa in $C_r$ are in the same state.

For $\left[\boldsymbol{P}^{k-v+3}\right]_i$ to be non-zero, there must exist some index $s$, such that
\begin{equation}
\left[\boldsymbol{M}^{\left[C_{v-1}\right]}\right]_{is}>0 \label{con1a}
\end{equation}
and
\begin{equation}
\left[\boldsymbol{P}^{k-v+2}\right]_s>0. \label{con2a}
\end{equation}

For Equation~\eqref{con1a} to be true,
\begin{equation}
\label{allcon1}
\begin{aligned}
\begin{cases}
&\prod_{a\in{}C_{v-1}}i_a=1 \text{ or } \prod_{a\in{}C_{v-1}}\left(1-i_a\right)=1, \\
&i_a=s_a \text{ for all } a\in\left[N\right]\setminus{}C_{v-1}.
\end{cases}
\end{aligned}
\end{equation}

For Equation~\eqref{con2a} to be true, by assumption,
\begin{equation}
\label{allcon2}
\begin{aligned}
\text{for all } r\in\left\{v,v+1,\ldots{},k\right\} \text{, } \prod_{a\in{}C_r}s_a=1 \text{ or } \prod_{a\in{}C_r}\left(1-s_a\right)=1.
\end{aligned}
\end{equation}

Combining Constraints~\eqref{allcon1} and Constraints~\eqref{allcon2},
\begin{equation}
\label{allcons}
\begin{aligned}
\begin{cases}
&\prod_{a\in{}C_{v-1}}i_a=1 \text{ or } \prod_{a\in{}C_{v-1}}\left(1-i_a\right)=1, \\
&\text{for all } r\in\left\{v,v+1,\ldots{},k\right\} \text{, } \prod_{a\in{}C_r}i_a=1 \text{ or } \prod_{a\in{}C_r}\left(1-i_a\right)=1.
\end{cases}
\end{aligned}
\end{equation}

Constraint~\eqref{allcons} can be simplified to
\begin{align*}
&\text{for all } r\in\left\{v-1,v,\ldots{},k\right\} \text{, } \prod_{a\in{}C_r}i_a=1 \text{ or } \prod_{a\in{}C_r}\left(1-i_a\right)=1.
\end{align*}

In summary,
\begin{align*}
\lim_{t\to\infty}\left[\boldsymbol{P}\right]_i=
\begin{cases}
c_i^{\left[\cup_{r=1}^{k}C_r\right]}>0 \quad &\text{if for all } r\in\left\{1,2,\ldots{},k\right\} \text{,} \\
&\prod_{a\in{}C_r}i_a=1 \text{ or } \prod_{a\in{}C_r}\left(1-i_a\right)=1, \\
0 \quad &\text{otherwise}.
\end{cases}
\end{align*}

\end{proof}

\section{Identifiability of \texorpdfstring{$4$}{4}-taxon CDMs}
\label{propsident}

\cite{sumner2012algebra} formally describe phylogenetic epoch models in their Definition~6.1 and introduce notation to compute the phylogenetic tensors. We use the same notation for our CDMs.

For each $4$-taxon CDM, the phylogenetic tensor $\boldsymbol{P}$ is transformed into the Hadamard basis $\boldsymbol{\widehat{P}}$ by multiplying by $\boldsymbol{H}_{16}=\boldsymbol{H}_2^{\otimes{}4}$, where
\begin{align*}
\boldsymbol{H}_2=\left[\begin{array}{cc}
1 & 1 \\
1 & -1
\end{array}\right].
\end{align*}

In the Hadamard basis, the phylogenetic tensor for CDM $5$ is
\begin{equation}
\label{hadphylotensor}
\begin{aligned}
\boldsymbol{\widehat{P}}=&\boldsymbol{H}_{16}\cdot{}\boldsymbol{P}=\left[\begin{array}{c}
q_{0000} \\
q_{0001} \\
q_{0010} \\
q_{0011} \\
q_{0100} \\
q_{0101} \\
q_{0110} \\
q_{0111} \\
q_{1000} \\
q_{1001} \\
q_{1010} \\
q_{1011} \\
q_{1100} \\
q_{1101} \\
q_{1110} \\
q_{1111} \\
\phantom{q_{1111}}
\end{array}\right] \\
=&\left[\begin{array}{c}
1 \\
\gamma \\
\gamma \\
\gamma^2+\left(1-\gamma^2\right)r_{0011} \\
\gamma \\
\gamma^2+\left(1-\gamma^2\right)r_{0101} \\
\gamma^2+\left(1-\gamma^2\right)r_{0110} \\
\gamma\left(\gamma^2+\left(1-\gamma^2\right)\left(r_{0011}+r_{0101}+r_{0110}-2r_{0111}\right)\right) \\
\gamma \\
\gamma^2+\left(1-\gamma^2\right)r_{1001} \\
\gamma^2+\left(1-\gamma^2\right)r_{1010} \\
\gamma\left(\gamma^2+\left(1-\gamma^2\right)\left(r_{0011}+r_{1001}+r_{1010}-2r_{1011}\right)\right) \\
\gamma^2+\left(1-\gamma^2\right)r_{1100} \\
\gamma\left(\gamma^2+\left(1-\gamma^2\right)\left(r_{0101}+r_{1001}+r_{1100}-2r_{1101}\right)\right) \\
\gamma\left(\gamma^2+\left(1-\gamma^2\right)\left(r_{0110}+r_{1010}+r_{1100}-2r_{1110}\right)\right) \\
\gamma^2\left(\gamma^2+\left(1-\gamma^2\right)\left(r_{0011}+r_{0101}+r_{0110}+r_{1001}+r_{1010}+r_{1100}\right.\right. \\
\left.\left.-2\left(r_{0111}+r_{1011}+r_{1101}+r_{1110}-2\delta\right)\right)\right)+\left(1-\gamma^2\right)^2r_{1111}
\end{array}\right].
\end{aligned}
\end{equation}

See Mathematica file S2.nb (text version S3.txt) on \url{https://github.com/jonathanmitchell88/CDMsSI} for a derivation of Equation~\eqref{hadphylotensor} and equations for $r_{0011}$, $r_{0101}$, $\ldots$, $r_{1111}$ and $\delta$ in terms of $x_i$ and $y_i$ for CDM $5$. CDMs $1-4$ are all nested in CDM $5$. Thus, their phylogenetic tensors are also in the form of Equation~\eqref{hadphylotensor}. However, the equations for $r_{0011}$, $r_{0101}$, $\ldots$, $r_{1111}$ and $\delta$ involve different expressions of $x_i$ and $y_i$.

For the proof that follows, the order of parameters is as in Figure~\ref{4taxonCDMs2}, with $x_i=\exp\left(-\left(a_i+b_i\right)\right)\in\left(0,1\right)$ for all $i\in\left\{1,2,\ldots{},11\right\}$. Note again that the exact location of the root on the outgroup edge is not identifiable; $t_1$ corresponds to the sum of epoch times of epochs from the root to the outgroup added to the first epoch time below the root.

To establish whether a CDM is identifiable or not, we must first determine a maximal set of independent elements of the transformed phylogenetic tensor. That is, a set with maximum cardinality such that there are no algebraic equations --- equalities --- involving multiple elements of the set. If the cardinality of the set equals the number of parameters, then the CDM is identifiable. For example, we can see that invariants $q_{0001}=q_{0010}=q_{0100}=q_{1000}=\gamma$ are equalities on all CDMs. Thus, we can only include one of $q_{0001}$, $q_{0010}$, $q_{0100}$ and $q_{1000}$ in the set.

To determine all such equalities, for a given CDM with $l+1$ parameters $x_1,x_2,\ldots{},x_l,\gamma$, we construct the ideal,
\begin{align*}
I=&\langle{}r_{0011}-f_{0011}\left(x_1,x_2,\ldots{},x_l\right),r_{0101}-f_{0101}\left(x_1,x_2,\ldots{},x_l\right),\ldots{}, \\
&r_{1111}-f_{1111}\left(x_1,x_2,\ldots{},x_l\right),\delta-f_{\delta}\left(x_1,x_2,\ldots{},x_l\right)\rangle \\
\subseteq{}&\mathbb{Q}\left[x_1,x_2,\ldots{},x_l,r_{0011},r_{0101},\ldots{},r_{1111},\delta\right],
\end{align*}
where each $r_{ijkl}-f_{ijkl}\left(x_1,x_2,\ldots{},x_l\right)$ and $\delta-f_{\delta}\left(x_1,x_2,\ldots{},x_l\right)$ is identically zero. (We can ignore $\gamma$ since $q_{ijkl}=\gamma^2+\left(1-\gamma^2\right)r_{ijkl}$ and including any of these invariants does not help us to determine any invariants involving multiple variables $r_{0011}$, $r_{0101}$, $\ldots{}$, $r_{1111}$, $\delta$.)

In the Macaulay2 file S4.m2 (output file S5.txt) on \url{https://github.com/jonathanmitchell88/CDMsSI} we derive the (reduced) Gr\"{o}bner basis for this ideal for a particular monomial order for CDM~$5$. Below we outline how this Gr\"{o}bner basis is computed.

In the Mathematica file S2.nb (text version S3.txt) we derive the following equations to input into the generators of the ideal:
\begin{equation}
\label{gens}
\begin{aligned}
\begin{cases}
f_{0011}=x_4x_5x_6x_7x_9x_{10}, \\
f_{0101}=x_{10}x_{11}\left(1-x_9\left(1-x_2x_3x_4x_6x_8\right)\right), \\
f_{0110}=x_7x_8x_9x_{11}\left(1-x_6\left(1-x_2x_3x_5\right)\right), \\
f_{0111}=x_2x_3x_4x_5x_6x_7x_8x_9x_{10}x_{11}, \\
f_{1001}=x_1x_2x_4x_9x_{10}, \\
f_{1010}=x_1x_2x_5x_6x_7, \\
f_{1011}=x_1x_2x_4x_5x_6x_7x_9x_{10}, \\
f_{1100}=x_1x_3x_6x_8x_9x_{11}, \\
f_{1101}=x_1x_2x_3x_4x_6x_8x_9x_{10}x_{11}, \\
f_{1110}=x_1x_2x_3x_5x_6x_7x_8x_9x_{11}, \\
f_{1111}=x_1x_7x_{10}x_{11}\left(x_4x_8x_9\left(x_2\left(1-x_6\right)+x_3x_5x_6\right)+x_2x_5x_6\left(1-x_9\right)\right), \\
\phantom{f_{1110}}\nphantom{$f_{\delta}$}f_{\delta}=x_1x_2x_3x_4x_5x_6x_7x_8x_9x_{10}x_{11}.
\end{cases}
\end{aligned}
\end{equation}

The functions $f_{0011}=f_{0011}\left(x_1,x_2,\ldots{},x_l\right)$, $f_{0101}=f_{0101}\left(x_1,x_2,\ldots{},x_l\right)$, $\ldots{}$, $f_{1111}=f_{1111}\left(x_1,x_2,\ldots{},x_l\right)$ and $f_\delta=f_\delta\left(x_1,x_2,\ldots{},x_l\right)$ depend on the CDM in question, for example, CDM $5$.

The monomial order is the elimination order, eliminating the block $x_1,x_2,\ldots,x_l$, with graded reverse lexicographic order on each block, $x_1>x_2>\ldots{}>x_l$ and $r_{0011}>r_{0101}>\ldots{}>r_{1111}>\delta$.

Next, we compute the (reduced) Gr\"{o}bner basis $I_G$ of $I$. Then $I_{G,q}=I_G\cap\mathbb{R}\left[r_{0011},r_{0101},\ldots{},r_{1111},\delta\right]$ is a Gr\"{o}bner basis for the elimination ideal involving only $r_{0011},r_{0101},\ldots{},r_{1111},\delta$.

Note that $q_{1111}$ is a function of both $r_{1111}$ and $\delta$, the only element of $\boldsymbol{\widehat{P}}$ that is a function of either. Thus, the maximum cardinality set can include at most one of $r_{1111}$ and $\delta$. In S4.m2 we find that when eliminating $r_{1111}$ there are no generators that involve $\delta$. Thus, $r_{1111}$ is eliminated and $\delta$ is another independent variable of the system when $r_{1111}$ is eliminated.

Note that there are still some algebraic equations --- equalities --- involving some elements of $\left\{r_{0011},r_{0101},\ldots{},r_{1110},\delta\right\}$. In S4.m2 (output file S5.txt) we find the largest cardinality subset with no algebraic equations involving multiple elements. This cardinality, plus one for $\gamma$, is the degrees of freedom of the phylogenetic tensor. Given a set of parameters of the CDM, if this degrees of freedom is less than the number of parameters, then the system of polynomial equations is underdetermined and that set of parameters is not identifiable. (Note that some individual parameters may still be indentifiable.) Otherwise, the set of parameters is identifiable. If that set of parameters is not identifiable, it may be possible to combine the parameters in a such a way that the new set of parameters is identifiable.

\subsection{Proof of Proposition~\ref{CDMs}}
\label{CDM5appendix}

See S4.m2 (output file S5.txt) and S6.m2 (output file S7.txt) on \url{https://github.com/jonathanmitchell88/CDMsSI} for the computations of the (reduced) Gr\"{o}bner bases of the ideals in this proof.

\setcounter{prop}{2}

\begin{prop}
The parameter set for CDM $5$ is identifiable.
\end{prop}

\begin{proof}

In S4.m2 (output file S5.txt), we see that there are $9$ elements of $\left\{r_{0011},r_{0101},\ldots{},r_{1110},\delta\right\}$ that are free to vary. However, CDM $5$ has $11$ parameters excluding $\gamma$. Thus, this set of parameters is not identifiable. However, recall in Section~\ref{pars} that taking some products of $x_i$ parameters may be required to obtain a set of identifiable parameters. Since there are $9$ elements of $\left\{r_{0011},r_{0101},\ldots{},r_{1110},\delta\right\}$ that are free to vary, we desire a set of $9$ parameters.

In S2.nb (text version S3.txt), we express $f_{0011},f_{0101},\ldots{},f_{1111},\delta$ in terms of the set of parameters $\left\{y_1,y_2,y_3,y_4,y_5,y_6,y_7,y_8,y_9\right\}$. Precisely,
\begin{align*}
\begin{cases}
y_1=x_1, \\
y_2=x_2, \\
y_3=x_3x_8x_{11}, \\
y_4=x_4x_{10}, \\
y_5=x_5x_7, \\
y_6=x_6, \\
y_7=x_7x_8x_{11}, \\
y_8=x_9, \\
y_9=x_{10}x_{11}.
\end{cases}
\end{align*}

In S6.m2 (output file S7.txt), we see that this set of parameters is identifiable. We note that $x_i\in\left(0,1\right)$ for all $i\in\left\{1,2,\ldots{},11\right\}$. It follows that $r_{0011},r_{0101},\ldots{},r_{1111},\delta\in\left(0,1\right)$ and $y_i\in\left(0,1\right)$ for all $i\in\left\{1,2,\ldots{},9\right\}$. In S2.nb (text version S3.txt), we see that the solutions to the system are
\begin{equation}
\label{ysolns}
\begin{aligned}
\begin{cases}
y_1=\frac{\delta}{r_{0111}}, \\
y_2=\frac{r_{0111}\sqrt{r_{1001}r_{1010}}}{\delta\sqrt{r_{0011}}}, \\
y_3=\frac{\delta}{\sqrt{r_{0011}r_{1001}r_{1010}}}, \\
y_4=\frac{r_{1101}\delta\sqrt{r_{0011}}}{r_{0111}r_{1100}\sqrt{r_{1001}r_{1010}}}, \\
y_5=\frac{\delta}{r_{1101}}, \\
y_6=\frac{r_{1101\sqrt{r_{0011}r_{1010}}}}{\delta\sqrt{r_{1001}}}, \\
y_7=\frac{\delta\left(r_{0110}r_{1101}\delta\sqrt{r_{0011}}-r_{0111}^2r_{1100}\sqrt{r_{1001}r_{1010}}\right)}{r_{0111}r_{1100}\sqrt{r_{0011}r_{1001}}\left(\delta\sqrt{r_{1001}}-r_{1101}\sqrt{r_{0011}r_{1010}}\right)}, \\
y_8=\frac{r_{0111}r_{1001}r_{1100}}{r_{1101}\delta}, \\
y_9=\frac{r_{1101}\left(r_{0101}\delta-r_{0111}r_{1101}\right)}{r_{1101}\delta-r_{0111}r_{1001}r_{1100}}.
\end{cases}
\end{aligned}
\end{equation}

Thus, the parameter set $\left\{y_1,y_2,y_3,y_4,y_5,y_6,y_7,y_8,y_9,\gamma\right\}$ on CDM $5$ is identifiable.

\end{proof}

Since CDMs $1-4$ are all nested in CDM $5$, the transformed phylogenetic tensors of CDMs $1-4$ can be determined directly from that of CDM $5$ by setting some parameters $x_i$ to $1$. Similarly, it is straightforward to prove that the equivalent sets of $y_i$ parameters are identifiable for each of CDMs $1-4$. The numbers of degrees of freedom for the phylogenetic tensors of CDMs $1-5$ are $6$, $7$, $8$, $9$ and $10$, respectively.

\section{Proof of Theorem~\ref{disttheorem}}
\label{theoremdist}

For a robust proof, we could consider the (reduced) Gr\"{o}bner bases of the ideals representing the parameter spaces of the CDMs and show that each CDM has a unique Gr\"{o}bner basis. The Gr\"{o}bner basis for CDM $5$ has already been computed in Section~\ref{propsident}. However, computation of the Gr\"{o}bner bases is slow and some bases contain many generators. Instead, it is sufficient to consider only a few constraints for each parameter space that exist for some CDMs and not others, greatly simplifying the proof.

\setcounter{thm}{4}

\begin{thm}
All pairs of $4$-taxon leaf-labeled CDMs of Section~\ref{4taxaidentifiability} are distinguishable.
\end{thm}

\begin{proof}

By Proposition~\ref{distinguishabilitydiffdim}, if two CDMs have parameter spaces with different dimensions, then they are distinguishable from each other. CDMs $1$, $2$, $3$, $4$ and $5$ have parameter space dimensions $6$, $7$, $8$, $9$ and $10$ respectively, corresponding to the numbers of free parameters.

All that is left to prove is that any two CDMs that differ only in their leaf labelings are distinguishable. The notation that follows is consistent with that of Section~\ref{propsident}. Recall that $y_i\in\left(0,1\right)$ for all $i\in\left\{1,2,\ldots{},9\right\}$.

\subsection*{CDM $5$}

See S8.nb (text version S9.txt) on \url{https://github.com/jonathanmitchell88/CDMsSI} for proofs of the following claims.

For leaf labelings $\left(o,\left(a,\left(b,c\right)\right)\right)$ and $\left(o,\left(a,\left(c,b\right)\right)\right)$,
\begin{align*}
r_{0011}r_{1001}r_{1010}-r_{1011}^2=0,
\end{align*}
while for the other leaf labelings
\begin{align*}
r_{0011}r_{1001}r_{1010}-r_{1011}^2>0.
\end{align*}

Thus, we need only show that CDMs with leaf labelings $\left(o,\left(a,\left(b,c\right)\right)\right)$ and $\left(o,\left(a,\left(c,b\right)\right)\right)$ are distinguishable. To do this, we show that the intersection of the parameter spaces of the two CDMs is the empty set. Letting $y_i$ be the parameters corresponding to leaf labeling $\left(o,\left(a,\left(b,c\right)\right)\right)$ and $z_i$ corresponding to $\left(o,\left(a,\left(c,b\right)\right)\right)$, we equate the equations for each element of the two phylogenetic tensors and solve for the $z_i$ parameters,
\begin{align*}
\begin{cases}
\phantom{y_1\left(y_4y_8\left(y_2y_7\left(1-y_6\right)+y_3y_5y_6\right)\right.}\nphantom{$y_4y_5y_6y_8$}y_4y_5y_6y_8=&z_4z_5z_6z_8, \\
\phantom{y_1\left(y_4y_8\left(y_2y_7\left(1-y_6\right)+y_3y_5y_6\right)\right.}\nphantom{$y_9\left(1-y_8\right)+y_2y_3y_4y_6y_8$}y_9\left(1-y_8\right)+y_2y_3y_4y_6y_8=&z_8\left(z_7\left(1-z_6\right)+z_2z_3z_5z_6\right), \\
\phantom{y_1\left(y_4y_8\left(y_2y_7\left(1-y_6\right)+y_3y_5y_6\right)\right.}\nphantom{$y_8\left(y_7\left(1-y_6\right)+y_2y_3y_5y_6\right)$}y_8\left(y_7\left(1-y_6\right)+y_2y_3y_5y_6\right)=&z_9\left(1-z_8\right)+z_2z_3z_4z_6z_8, \\
\phantom{y_1\left(y_4y_8\left(y_2y_7\left(1-y_6\right)+y_3y_5y_6\right)\right.}\nphantom{$y_2y_3y_4y_5y_6y_8$}y_2y_3y_4y_5y_6y_8=&z_2z_3z_4z_5z_6z_8, \\
\phantom{y_1\left(y_4y_8\left(y_2y_7\left(1-y_6\right)+y_3y_5y_6\right)\right.}\nphantom{$y_1y_2y_4y_8$}y_1y_2y_4y_8=&z_1z_2z_5z_6, \\
\phantom{y_1\left(y_4y_8\left(y_2y_7\left(1-y_6\right)+y_3y_5y_6\right)\right.}\nphantom{$y_1y_2y_5y_6$}y_1y_2y_5y_6=&z_1z_2z_4z_8, \\
\phantom{y_1\left(y_4y_8\left(y_2y_7\left(1-y_6\right)+y_3y_5y_6\right)\right.}\nphantom{$y_1y_2y_4y_5y_6y_8$}y_1y_2y_4y_5y_6y_8=&z_1z_2z_4z_5z_6z_8, \\
\phantom{y_1\left(y_4y_8\left(y_2y_7\left(1-y_6\right)+y_3y_5y_6\right)\right.}\nphantom{$y_1y_3y_6y_8$}y_1y_3y_6y_8=&z_1z_3z_6z_8, \\
\phantom{y_1\left(y_4y_8\left(y_2y_7\left(1-y_6\right)+y_3y_5y_6\right)\right.}\nphantom{$y_1y_2y_3y_4y_6y_8$}y_1y_2y_3y_4y_6y_8=&z_1z_2z_3z_5z_6z_8, \\
\phantom{y_1\left(y_4y_8\left(y_2y_7\left(1-y_6\right)+y_3y_5y_6\right)\right.}\nphantom{$y_1y_2y_3y_5y_6y_8$}y_1y_2y_3y_5y_6y_8=&z_1z_2z_3z_4z_6z_8, \\
\phantom{y_1\left(y_4y_8\left(y_2y_7\left(1-y_6\right)+y_3y_5y_6\right)\right.}\nphantom{$y_1\left(y_4y_8\left(y_2y_7\left(1-y_6\right)+y_3y_5y_6\right)\right.$}y_1\left(y_4y_8\left(y_2y_7\left(1-y_6\right)+y_3y_5y_6\right)\right.\mathrel{\raisebox{-1em}{$=$}}&z_1\left(z_4z_8\left(z_2z_7\left(1-z_6\right)+z_3z_5z_6\right)\right. \\
\phantom{y_1\left(y_4y_8\left(y_2y_7\left(1-y_6\right)+y_3y_5y_6\right)\right.}\nphantom{$\left.+y_2y_5y_6y_9\left(1-y_8\right)\right)$}\left.+y_2y_5y_6y_9\left(1-y_8\right)\right)&\left.+z_2z_5z_6z_9\left(1-z_8\right)\right). \\
\phantom{y_1\left(y_4y_8\left(y_2y_7\left(1-y_6\right)+y_3y_5y_6\right)\right.}\nphantom{$y_1y_2y_3y_4y_5y_6y_8$}y_1y_2y_3y_4y_5y_6y_8=&z_1z_2z_3z_4z_5z_6z_8.
\end{cases}
\end{align*}
Solving this system of equations --- see S10.m2 (output file S11.txt) and the expressions simplified in S8.nb (text version S9.txt) on \url{https://github.com/jonathanmitchell88/CDMsSI} --- we obtain
\begin{align*}
z_1z_2^2z_3z_4z_6z_8\left(1-z_6\right)\left(1-z_8\right)\left(z_4z_7z_8-z_5z_9\right)=0,
\end{align*}
which has no solutions since $z_i\in\left(0,1\right)$ for all $i\in\left\{1,2,\ldots,8\right\}$ and the generating parameter must be a generic point in the parameter space, i.e. $z_4z_7z_8-z_5z_9\neq{}0$. Thus, for CDM $5$, any two CDMs with different leaf labelings are distinguishable.

\subsection*{CDM $4$}

The proof is identical to that of CDM $5$, but with the addition of $y_9=z_9=1$. Again, see S8.nb (text version S9.txt) and S10.m2 (output file S11.txt). We obtain
\begin{align*}
z_1z_2z_3z_4z_5z_6z_8\left(1-z_6\right)\left(1-z_7z_8\right)=0,
\end{align*}
which again has no solutions. Thus, for CDM $4$, any two CDMs with different leaf labelings are distinguishable.

\subsection*{CDM $3$}

See S8.nb (text version S9.txt) for proofs of the following claims.

For leaf labeling pairs $\left(o,\left(a,\left(b,c\right)\right)\right)$ and $\left(o,\left(c,\left(b,a\right)\right)\right)$, $\left(o,\left(a,\left(c,b\right)\right)\right)$ and
$\left(o,\left(b,\left(c,a\right)\right)\right)$ and $\left(o,\left(b,\left(a,c\right)\right)\right)$ and $\left(o,\left(c,\left(a,b\right)\right)\right)$,
\begin{align*}
\begin{cases}
\min\left(r_{0011}r_{1100},r_{0101}r_{1010},r_{0110}r_{1001}\right)=r_{0101}r_{1010}, \\
\min\left(r_{0011}r_{1100},r_{0101}r_{1010},r_{0110}r_{1001}\right)=r_{0110}r_{1001}, \\
\min\left(r_{0011}r_{1100},r_{0101}r_{1010},r_{0110}r_{1001}\right)=r_{0011}r_{1100},
\end{cases}
\end{align*}
respectively, where each equation corresponds to a leaf labeling pair. Thus, any CDM from one pair is distinguishable from a CDM from another pair.

All that is left is to prove that CDMs from an arbitrary pair, for example, $\left(o,\left(a,\left(b,c\right)\right)\right)$ and $\left(o,\left(c,\left(b,a\right)\right)\right)$, are distinguishable. For leaf labeling $\left(o,\left(a,\left(b,c\right)\right)\right)$, but not $\left(o,\left(c,\left(b,a\right)\right)\right)$,
\begin{align*}
r_{0011}r_{1001}r_{1010}-r_{1011}^2=0.
\end{align*}
For leaf labeling $\left(o,\left(c,\left(b,a\right)\right)\right)$, but not $\left(o,\left(a,\left(b,c\right)\right)\right)$,
\begin{align*}
r_{0011}r_{1001}r_{1010}-r_{1011}^2>0.
\end{align*}
Thus for CDM $3$, any two CDMs with different leaf labelings are distinguishable.

\subsection*{CDM $2$}

See S8.nb (text version S9.txt) for proofs of the following claims.

The constraints for CDM $2$ include those described above for CDM $3$. Thus for CDM $2$, any two CDMs with different leaf labelings are distinguishable.

\subsection*{CDM $1$}

See S8.nb (text version S9.txt) for proofs of the following claims.

For leaf labelings $\left(o,\left(a,\left(b,c\right)\right)\right)$, $\left(o,\left(b,\left(a,c\right)\right)\right)$ and $\left(o,\left(c,\left(a,b\right)\right)\right)$,
\begin{align*}
\begin{cases}
r_{0101}r_{1010}=r_{0110}r_{1001}<r_{0011}r_{1100}, \\
r_{0011}r_{1100}=r_{0110}r_{1001}<r_{0101}r_{1010}, \\
r_{0011}r_{1100}=r_{0101}r_{1010}<r_{0110}r_{1001},
\end{cases}
\end{align*}
respectively. Thus, for CDM $1$ any two CDMs with different leaf labelings are distinguishable.

\end{proof}

\section{Proof of Theorem~\ref{dist}}
\label{theorem6}

\setcounter{thm}{5}

\begin{thm}[Distance on the topology of an $N$-taxon principal tree]
Let $\mathcal{T}$ be a principal tree, with outgroup $o$. Suppose $\mathcal{T}$ is given the rooted triple metrization. Then the distance $d_\mathcal{T}\left(x,y\right)$ between leaf taxa $x$ and $y$ is
\begin{align*}
d_\mathcal{T}\left(x,y\right)=\begin{cases}
0 & \text{ if } $x=y$, \\
2N-2 & \text{ if $x\neq{}y$ and one of $x=o$, $y=o$}, \\
2\left|R_{x,y}\right|+2 & \text{ otherwise},
\end{cases}
\end{align*}
where $R_{x,y}$ is the set of rooted $4$-taxon principal trees displayed on $\mathcal{T}$ with outgroup $o$ displaying both $x$ and $y$, where $x$ and $y$ are non-sisters.

\end{thm}

\begin{proof}
Clearly, if $x=y$ then $d_\mathcal{T}\left(x,y\right)=0$.

Next suppose $x\neq{}y$ and one of $x=o$, $y=o$. With no loss of generality, assume $y=o$. Then
\begin{align*}
d_\mathcal{T}\left(x,y\right)=d_\mathcal{T}\left(x,o\right)=d_\mathcal{T}\left(x,v\right)+d_\mathcal{T}\left(v,o\right),
\end{align*}
where $v$ is the most recent common ancestor (MRCA) of $x$ and $o$. Since $y=o$, $v$ must be the root of $\mathcal{T}$. Then from the rooted triple metrization, by the same arguments as \cite{rhodes2019topological},
\begin{align*}
d_\mathcal{T}\left(x,v\right)=d_\mathcal{T}\left(v,o\right)=N-1
\end{align*}
and
\begin{align*}
d_\mathcal{T}\left(x,o\right)=2N-2.
\end{align*}

Finally, suppose $x\neq{}y$ and $x,y\neq{}o$. Again suppose that $v$ is the MRCA of $x$ and $y$. Then again by the same arguments as \cite{rhodes2019topological},
\begin{align*}
d_\mathcal{T}\left(x,y\right)=2k-2,
\end{align*}
where $k$ is the number of leaf taxa descended from $v$.

For $x$ and $y$ to be non-sisters on a rooted $4$-taxon principal tree displayed on $\mathcal{T}$ with outgroup $o$, we require the leaf taxon that is not $x$, $y$ or $o$ to be one of the $k-2$ leaf taxa descended from $v$ that is not $x$ or $y$. Thus,
\begin{align*}
\left|R_{x,y}\right|=k-2
\end{align*}
and
\begin{align*}
d_\mathcal{T}\left(x,y\right)=2\left|R_{x,y}\right|+2.
\end{align*}

\end{proof}

\section{Inferring topologies of \texorpdfstring{$N$}{N}-taxon principal trees}
\label{toptreeconsistency}

We prove that consistent inference of the topology of the $N$-taxon principal tree follows from consistent inference of the principal trees of the displayed $4$-taxon CDMs. However, it is possible that a displayed $4$-taxon CDM does not meet the assumptions of Section~\ref{ass}. Specifically, even if an $N$-taxon CDM meets the assumptions, some displayed $4$-taxon CDMs may have sister convergence. By assuming that all convergence parameters of the $N$-taxon CDM are sufficiently ``small'', then all convergence parameters of the displayed $4$-taxon CDMs, including those of sister convergence groups of the displayed $4$-taxon CDMs are ``small''. Then all topologies of the displayed $4$-taxon principal trees are inferred consistently by Algorithm~\ref{algorithmprincipaltree2}.

To prove this result, we first prove a proposition similar to Proposition~1.2 of \cite{haughton1988choice}. Proposition~1.2 states that if the generating model is among the set of candidate models, the probability that the model selected by the BIC is the generating model converges to $1$. Our adaptation relaxes Proposition~1.2, such that none of the candidate models are the generating model, but some candidate models are sufficiently ``close'' to the generating model. That is, the generating parameter is a ``small'' perturbation from a point in the parameter space of a candidate model. We then use our proposition to prove that all topologies of the displayed $4$-taxon principal trees are inferred consistently by Algorithm~\ref{algorithmprincipaltree2}.

For the following proposition, $f\left(X,\phi\right)=\exp\left(X\phi-b\left(\phi\right)\right)$ is the density for a regular exponential family, $m_1$ and $m_2$ are the parameter spaces of two models, $\interior\Theta$ is the interior of some topological space $\Theta$, $\overline{m}_1$ and $\overline{m}_2$ are the Zariski closures of $m_1$ and $m_2$, respectively and $E_{\theta}X_i=\nabla{}b\left(\theta\right)$ is the expected value of random variable $X_i$ given generating parameter $\theta$. The function $g\left(\phi\right)=\nabla{}b\left(\theta\right)\phi-b\left(\phi\right)$ for $\phi\in\Theta$ attains its unique maximum at $\theta$ \citep{barndorff1978information}.

\begin{proposition}
\label{conj}

Let $m_1$ and $m_2$ be two different models satisfying $m_1\cap{}m_2=\emptyset$. Then there exists some $\theta\in\interior\Theta$, $\theta\notin\overline{m}_1,\theta\notin\overline{m}_2$, with a neighborhood $\mathfrak{N}$ of $\theta$ such that $\mathfrak{N}\cap{}m_1=\emptyset$, $\mathfrak{N}\cap{}m_2\neq{}\emptyset$ and
\begin{align*}
\lim_{n\to\infty}P_{\theta}^{n}\left(\gamma\left(n,1\right)<\gamma\left(n,2\right)\right)=1.
\end{align*}

\begin{proof}

The proof requires only a slight modification to the proof of Proposition~1.2 of \cite{haughton1988choice}.

From \cite{haughton1988choice}, since $\mathfrak{N}\cap{}m_1=\emptyset$,
\begin{equation}
\sup_{\phi\in{}m_1\cap{}\Theta}\nabla{}b\left(\theta\right)\phi-b\left(\phi\right)+\epsilon\leq{}\nabla{}b\left(\theta\right)\theta-b\left(\theta\right) \label{eqn1}
\end{equation}
and asymptotically with probability $1$,
\begin{equation}
\left|\sup_{\phi\in{}m_i\cap{}\Theta}\left(Y_n\phi-b\left(\phi\right)\right)-\sup_{\phi\in{}m_i\cap{}\Theta}\nabla{}b\left(\theta\right)\phi-b\left(\phi\right)\right|<\frac{\epsilon}{4}, \label{eqn2}
\end{equation}
where $\epsilon>0$.

Although $\mathfrak{N}\cap{}m_2\neq{}\emptyset$, $g\left(\phi\right)$ attains its maximum at $\theta$ and $\theta\notin\overline{m}_2$. Thus, we can choose $\widetilde{\epsilon}>0$ such that
\begin{equation}
\sup_{\phi\in{}m_2\cap{}\Theta}\nabla{}b\left(\theta\right)\phi-b\left(\phi\right)+\widetilde{\epsilon}=\nabla{}b\left(\theta\right)\theta-b\left(\theta\right). \label{eqn3}
\end{equation}

We consider the two possible signs of the argument of the absolute value in Inequality~\eqref{eqn2}. If
\begin{align*}
\sup_{\phi\in{}m_i\cap{}\Theta}\left(Y_n\phi-b\left(\phi\right)\right)-\sup_{\phi\in{}m_i\cap{}\Theta}\nabla{}b\left(\theta\right)\phi-b\left(\phi\right)\geq{}0,
\end{align*}
then from Inequality~\ref{eqn2},
\begin{align*}
\sup_{\phi\in{}m_1\cap{}\Theta}\left(Y_n\phi-b\left(\phi\right)\right)<&\sup_{\phi\in{}m_1\cap{}\Theta}\nabla{}b\left(\theta\right)\phi-b\left(\phi\right)+\frac{\epsilon}{4}.
\end{align*}

Similarly, if
\begin{align*}
\sup_{\phi\in{}m_i\cap{}\Theta}\left(Y_n\phi-b\left(\phi\right)\right)-\sup_{\phi\in{}m_i\cap{}\Theta}\nabla{}b\left(\theta\right)\phi-b\left(\phi\right)<0,
\end{align*}
then
\begin{align*}
\sup_{\phi\in{}m_1\cap{}\Theta}\left(Y_n\phi-b\left(\phi\right)\right)<&\sup_{\phi\in{}m_1\cap{}\Theta}\nabla{}b\left(\theta\right)\phi-b\left(\phi\right) \\
<&\sup_{\phi\in{}m_1\cap{}\Theta}\nabla{}b\left(\theta\right)\phi-b\left(\phi\right)+\frac{\epsilon}{4}.
\end{align*}

Thus, from Inequalities~\eqref{eqn1}~and~\eqref{eqn2}, asymptotically with probability $1$,
\begin{equation}
\sup_{\phi\in{}m_1\cap{}\Theta}\left(Y_n\phi-b\left(\phi\right)\right)<\sup_{\phi\in{}m_1\cap{}\Theta}\nabla{}b\left(\theta\right)\phi-b\left(\phi\right)+\frac{\epsilon}{4} \leq{}\nabla{}b\left(\theta\right)\theta-b\left(\theta\right)-\frac{3\epsilon}{4}.
\label{eqn4}
\end{equation}

By similar arguments, from Inequality~\eqref{eqn2} and Equation~\eqref{eqn3}, asymptotically with probability $1$,
\begin{equation}
\sup_{\phi\in{}m_2\cap{}\Theta}\left(Y_n\phi-b\left(\phi\right)\right)>\sup_{\phi\in{}m_2\cap{}\Theta}\nabla{}b\left(\theta\right)\phi-b\left(\phi\right)-\frac{\epsilon}{4}=\nabla{}b\left(\theta\right)\theta-b\left(\theta\right)-\widetilde{\epsilon}-\frac{\epsilon}{4}.
\label{eqn5}
\end{equation}

By Inequalities~\eqref{eqn4}~and~\eqref{eqn5},
\begin{align*}
\sup_{\phi\in{}m_1\cap{}\Theta}\left(Y_n\phi-b\left(\phi\right)\right)<&\nabla{}b\left(\theta\right)\theta-b\left(\theta\right)-\frac{3\epsilon}{4} \\
=&\nabla{}b\left(\theta\right)\theta-b\left(\theta\right)-\widetilde{\epsilon}-\frac{\epsilon}{4}+\widetilde{\epsilon}-\frac{\epsilon}{2} \\
<&\sup_{\phi\in{}m_2\cap{}\Theta}\left(Y_n\phi-b\left(\phi\right)\right)+\widetilde{\epsilon}-\frac{\epsilon}{2} \\
=&\sup_{\phi\in{}m_2\cap{}\Theta}\left(Y_n\phi-b\left(\phi\right)\right)-\delta,
\end{align*}
where $\delta=\frac{\epsilon}{2}-\widetilde{\epsilon}$.

If it is possible to choose $\delta>0$,
then asymptotically with probability $1$,
\begin{align*}
\sup_{\phi\in{}m_1\cap{}\Theta}\left(Y_n\phi-b\left(\phi\right)\right)+\delta<\sup_{\phi\in{}m_2\cap{}\Theta}\left(Y_n\phi-b\left(\phi\right)\right).
\end{align*}

We are free to choose any $\theta\in\interior\Theta$. Thus, we choose $\theta$ to be an arbitrarily small perturbation from some point in $m_2$. Then $\widetilde{\epsilon}>0$ is arbitrarily small and $\delta>0$. The remainder of the proof then follows from \cite{haughton1988choice}.

\end{proof}

\end{proposition}

A convergence group on the generating $N$-taxon CDM may be a sister convergence group on some displayed $4$-taxon CDMs and a non-sister convergence group on others. Thus, we must assume that all convergence parameters of the generating $N$-taxon CDM are ``small'' relative to the divergence parameters.

\bigskip

Next, we adapt Theorem~3 of \cite{steel1992complexity} to prove that the $N$-taxon principal tree can be identified from the set of $4$-taxon principal trees that include the outgroup.

\begin{theorem}[Steel, 1992]
\label{corsteel1992}
For a set of rooted triples $R$, $\left<R\right>=\left\{T\right\}$ if and only if $R$ is consistent with $T$, and for each internal edge $e$ of $T$ there is a rooted triple in $R$ which distinguishes $e$.
\end{theorem}

The consequence of Theorem~\ref{corsteel1992} of \cite{steel1992complexity} is that if all trees of a set of (binary) rooted $3$-taxon trees $R$ are displayed on a (binary) rooted $N$-taxon tree $T$ and each internal edge of $T$ is an internal edge of at least one tree in $R$, then $T$ is the only $N$-taxon tree that displays all the $3$-taxon trees of $R$. In other words, the $N$-taxon tree $T$ can be identified from the set of $3$-taxon trees $R$.

\cite{steel1992complexity} note that an analogous theorem exists for unrooted quartets. Thus, the \emph{unrooted} principal tree of the $N$-taxon CDM can be identified from the set of $\binom{N-1}{3}$ topologies of \emph{unrooted} $4$-taxon principal trees that include the outgroup displayed on the \emph{unrooted} principal tree of the $N$-taxon CDM. The principal tree of the $N$-taxon CDM is then rooted by the outgroup.

\subsection{Proof of Theorem~\ref{principaltreeconsistent}}
\label{theorem7}

Finally, from Proposition~\ref{conj} and Theorem~\ref{corsteel1992} adapted to unrooted quartets that include the outgroup, we can prove Theorem~\ref{principaltreeconsistent}.

\setcounter{thm}{6}

\begin{thm}

Suppose CDM $\mathcal{N}$ has topology of principal tree $\mathcal{T}$. Suppose the BIC is used for model selection in step $2$ of Algorithm~\ref{algorithmprincipaltree2}. Suppose $\widehat{\mathcal{T}}$ is the estimate of $\mathcal{T}$ inferred by Algorithm~\ref{algorithmprincipaltree2}. Then there exists some constant $c>0$ such that if the largest convergence parameter of $\mathcal{N}$ is less than $c$,
\begin{align*}
\lim_{n\to\infty}\mathbb{P}\left(\widehat{\mathcal{T}}=\mathcal{T}\right)=1.
\end{align*}

\end{thm}

\begin{proof}

Suppose $\mathcal{N}$ has a displayed $4$-taxon CDM $\mathcal{N}_4$ with topology of principal tree $\mathcal{T}_4=\left(o,\left(a,\left(b,c\right)\right)\right)$. Then from the proof of Theorem~\ref{disttheorem}, for $\mathcal{N}_4$,
\begin{align*}
r_{0011}r_{1001}r_{1010}-r_{1011}^2=0,
\end{align*}
while for some $4$-taxon CDM with topology of principal tree $\mathcal{T}'_4\neq{}\left(o,\left(a,\left(b,c\right)\right)\right)$,
\begin{align*}
r_{0011}r_{1001}r_{1010}-r_{1011}^2>0.
\end{align*}

Suppose $m_1$ corresponds to the union of parameter spaces for CDMs $1-5$ for the topology of principal tree $\mathcal{T}_4$. Suppose also that $m_2$ corresponds to the union of parameter spaces for CDMs $1-5$ for any $4$-taxon topology of principal tree that is not $\mathcal{T}_4$. Then $m_1\cap{}m_2=\emptyset$ --- recall from Assumption~\ref{genparam} that the generating parameter is a generic point. Suppose $\theta\notin\overline{m}_1,\overline{m}_2$. Then if $c>0$ is sufficiently small, Proposition~\ref{conj} holds and $m_1$ is selected by the BIC asymptotically with probability $1$.

\bigskip

Next, we prove the claim that the set of inferred topologies of $4$-taxon principal trees equals the set of topologies of the principal trees of the $4$-taxon CDMs displayed on $\mathcal{N}$. Then from the adaptation of Theorem~\ref{corsteel1992} to unrooted quartets, the topology of the principal tree of $\mathcal{N}$ is the only topology that displays all inferred $4$-taxon principal trees. Thus, any consistent supertree inference method used in step~3 of Algorithm~\ref{algorithmprincipaltree2} infers the topology of the principal tree of $\mathcal{N}$ consistently and the proof is complete.

All that is left to prove is the claim that the probability of the set of inferred $4$-taxon principal trees equalling the set of topologies of principal trees of $4$-taxon CDMs displayed on $\mathcal{N}$ converges to $1$.

\bigskip

Suppose $A_i$ is the event where the topology of the $i^{th}$ $4$-taxon principal tree is inferred incorrectly, given some arbitrary order. Then, by Proposition~\ref{conj}, there exists some sample size $n$ such that for $n'>n$, $\mathbb{P}\left(A_i\right)<\epsilon_i$ for some arbitrarily small $\epsilon_i>0$. Then by Boole's inequality,
\begin{align*}
\mathbb{P}\left(\cup_{i=1}^{\binom{N-1}{3}}A_i\right)\leq{}\sum_{i=1}^{\binom{N-1}{3}}\mathbb{P}\left(A_i\right)<\sum_{i=1}^{\binom{N-1}{3}}\epsilon_i,
\end{align*}
an arbitrarily small positive quantity. Thus, the set of topologies of the inferred $4$-taxon principal trees of step $2$ of Algorithm~\ref{algorithmprincipaltree2} equals the set of topologies of the principal trees of the $4$-taxon CDMs displayed on $\mathcal{N}$ with probability converging to $1$.

\end{proof}

\section{Proof of Proposition~\ref{nonsistervtree}}
\label{equivalentCDMs}

\setcounter{prop}{7}

\begin{prop}
Let CDM $\mathcal{N}_i$ have topology of principal tree $\mathcal{T}_i$ and $t_{i,j}$ be the epoch length of epoch $j$. Let $\mathcal{N}_1$ be the tree, with $\mathcal{T}_1=\left(o,\left(a,\left(b,c\right)\right)\right)$, $t_{1,2}\to\infty$ and $t_{1,3}\to{}0$. Let $\mathcal{N}_2$ and $\mathcal{N}_3$ be CDMs with $\mathcal{T}_2=\left(o,\left(b,\left(a,c\right)\right)\right)$ and $\mathcal{T}_3=\left(o,\left(c,\left(a,b\right)\right)\right)$, each with a non-sister convergence group $\left\{\left\{b\right\},\left\{c\right\}\right\}$ in the tip epoch, with $t_{2,4},t_{3,4}\to\infty$. Then the sets of possible phylogenetic tensors of $\mathcal{N}_1$, $\mathcal{N}_2$ and $\mathcal{N}_3$ converge to the same set.
\end{prop}

\begin{proof}
Consider phylogenetic tensors for $\mathcal{N}_1$, $\mathcal{N}_2$ and $\mathcal{N}_3$ with arbitrary finite, positive substitution rates and epoch times. Let apostrophe superscripts denote parameters of $\mathcal{N}_2$.

For generic parameters on $\mathcal{N}_1$, after setting $x_6=x_7=x_8=x_9=x_{10}=x_{11}=1$, Equations~\eqref{gens} reduce to
\begin{align*}
\begin{cases}
f_{0011}=x_4x_5, \\
f_{0101}=x_2x_3x_4, \\
f_{0110}=x_2x_3x_5, \\
f_{0111}=x_2x_3x_4x_5, \\
f_{1001}=x_1x_2x_4, \\
f_{1010}=x_1x_2x_5, \\
f_{1011}=x_1x_2x_4x_5, \\
f_{1100}=x_1x_3, \\
f_{1101}=x_1x_2x_3x_4, \\
f_{1110}=x_1x_2x_3x_5, \\
f_{1111}=x_1x_3x_4x_5, \\
\phantom{f_{1110}}\nphantom{$f_{\delta}$}f_{\delta}=x_1x_2x_3x_4x_5.
\end{cases}
\end{align*}

Taking the limit as the epoch time of the second epoch diverges and the epoch time of the tip epoch converges to $0$ is equivalent to $x_1,x_2,x_3\to{}0$ and $x_4,x_5\to{}1$. Thus,
\begin{align*}
\begin{cases}
f_{0011}\to{}1, \\
f_{0101}\to{}0, \\
f_{0110}\to{}0, \\
f_{0111}\to{}0, \\
f_{1001}\to{}0, \\
f_{1010}\to{}0, \\
f_{1011}\to{}0, \\
f_{1100}\to{}0, \\
f_{1101}\to{}0, \\
f_{1110}\to{}0, \\
f_{1111}\to{}0, \\
\phantom{f_{1110}}\nphantom{$f_{\delta}$}f_{\delta}\to{}0.
\end{cases}
\end{align*}

For generic parameters on $\mathcal{N}_2$, after setting $x'_7=x'_8=x'_9=x'_{10}=x'_{11}=1$ and permuting leaf labels, Equations~\eqref{gens} reduce to
\begin{align*}
\begin{cases}
f_{0011}=1-x'_6\left(1-x'_2x'_3x'_5\right), \\
f_{0101}=x'_4x'_5x'_6, \\
f_{0110}=x'_2x'_3x'_4x'_6, \\
f_{0111}=x'_2x'_3x'_4x'_5x'_6, \\
f_{1001}=x'_1x'_2x'_5x'_6, \\
f_{1010}=x'_1x'_3x'_6, \\
f_{1011}=x'_1x'_2x'_3x'_5x'_6, \\
f_{1100}=x'_1x'_2x'_4, \\
f_{1101}=x'_1x'_2x'_4x'_5x'_6, \\
f_{1110}=x'_1x'_2x'_3x'_4x'_6, \\
f_{1111}=x'_1x'_4\left(x'_2\left(1-x'_6\right)+x'_3x'_5x'_6\right), \\
\phantom{f_{1110}}\nphantom{$f_{\delta}$}f_{\delta}=x'_1x'_2x'_3x'_4x'_5x'_6.
\end{cases}
\end{align*}

Taking the limit as the epoch time of the tip epoch diverges is equivalent to $x'_1,x'_4,x'_6\to{}0$. Thus,
\begin{align*}
\begin{cases}
f_{0011}\to{}1, \\
f_{0101}\to{}0, \\
f_{0110}\to{}0, \\
f_{0111}\to{}0, \\
f_{1001}\to{}0, \\
f_{1010}\to{}0, \\
f_{1011}\to{}0, \\
f_{1100}\to{}0, \\
f_{1101}\to{}0, \\
f_{1110}\to{}0, \\
f_{1111}\to{}0, \\
\phantom{f_{1110}}\nphantom{$f_{\delta}$}f_{\delta}\to{}0.
\end{cases}
\end{align*}

It is straightforward to show that Equations~\eqref{gens} have the same limits for $\mathcal{N}_3$ and this is omitted.

\end{proof}

\section{Controlling overfitting the CDM}
\label{overfitting}

For Algorithm~\ref{algorithmtopCDM} that infers the convergence groups of the $N$-taxon CDM, we introduce two constraints to avoid overfitting the CDM with too many convergence groups. The constraints are used to determine whether convergence groups are inferred on the CDM in a stepwise fashion.

The constraints ensure that only a small number of expected proportions of converging quartets can exceed their observed values. Furthermore, convergence groups can only be included if expected proportions exceed their observed values by small amounts. Suppose $\boldsymbol{O}$ and $\boldsymbol{E}$ are the matrices of observed and expected proportions of converging quartets. Suppose $C=\left\{c_1,c_2\right\}$ is an arbitrary convergence group, with $c_1\cup{}c_2\subset\left[N\right]$. Then $a_1=\max_{i\in{}c_1,j\in{}c_2}\left|\left[\boldsymbol{O}\right]_{ij}-\left[\boldsymbol{E}\right]_{ij}\right|$ and $a_2=\frac{1}{\left(N-1\right)^2}\sum_{i=1}^{N}\sum_{j=1}^{N}\delta_{ij}$, where $\delta_{ij}=0$ if $\left[\boldsymbol{E}\right]_{ij}\leq{}\left[\boldsymbol{O}\right]_{ij}$ and $1$ otherwise. We include the convergence group on the CDM only if $a_1\leq{}A_1$ and $a_2\leq{}A_2$, where $A_1,A_2\in\left[0,1\right]$. Note that the denominator of $a_2$ has $N-1$ since the outgroup is not involved in convergence. ($\left[\boldsymbol{O}\right]_{iN}=\left[\boldsymbol{O}\right]_{Nj}=\left[\boldsymbol{E}\right]_{iN}=\left[\boldsymbol{E}\right]_{Nj}=0$ for all $i,j\in\left\{1,2,\ldots{},N\right\}$.)

Further control of overfitting is achieved with a multiple comparisons correction, favoring $4$-taxon trees over non-tree CDMs. For a given $4$-taxon set that includes the outgroup taxon, the model selection criterion values are first converted into weights, for example, AIC or BIC \citep{anderson2004model}. These weights are a ``tree weight'' corresponding to the AIC or BIC of the tree and a ``non-tree weight'' corresponding to the AIC or BIC of the best fitting non-tree CDM. Tree weights could then be multiplied by some positive constant $b\geq{}1$ to achieve further control for overfitting. A multiple comparisons correction, such as the Holm-Bonferroni method \citep{holm1979simple}, could then be applied to the tree weights over all $4$-taxon sets that include the outgroup taxon, as if the weights were p-values. If the tree is ``rejected'', then the non-tree CDM with the lowest AIC or BIC is selected.

\section{Proof of Proposition~\ref{convratios}}
\label{convratiosproof}

\setcounter{prop}{8}

\begin{prop}
For convergence group $C=\left\{c_1,c_2\right\}$ on CDM $\mathcal{N}$, let $a\in{}c_1$ and $b\in{}c_2$. Let $v$ be the MRCA node of $a$ and $b$, $X_v$ be the set of leaf taxa descending from $v$ and $X_C=c_1\cup{}c_2$. Then the expected proportion of converging quartets for $\left\{a,b\right\}$ is
\begin{align*}
\frac{\left|X_v\setminus{}X_C\right|}{N-3}=\frac{\left|X_v\right|-\left|X_C\right|}{N-3},
\end{align*}
where $\left|X_v\right|$ and $\left|X_C\right|$ are the cardinalities of sets $X_v$ and $X_C$.
\end{prop}

\begin{proof}
To determine the expected proportions of converging quartets, suppose taxa $a$ and $b$ are converging. Then convergence between these taxa can only be inferred on $4$-taxon CDMs with topology of principal tree $\left(o,\left(a,\left(b,c\right)\right)\right)$ or $\left(o,\left(b,\left(a,c\right)\right)\right)$, for some arbitrary taxon $c$. With no loss of generality, we assume that the topology of the principal tree of some $4$-taxon CDM is $\left(o,\left(a,\left(b,c\right)\right)\right)$. To determine the expected proportions, we must determine the number of $4$-taxon CDMs displayed on $\mathcal{N}$, displaying both $a$ and $b$ where they appear as non-sisters.

We start with the rooted tree $\left(o,\left(a,b\right)\right)$ and append taxon $c$ and include a convergence group $C$. One edge corresponding to the convergence group $C$ must be ancestral to $a$, while the other must be ancestral to $b$. Thus, for $C$ to be a non-sister convergence group, the remaining taxon $c$ must be placed on an edge directly descended from $v$, corresponding to a speciation event before the epoch $C$ is in. Thus, $c$ could be any of the $\left|X_v\setminus{}X_C\right|=\left|X_v\right|-\left|X_C\right|$ taxa out of the $N-3$ possible taxa that are not $o$, $a$ or $b$.
\end{proof}

\section{Proof of Proposition~\ref{nosharedpair}}
\label{nosharedpairproof}

\setcounter{prop}{10}

\begin{prop}
An arbitrary pair of distinct convergence groups on CDM $\mathcal{N}$ share no pair of converging leaf taxa.
\end{prop}

\begin{proof}
Suppose $C_1$ and $C_2$ are two distinct convergence groups on $\mathcal{N}$. By Assumption~\ref{onecongroup} of Section~\ref{ass}, there can be at most one convergence group in each epoch. Thus, $C_1$ is either in an epoch before or after $C_2$. With no loss of generality, we assume that $C_1$ is in an epoch before $C_2$.

In order to share at least one pair of converging taxa, $C_2$ must be nested in $C_1$. However, by Assumption~\ref{nonest} of Section~\ref{ass}, there can be no convergence groups nested in other convergence groups.

\end{proof}

\section{Proof of Proposition~\ref{ident4taxoncdms}}
\label{ident4taxoncdmsa}

We assume that the topology of the principal tree of $\mathcal{N}$ is known. However, we note that if it is not known, from Theorem~\ref{principaltreeconsistent} it can be inferred consistently.

\setcounter{prop}{11}

\begin{prop}
The set of all convergence groups on CDM $\mathcal{N}$ can be identified from the set of displayed $4$-taxon CDMs after suppressing sister convergence groups.
\end{prop}

\begin{proof}

The set of displayed $4$-taxon CDMs after suppressing sister convergence groups defines a matrix of proportions of converging quartets. However, in general the set of all convergence groups on $\mathcal{N}$ cannot be identified from the matrix (see Figure~\ref{cdmssamematrix}). Instead, we can identify a \emph{set} of possible sets of convergence groups on $\mathcal{N}$ that correspond to the matrix of proportions of converging quartets. Since the set of displayed $4$-taxon CDMs after suppressing sister convergence groups is assumed known, for the remainder of the proof we can restrict to this set of sets of convergence groups. We must then prove that we can identify the specific set of all convergence groups of $\mathcal{N}$.

\bigskip

If $\mathcal{N}$ is a tree, then the set of displayed $4$-taxon CDMs after suppressing sister convergence groups is a set of trees. Thus, by Corollary~\ref{cor}, the matrix of proportions of converging quartets is the zero matrix. Alternatively, if $\mathcal{N}$ is not a tree, then $\mathcal{N}$ must have at least one non-sister convergence group. Call one such non-sister convergence group $C=\left\{c_1,c_2\right\}$, with $v$ the most recent common ancestral node of $c_1$ and $c_2$. Then by Proposition~\ref{convratios}, the expected proportion of converging quartets for $a\in{}c_1$ and $b\in{}c_2$ is $\frac{\left|X_v\right|-\left|X_C\right|}{N-3}$, where $X_v$ is the set of all taxa descending from $v$ and $\left|X_C\right|=\left|c_1\right|+\left|c_2\right|$. By the definition of non-sister convergence groups, $\left|X_v\right|-\left|X_C\right|>0$. Thus, the matrix of converging quartets is not the zero matrix. Thus, if $\mathcal{N}$ is a tree, the set of convergence groups can be identified from the set of displayed $4$-taxon CDMs after suppressing sister convergence groups via the matrix of converging quartets.

\bigskip

For the remainder of the proof, we can assume that $\mathcal{N}$ is not a tree. Then the set of non-sister convergence groups defines a set $S$ of $4$-taxon CDMs displayed on $\mathcal{N}$ with non-sister convergence groups after suppressing sister convergence groups --- note that $4$-taxon CDMs of $S$ can have one or two non-sister convergence groups. Suppose similarly that $S'$ is a set of $4$-taxon CDMs defined by a set of non-sister convergence groups not on $\mathcal{N}$ but with the same matrix of proportions of converging quartets as the set of non-sister convergence groups on $\mathcal{N}$. We must prove that there exists some $4$-taxon CDM in $S$ that is not in $S'$. Then we can identify the set of convergence groups on $\mathcal{N}$ from the set of $4$-taxon CDMs.

We prove that there is some $4$-taxon CDM in $S$ that is not in $S'$. We first consider an arbitrary $4$-taxon CDM $\mathcal{N}_4$ in $S$. Consider arbitrary leaf taxon pair $\left\{a,b\right\}$, where $a\in{}c_1$ and $b\in{}c_2$. Furthermore, assume $c\in{}X_v\setminus{}X_C$. Then with no loss of generality, we can assume the topology of the principal tree of $\mathcal{N}_4$ is $\left(o,\left(b,\left(a,c\right)\right)\right)$.

Suppose that $C'=\left\{c'_1,c'_2\right\}$ is one such non-sister convergence group that defines $S'$, with $c'_1$, $c'_2$, $v'$, $X_v'$ and $X_{C'}$ as in Proposition~\ref{convratios}. Now consider $4$-taxon CDM $\mathcal{N}'_4$, defined by $C'$ and on leaf taxon set $\left\{o,a,b,c\right\}$, with topology of principal tree $\left(o,\left(b,\left(a,c\right)\right)\right)$. Since we require a non-sister convergence group on $\mathcal{N}'_4$ where $a$ and $b$ are both converging, we must have either $a\in{}c'_1$ and $b\in{}c'_2$ or $a\in{}c'_2$ and $b\in{}c'_1$. With no loss of generality, we assume that $a\in{}c'_1$ and $b\in{}c'_2$. Then $c_1\subseteq{}c'_1$ or $c_1\supset{}c'_1$. Similarly, $c_2\subseteq{}c'_2$ or $c_2\supset{}c'_2$. Both $v$ and $v'$ are the MRCA of $a$ and $b$. Thus $v'=v$.

Now assume that $X_{C'}=X_C$. Then $c'_1=c_1$ and $c'_2=c_2$ and in turn, $C'=C$. Thus, $S'$ is defined by a set of convergence groups that includes $C$ and the $4$-taxon CDM is in $S'$. Thus, we can assume that $X_{C'}\neq{}X_C$ and we cannot have both $c'_1=c_1$ and $c'_2=c_2$. However, since the matrices of proportions of converging quartets must be the same for the two sets of convergence groups, we must have
\begin{align*}
\frac{\left|X_v\right|-\left|X_C\right|}{N-3}=\frac{\left|X_{v'}\right|-\left|X_{C'}\right|}{N-3},
\end{align*}
which simplifies to $\left|X_C\right|=\left|X_{C'}\right|$, since $v'=v$. Thus, either $c_1\subset{{}c'_1}$ and $c_2\supset{}c'_2$ or $c_1\supset{{}c'_1}$ and $c_2\subset{}c'_2$. With no loss of generality, we assume that $c_1\subset{{}c'_1}$ and $c_2\supset{}c'_2$.

Then there exists some choice of $c$ such that $c\in{}c_1'\setminus{}c_1$. For $c_2\supset{}c'_2$, there must similarly be some taxon $d\in{}c_2\setminus{}c'_2$. Thus, we are assuming that $N\geq{}5$ --- the outgroup and taxa $a$, $b$, $c$ and $d$. Then $a,c\in{}c'_1$ and $b\in{}c'_2$. Thus, before suppressing sister convergence groups to form $\mathcal{N}'_4$, $C'$ must correspond with a sister convergence group on the $4$-taxon CDM on leaf taxa $\left\{o,a,b,c\right\}$ --- see Figure~\ref{prop7} for a graphical depiction of $C$ and $C'$. Then any other choice of convergence group that defines $S'$, say $C''$, must satisfy $c''_1\supset{}c_1$ and the claim follows. Finally, since we have assumed $N\geq{}5$, we must also consider $N=4$. For $N=4$, it is clear from the identifiability and distinguishability of all CDMs with no sister convergence that the claim holds.

\end{proof}

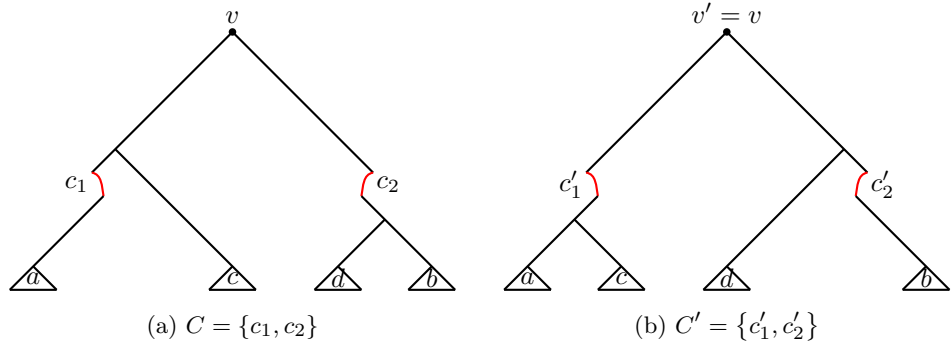
\begin{figure}[!htb]
\centering
\hspace*{\fill}
\begin{subfigure}{0.49\linewidth}
\centering
\begin{tikzpicture}[scale=0.31]
\draw[thick] (0,0) -- (-6,-6);
\draw[thick] (-5,-5) -- (1,-11);
\draw[thick] (0,-10) -- (-1,-11);
\draw[thick] (-1,-11) -- (1,-11);
\draw[thick,red] (-6,-6) to[out=350,in=100] (-5.5,-7);
\draw[thick] (-5.5,-7) -- (-9.5,-11);
\draw[thick] (-8.5,-10) -- (-7.5,-11);
\draw[thick] (-9.5,-11) -- (-7.5,-11);
\draw[thick] (0,0) -- (6,-6);
\draw[thick,red] (6,-6) to[out=190,in=80] (5.5,-7);
\draw[thick] (5.5,-7) -- (9.5,-11);
\draw[thick] (6.5,-8) -- (3.5,-11);
\draw[thick] (4.5,-10) -- (5.5,-11);
\draw[thick] (3.5,-11) -- (5.5,-11);
\draw[thick] (8.5,-10) -- (7.5,-11);
\draw[thick] (7.5,-11) -- (9.5,-11);
\node[] at (-8.5,-10.5) {$a$};
\node[] at (0,-10.5) {$c$};
\node[] at (4.5,-10.5) {$d$};
\node[] at (8.5,-10.5) {$b$};
\node[left] at (-5.75,-6.5) {$c_1$};
\node[right] at (5.75,-6.5) {$c_2$};
\tkzDefPoint(0,0){v}
\tkzLabelPoint[above](v){$v$}
\node at (v)[circle,fill,inner sep=1pt]{};
\end{tikzpicture}
\caption{$C=\left\{c_1,c_2\right\}$}
\label{C}
\end{subfigure}
\hfill
\begin{subfigure}{0.49\linewidth}
\centering
\begin{tikzpicture}[scale=0.31]
\draw[thick] (0,0) -- (6,-6);
\draw[thick] (5,-5) -- (-1,-11);
\draw[thick] (0,-10) -- (1,-11);
\draw[thick] (1,-11) -- (-1,-11);
\draw[thick,red] (6,-6) to[out=190,in=80] (5.5,-7);
\draw[thick] (5.5,-7) -- (9.5,-11);
\draw[thick] (8.5,-10) -- (7.5,-11);
\draw[thick] (9.5,-11) -- (7.5,-11);
\draw[thick] (0,0) -- (-6,-6);
\draw[thick,red] (-6,-6) to[out=350,in=100] (-5.5,-7);
\draw[thick] (-5.5,-7) -- (-9.5,-11);
\draw[thick] (-6.5,-8) -- (-3.5,-11);
\draw[thick] (-4.5,-10) -- (-5.5,-11);
\draw[thick] (-3.5,-11) -- (-5.5,-11);
\draw[thick] (-8.5,-10) -- (-7.5,-11);
\draw[thick] (-7.5,-11) -- (-9.5,-11);
\node[] at (-8.5,-10.5) {$a$};
\node[] at (-4.5,-10.5) {$c$};
\node[] at (0,-10.5) {$d$};
\node[] at (8.5,-10.5) {$b$};
\node[right] at (5.75,-6.5) {$c'_2$};
\node[left] at (-5.75,-6.5) {$c'_1$};
\tkzDefPoint(0,0){v}
\tkzLabelPoint[above](v){$v'=v$}
\node at (v)[circle,fill,inner sep=1pt]{};
\end{tikzpicture}
\caption{$C'=\left\{c'_1,c'_2\right\}$}
\label{Cdash}
\end{subfigure}
\hspace*{\fill}
\caption{Convergence groups $C$ and $C'$. Labels $c_1$, $c_2$, $c'_1$ and $c'_2$ indicate the set of leaf taxa below that edge. Triangles are displayed CDMs. Labels inside triangles indicate one of possibly many taxa on leaves of those displayed CDMs. There may be more displayed CDMs not drawn below $v$ that are not below either $C$ or $C'$}
\label{prop7}
\end{figure}

\section{Proof of Theorem~\ref{conscongroups}}
\label{theorem13}

\setcounter{thm}{12}

\begin{thm}
Suppose CDM $\mathcal{N}$ has topology of principal tree $\mathcal{T}$ and convergence groups $\mathcal{G}$. Suppose for all $l$, $\alpha_l=\beta_l$. Suppose for convergence group $\mathcal{C}_i=\left\{c_{1,i},c_{2,i}\right\}$ that if $a\in{}c_{1,i}\cup{}c_{2,i}$, then $a\notin{}c_{1,j}\cup{}c_{2,j}$ for any $j\neq{}i$. Suppose $\mathcal{T}$ is input into Algorithm~\ref{algorithmtopCDM}, the BIC is used for model selection in step~2, there are no multiple comparisons corrections and the tolerance criterion is $u=1$. Suppose $\widehat{\mathcal{G}}$ is the estimate of $\mathcal{G}$ inferred by Algorithm~\ref{algorithmtopCDM}. Then there exists some constant $c>0$ such that if the largest convergence parameter of $\mathcal{N}$ is less than $c$,
\begin{align*}
\lim_{n\to\infty}\mathbb{P}\left(\widehat{\mathcal{G}}=\mathcal{G}\right)=1.
\end{align*}
\end{thm}

\begin{proof}

We start by finding expressions for the transformed phylogenetic tensors for various $4$-taxon CDMs with and without sister convergence. We prove that the CDMs with sister convergence are not distinguishable from the CDMs with the sister convergence groups suppressed. Thus, regardless of whether the $4$-taxon CDMs have sister convergence groups or not, the non-sister convergence group is inferred consistently.

Since $\alpha_l=\beta_l$, $\gamma=0$ and the transformed phylogenetic tensor for a $4$-taxon CDM of Equation~\eqref{hadphylotensor} simplifies to
\begin{align*}
\widehat{P}=&\left[\begin{array}{c}
1 \\
0 \\
0 \\
r_{0011} \\
0 \\
r_{0101} \\
r_{0110} \\
0 \\
0 \\
r_{1001} \\
r_{1010} \\
0 \\
r_{1100} \\
0 \\
0 \\
r_{1111}
\end{array}\right].
\end{align*}

\bigskip

With no loss of generality, we assume the topology of the principal tree of an arbitrary $4$-taxon CDM displayed on $\mathcal{N}$ is $\left(o\left(a,\left(b,c\right)\right)\right)$. Then of the possible convergence groups on the $4$-taxon CDM, the convergence group in the epoch closest to the root is the sister convergence group $C=\left\{\left\{a\right\},\left\{b,c\right\}\right\}$. Thus, we first consider the distinguishability of two $4$-taxon CDMs, one a tree and the other with this convergence group. For both CDMs we assume the tip epoch has epoch time $0$. The first, which we call $\mathcal{N}_{4,1}$, is the tree $\left(o,\left(a,\left(b,c\right)\right)\right)$. Since the tip epoch has epoch time $0$, taxa $b$ and $c$ are identical. The second CDM, which we call $\mathcal{N}_{4,2}$, has a single convergence group, $C=\left\{\left\{a\right\},\left\{b,c\right\}\right\}$, followed by a speciation event involving $b$ and $c$. Again, since the tip epoch has epoch time $0$, taxa $b$ and $c$ are identical. See Figure~\ref{thm4cdms} for a graphical depiction of the two CDMs. Suppose $\mathcal{N}_{4,1}$ has parameters with no apostrophes and $\mathcal{N}_{4,2}$ has parameters with apostrophes.

For $\mathcal{N}_{4,1}$ (see Mathematica file S12.nb (text version S13.txt) on \url{https://github.com/jonathanmitchell88/CDMsSI} for a derivation),
\begin{align*}
\begin{cases}
r_{0011}=&1, \\
r_{0101}=&x_2x_3, \\
r_{0110}=&x_2x_3, \\
r_{1001}=&x_1x_2, \\
r_{1010}=&x_1x_2, \\
r_{1100}=&x_1x_3, \\
r_{1111}=&x_1x_3.
\end{cases}
\end{align*}

For $\mathcal{N}_{4,2}$ (see Mathematica file S12.nb (text version S13.txt) for a derivation),
\begin{align*}
\begin{cases}
r_{0011}=&1, \\
r_{0101}=&1-x'_4\left(1-x'_2x'_3\right), \\
r_{0110}=&1-x'_4\left(1-x'_2x'_3\right), \\
r_{1001}=&x'_1x'_2x'_4, \\
r_{1010}=&x'_1x'_2x'_4, \\
r_{1100}=&x'_1x'_3x'_4, \\
r_{1111}=&x'_1x'_3x'_4.
\end{cases}
\end{align*}

For $\mathcal{N}_{4,1}$,
\begin{align*}
\begin{cases}
\phantom{r_{0101}}\nphantom{$x_1$}x_1=&\sqrt{\frac{r_{1001}r_{1100}}{r_{0101}}}, \\
\phantom{r_{0101}}\nphantom{$x_2$}x_2=&\sqrt{\frac{r_{0101}r_{1001}}{r_{1100}}}, \\
\phantom{r_{0101}}\nphantom{$x_3$}x_3=&\sqrt{\frac{r_{0101}r_{1100}}{r_{1001}}}, \\
r_{0101}=&r_{0110}, \\
r_{1001}=&r_{1010}, \\
r_{1100}=&r_{1111}.
\end{cases}
\end{align*}

Since $x_1,x_2,x_3\in\left(0,1\right)$, for $\mathcal{N}_{4,1}$,
\begin{align*}
\begin{cases}
\phantom{r_{0101}r_{1001}}\nphantom{$r_{0101}$}r_{0101}=&r_{0110}, \\
\phantom{r_{0101}r_{1001}}\nphantom{$r_{1001}$}r_{1001}=&r_{1010}, \\
\phantom{r_{0101}r_{1001}}\nphantom{$r_{1100}$}r_{1100}=&r_{1111}, \\
r_{0101}r_{1001}<&r_{1100}, \\
r_{0101}r_{1100}<&r_{1001}, \\
r_{1001}r_{1100}<&r_{0101}.
\end{cases}
\end{align*}

Similarly, for $\mathcal{N}_{4,2}$,
\begin{align*}
\begin{cases}
r_{0101}=&r_{0110}, \\
r_{1001}=&r_{1010}, \\
r_{1100}=&r_{1111}.
\end{cases}
\end{align*}

Since we are assuming that all convergence parameters of $\mathcal{N}$ are less than some constant $c>0$, we can assume that $x'_4=1-\epsilon$, where $\epsilon>0$ is some small positive constant. Then for $\mathcal{N}_{4,2}$,
\begin{align*}
\begin{cases}
r_{1100}-r_{0101}r_{1001}=&x'_{1}x'_{3}\left(1-x_{2}^{'2}\right)+O\left(\epsilon\right), \\
r_{1001}-r_{0101}r_{1100}=&x'_{1}x'_{2}\left(1-x_{3}^{'2}\right)+O\left(\epsilon\right), \\
r_{0101}-r_{1001}r_{1100}=&x'_{2}x'_{3}\left(1-x_{1}^{'2}\right)+O\left(\epsilon\right).
\end{cases}
\end{align*}

Since $c>0$ can be chosen, there exists some choice of $\epsilon>0$ sufficiently small such that for $\mathcal{N}_{4,2}$,
\begin{align*}
\begin{cases}
\phantom{r_{0101}r_{1001}}\nphantom{$r_{0101}$}r_{0101}=&r_{0110}, \\
\phantom{r_{0101}r_{1001}}\nphantom{$r_{1001}$}r_{1001}=&r_{1010}, \\
\phantom{r_{0101}r_{1001}}\nphantom{$r_{1100}$}r_{1100}=&r_{1111}, \\
r_{0101}r_{1001}<&r_{1100}, \\
r_{0101}r_{1100}<&r_{1001}, \\
r_{1001}r_{1100}<&r_{0101}.
\end{cases}
\end{align*}

Thus, $\mathcal{N}_{4,1}$ and $\mathcal{N}_{4,2}$ are not distinguishable for this choice of $c>0$. Thus, any $4$-taxon CDM with $\alpha_l=\beta_l$ and this sister convergence group is not distinguishable from the CDM that results from suppressing the sister convergence. Thus, to determine the transformed phylogenetic tensor of any $4$-taxon CDM with $\alpha_l=\beta_l$, we can assume there is no sister convergence in this epoch.

\bigskip

The next closest epoch to the root that could have a convergence group is the epoch just after taxa $b$ and $c$ have diverged. Thus, we compare the tree $\left(o,\left(a,\left(b,c\right)\right)\right)$, which we call $\mathcal{N}_{4,3}$, and the CDM with topology of principal tree $\left(o,\left(a,\left(b,c\right)\right)\right)$ and sister convergence group $\left\{\left\{b\right\},\left\{c\right\}\right\}$ in the tip epoch, which we call $\mathcal{N}_{4,4}$. See Figure~\ref{thm4cdms2} for a graphical depiction of the two CDMs. Again, suppose $\mathcal{N}_{4,3}$ has parameters with no apostrophes and $\mathcal{N}_{4,4}$ has parameters with apostrophes.

For $\mathcal{N}_{4,3}$ (see Mathematica file S12.nb (text version S13.txt) for a derivation),
\begin{align*}
\begin{cases}
r_{0011}=&x_4x_5, \\
r_{0101}=&x_2x_3x_4, \\
r_{0110}=&x_2x_3x_5, \\
r_{1001}=&x_1x_2x_4, \\
r_{1010}=&x_1x_2x_5, \\
r_{1100}=&x_1x_3, \\
r_{1111}=&x_1x_3x_4x_5.
\end{cases}
\end{align*}

For $\mathcal{N}_{4,4}$ (see Mathematica file S12.nb (text version S13.txt) for a derivation),
\begin{align*}
\begin{cases}
r_{0011}=&1-x'_6\left(1-x'_4x'_5\right), \\
r_{0101}=&x'_2x'_3x'_4x'_6, \\
r_{0110}=&x'_2x'_3x'_5x'_6, \\
r_{1001}=&x'_1x'_2x'_4x'_6, \\
r_{1010}=&x'_1x'_2x'_5x'_6, \\
r_{1100}=&x'_1x'_3, \\
r_{1111}=&x'_1x'_3\left(1-x'_6\left(x'_4x'_5\right)\right).
\end{cases}
\end{align*}

For $\mathcal{N}_{4,3}$,
\begin{align*}
\begin{cases}
\phantom{r_{0101}r_{1010}}\nphantom{$x_1$}x_1=&\sqrt{\frac{r_{1001}r_{1100}}{r_{0101}}}, \\
\phantom{r_{0101}r_{1010}}\nphantom{$x_2$}x_2=&\sqrt{\frac{r_{0110}r_{1001}}{r_{0011}r_{1100}}}, \\
\phantom{r_{0101}r_{1010}}\nphantom{$x_3$}x_3=&\sqrt{\frac{r_{0101}r_{1100}}{r_{1001}}}, \\
\phantom{r_{0101}r_{1010}}\nphantom{$x_4$}x_4=&\sqrt{\frac{r_{0011}r_{0101}}{r_{0110}}}, \\
\phantom{r_{0101}r_{1010}}\nphantom{$x_5$}x_5=&\sqrt{\frac{r_{0011}r_{0110}}{r_{0101}}}, \\
r_{0101}r_{1010}=&r_{0110}r_{1001}, \\
r_{0011}r_{1100}=&r_{1111}.
\end{cases}
\end{align*}

Since $x_1,x_2,x_3,x_4,x_5\in\left(0,1\right)$, for $\mathcal{N}_{4,3}$,
\begin{align*}
\begin{cases}
r_{0101}r_{1010}=&r_{0110}r_{1001}, \\
r_{0011}r_{1100}=&r_{1111}, \\
r_{0011}r_{0101}<&r_{0110}, \\
r_{0011}r_{0110}<&r_{0101}, \\
r_{0101}r_{1100}<&r_{1001}, \\
r_{0110}r_{1001}<&r_{0011}r_{1100}, \\
r_{1001}r_{1100}<&r_{0101}.
\end{cases}
\end{align*}

Similarly, for $\mathcal{N}_{4,4}$,
\begin{align*}
\begin{cases}
r_{0101}r_{1010}=&r_{0110}r_{1001}, \\
r_{0011}r_{1100}=&r_{1111}.
\end{cases}
\end{align*}

Since we are assuming that all convergence parameters of $\mathcal{N}$ are less than some constant $c>0$, we can assume that $x'_6=1-\epsilon$, where $\epsilon>0$ is some small positive constant. Then for $\mathcal{N}_{4,4}$,
\begin{align*}
\begin{cases}
\phantom{r_{0011}r_{1100}-r_{0110}r_{1001}}\nphantom{$r_{0110}-r_{0011}r_{0101}$}r_{0110}-r_{0011}r_{0101}=&x'_2x'_3x'_5\left(1-x_4^{'2}\right)+O\left(\epsilon\right), \\
\phantom{r_{0011}r_{1100}-r_{0110}r_{1001}}\nphantom{$r_{0101}-r_{0011}r_{0110}$}r_{0101}-r_{0011}r_{0110}=&x'_2x'_3x'_4\left(1-x_5^{'2}\right)+O\left(\epsilon\right), \\
\phantom{r_{0011}r_{1100}-r_{0110}r_{1001}}\nphantom{$r_{1001}-r_{0101}r_{1100}$}r_{1001}-r_{0101}r_{1100}=&x'_1x'_2x'_4\left(1-x_3^{'2}\right)+O\left(\epsilon\right), \\
r_{0011}r_{1100}-r_{0110}r_{1001}=&x'_1x'_3x'_4x'_5\left(1-x_2^{'2}\right)+O\left(\epsilon\right), \\
\phantom{r_{0011}r_{1100}-r_{0110}r_{1001}}\nphantom{$r_{0101}-r_{1001}r_{1100}$}r_{0101}-r_{1001}r_{1100}=&x'_2x'_3x'_4\left(1-x_1^{'2}\right)+O\left(\epsilon\right).
\end{cases}
\end{align*}

Since $c>0$ can be chosen, there exists some choice of $\epsilon>0$ sufficiently small such that for $\mathcal{N}_{4,4}$,
\begin{align*}
\begin{cases}
r_{0101}r_{1010}=&r_{0110}r_{1001}, \\
r_{0011}r_{1100}=&r_{1111}, \\
r_{0011}r_{0101}<&r_{0110}, \\
r_{0011}r_{0110}<&r_{0101}, \\
r_{0101}r_{1100}<&r_{1001}, \\
r_{0110}r_{1001}<&r_{0011}r_{1100}, \\
r_{1001}r_{1100}<&r_{0101}.
\end{cases}
\end{align*}

Thus, $\mathcal{N}_{4,3}$ and $\mathcal{N}_{4,4}$ are not distinguishable for this choice of $c>0$. Thus, any $4$-taxon CDM with $\alpha_l=\beta_l$ and this sister convergence group is not distinguishable from the CDM that results from suppressing the sister convergence. Thus, to determine the transformed phylogenetic tensor of any $4$-taxon CDM with $\alpha_l=\beta_l$, we can again assume there is no sister convergence in this epoch.

\bigskip

By the assumption that no leaf taxa belong to more than one convergence group, there can be no more than one convergence group on any arbitrary $4$-taxon CDM displayed on $\mathcal{N}$. Thus, taking into consideration $\mathcal{N}_{4,1}$ and $\mathcal{N}_{4,2}$ not being distinguishable and $\mathcal{N}_{4,3}$ and $\mathcal{N}_{4,4}$ not being distinguishable, we can conclude that any arbitrary $4$-taxon CDM displayed on $\mathcal{N}$ is not distinguishable from the $4$-taxon CDM that results from suppressing any sister convergence group, which is one of CDM $1-3$ of Figure~\ref{4taxonCDMs2}.

Next, we establish that CDM $3$ is identifiable under these assumptions. For this CDM, which we call $\mathcal{N}_{4,5}$ (see Mathematica file S12.nb (text version S13.txt)),
\begin{align*}
\begin{cases}
r_{0011}=&x_4x_5x_6x_7, \\
r_{0101}=&x_2x_3x_4x_6x_8, \\
r_{0110}=&x_7x_8\left(1-x_6\left(1-x_2x_3x_5\right)\right), \\
r_{1001}=&x_1x_2x_4, \\
r_{1010}=&x_1x_2x_5x_6x_7, \\
r_{1100}=&x_1x_3x_6x_8, \\
r_{1111}=&x_1x_4x_7x_8\left(x_2\left(1-x_6\right)+x_3x_5x_6\right).
\end{cases}
\end{align*}

In terms of the set of parameters $\left\{y_1,y_2,y_3,y_4,y_5,y_6,y_7,y_8,y_9\right\}$ of Appendix~\ref{CDM5appendix},
\begin{align*}
\begin{cases}
r_{0011}=&y_4y_5y_6, \\
r_{0101}=&y_2y_3y_4y_6, \\
r_{0110}=&y_7\left(1-y_6\right)+y_2y_3y_5y_6, \\
r_{1001}=&y_1y_2y_4, \\
r_{1010}=&y_1y_2y_5y_6, \\
r_{1100}=&y_1y_3y_6, \\
r_{1111}=&y_1\left(y_2y_4y_7\left(1-y_6\right)+y_3y_5y_5y_6\right).
\end{cases}
\end{align*}

In S14.m2 (output file S15.txt) on \url{https://github.com/jonathanmitchell88/CDMsSI}, we see that the set of parameters $\left\{y_1,y_2,y_3,y_4,y_5,y_6,y_7\right\}$ is identifiable. It follows that CDMs $1$~and~$2$ are also identifiable.

\bigskip

Thus, using similar arguments to those of the proof of Theorem~\ref{principaltreeconsistent}, with probability converging to $1$, step $2$ of Algorithm~\ref{algorithmtopCDM} infers all the $4$-taxon CDMs with the outgroup that are displayed on $\mathcal{N}$ after suppressing sister convergence groups.

If $\mathcal{N}$ is a tree, then $s=0$ in step $4$ of Algorithm~\ref{algorithmtopCDM}, the algorithm terminates and the tree is returned. If $\mathcal{N}$ is not a tree, since $u=1$, a potential convergence group on $\mathcal{N}$ is only considered if, for all pairs of converging taxa in the convergence group, the inferred $4$-taxon CDMs with that pair of taxa as non-sisters all have the pair converging. Thus, asymptotically with probability $1$, only convergence groups on $\mathcal{N}$ can be on the inferred $N$-taxon CDM. If not all convergence groups of $\mathcal{N}$ have been included on the inferred CDM, then there are some elements of $\boldsymbol{O}$ that are non-zero corresponding to elements of $\boldsymbol{E}$ that are zero. These elements correspond to the pairs of converging taxa in convergence groups of $\mathcal{N}$ that are not yet on the inferred CDM. Including these convergence groups on the inferred CDM makes these elements of $\boldsymbol{E}$ equal to the corresponding elements of $\boldsymbol{O}$, decreasing the sum of squared differences. Once all convergence groups of $\mathcal{N}$ have been appended to the inferred CDM, $\boldsymbol{O}=\boldsymbol{E}$. Thus, no more convergence groups can be appended to the inferred CDM to decrease the sum of squared differences and the algorithm terminates.

\end{proof}

\begin{figure}[!htb]
\centering
\hspace*{\fill}
\begin{subfigure}{0.49\linewidth}
\centering
\begin{tikzpicture}[scale=1]
\draw[thick] (0,0) -- (-2.1,-2.1);
\draw[thick] (0,0) -- (2.1,-2.1);
\draw[thick] (1,-1) -- (-0.1,-2.1);
\draw[thick] (2,-2) -- (1.9,-2.1);
\node[below] at (-2.1,-2.1) {\strut{$o$}};
\node[below] at (-0.1,-2.1) {\strut{$a$}};
\node[below] at (1.9,-2.1) {\strut{$b$}};
\node[below] at (2.1,-2.1) {\strut{$c$}};
\draw[dashed] (-2.6,0) -- (2.6,0);
\draw[dashed] (-2.6,-1) -- (2.6,-1);
\draw[dashed] (-2.6,-2) -- (2.6,-2);
\draw[dashed] (-2.6,-2.1) -- (2.6,-2.1);
\node[left] at (-0.5,-0.5) {$x_1$};
\node[right] at (1.5,-1.5) {$x_2$};
\node[left] at (0.5,-1.5) {$x_3$};
\node[right] at (2.6,-2.05) {\footnotesize{$0$}};
\end{tikzpicture}
\caption{$\mathcal{N}_{4,1}$ --- $x_1$ represents the entire outgroup edge when unrooted.}
\label{N41}
\end{subfigure}
\hfill
\begin{subfigure}{0.49\linewidth}
\centering
\begin{tikzpicture}[scale=1]
\draw[thick] (0,0) -- (-2.6,-2.6);
\draw[thick] (0,0) -- (2,-2);
\draw[thick] (1,-1) -- (0,-2);
\draw[thick,red] (0,-2) to[out=350,in=100] (0.25,-2.5);
\draw[thick,red] (2,-2) to[out=190,in=80] (1.75,-2.5);
\draw[thick] (0.25,-2.5) -- (0.15,-2.6);
\draw[thick] (1.75,-2.5) -- (1.65,-2.6);
\draw[thick] (1.75,-2.5) -- (1.85,-2.6);
\node[below] at (-2.6,-2.6) {\strut{$o$}};
\node[below] at (0.15,-2.6) {\strut{$a$}};
\node[below] at (1.65,-2.6) {\strut{$b$}};
\node[below] at (1.85,-2.6) {\strut{$c$}};
\draw[dashed] (-3.1,0) -- (2.35,0);
\draw[dashed] (-3.1,-1) -- (2.35,-1);
\draw[dashed] (-3.1,-2) -- (2.35,-2);
\draw[dashed] (-3.1,-2.5) -- (2.35,-2.5);
\draw[dashed] (-3.1,-2.6) -- (2.35,-2.6);
\node[left] at (-0.5,-0.5) {$x'_1$};
\node[right] at (1.5,-1.5) {$x'_2$};
\node[left] at (0.5,-1.5) {$x'_3$};
\node[right] at (1.875,-2.25) {$x'_4$};
\node[left] at (0.125,-2.25) {$x'_4$};
\node[right] at (2.35,-2.55) {\footnotesize{$0$}};
\end{tikzpicture}
\caption{$\mathcal{N}_{4,2}$ --- $x'_1$ represents the entire outgroup edge when unrooted.}
\label{N42}
\end{subfigure}
\hspace*{\fill}
\caption{Two CDMs that are not distinguishable under the assumptions of Theorem~\ref{conscongroups}}
\label{thm4cdms}
\end{figure}
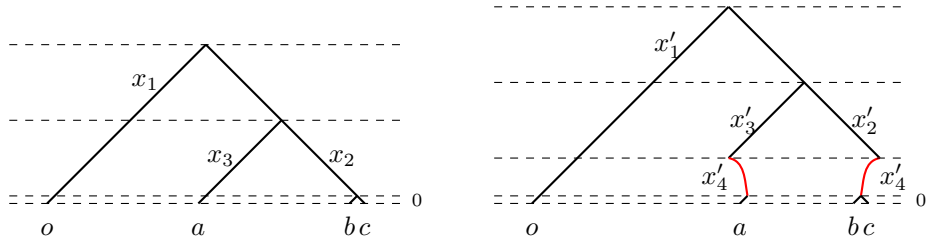

\begin{figure}[!htb]
\centering
\hspace*{\fill}
\begin{subfigure}{0.49\linewidth}
\centering
\begin{tikzpicture}[scale=0.85]
\draw[thick] (0,0) -- (-3,-3);
\draw[thick] (0,0) -- (3,-3);
\draw[thick] (1,-1) -- (-1,-3);
\draw[thick] (2,-2) -- (1,-3);
\node[left] at (-1.5,-1.5) {$x_1$};
\node[right] at (1.5,-1.5) {$x_2$};
\node[left] at (0,-2) {$x_3$};
\node[right] at (2.5,-2.5) {$x_4$};
\node[left] at (1.5,-2.5) {$x_5$};
\node[below] at (-3,-3) {\strut{$o$}};
\node[below] at (-1,-3) {\strut{$a$}};
\node[below] at (1,-3) {\strut{$b$}};
\node[below] at (3,-3) {\strut{$c$}};
\end{tikzpicture}
\caption{$\mathcal{N}_{4,3}$}
\label{N43}
\end{subfigure}
\hfill
\begin{subfigure}{0.49\linewidth}
\centering
\begin{tikzpicture}[scale=0.85]
\draw[thick] (0,0) -- (-3.5,-3.5);
\draw[thick] (0,0) -- (3,-3);
\draw[thick] (1,-1) -- (-1.5,-3.5);
\draw[thick] (2,-2) -- (1,-3);
\draw[thick,red] (1,-3) to[out=350,in=100] (1.25,-3.5);
\draw[thick,red] (3,-3) to[out=190,in=80] (2.75,-3.5);
\node[left] at (-1.75,-1.75) {$x'_1$};
\node[right] at (1.5,-1.5) {$x'_2$};
\node[left] at (-0.25,-2.25) {$x'_3$};
\node[right] at (2.5,-2.5) {$x'_4$};
\node[left] at (1.5,-2.5) {$x'_5$};
\node[right] at (2.875,-3.25) {$x'_6$};
\node[left] at (1.125,-3.25) {$x'_6$};
\node[below] at (-3.5,-3.5) {\strut{$o$}};
\node[below] at (-1.5,-3.5) {\strut{$a$}};
\node[below] at (1.25,-3.5) {\strut{$b$}};
\node[below] at (2.75,-3.5) {\strut{$c$}};
\end{tikzpicture}
\caption{$\mathcal{N}_{4,4}$}
\label{N44}
\end{subfigure}
\hspace*{\fill}
\caption{Two CDMs that are not distinguishable under the assumptions of Theorem~\ref{conscongroups}}
\label{thm4cdms2}
\end{figure}
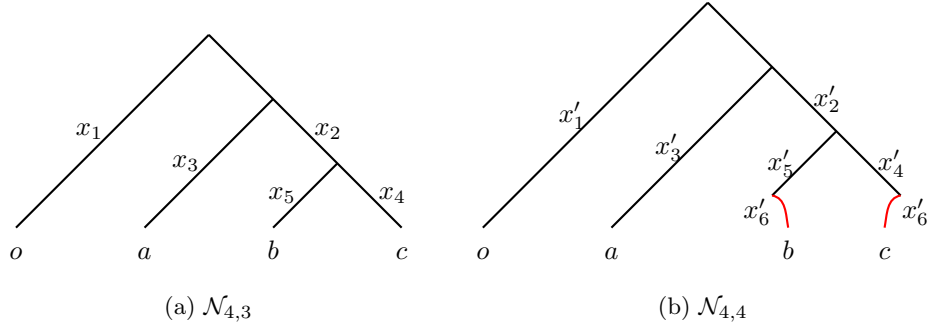

\section{Inferring convergence group orders on \texorpdfstring{$N$}{N}-taxon CDMs}
\label{grouporders}

The next algorithms infer partial orders on the convergence groups and determine whether or not there is a convergence group in the tip epoch. CDMs $4$ and $5$ have two convergence groups and thus provide power to determine convergence group orders. Whether or not there is a convergence group in the tip epoch can also be determined from the inferred $4$-taxon CDMs. For example, CDM $2$ and CDM $3$ differ by CDM $2$ having its convergence group in the tip epoch versus CDM $3$ having its convergence group in the epoch before the tip epoch.

Suppose an inferred $4$-taxon CDM has two non-sister convergence groups. One of the edges of the $4$-taxon principal tree corresponds to a converging taxon in both convergence groups. The order of these convergence groups may not be determined by the matrix of edge partial orders from Algorithm~\ref{algorithmtopCDM}. If the convergence group order is not determined, we determine which order is best supported by selecting CDMs from those with the appropriate convergence groups with a model selection procedure.

For convergence groups $C_i$ and $C_j$, whose partial order has not been determined, we tally proportions of $4$-taxon CDMs best supported by the two orders to obtain a matrix of ``observed'' convergence group order ratios. Convergence group orders are resolved in a stepwise fashion by minimizing the sum of squared differences between matrices of observed and ``expected'' partial convergence group order ratios. The matrix of inferred convergence group order ratios is updated after each convergence group order is inferred. We discard $4$-taxon CDMs with convergence group orders not consistent with the matrix of inferred convergence group orders.

Suppose an arbitrary convergence group is $C_i=\left\{c_{1,i},c_{2,i}\right\}$. On the $N$-taxon CDM, if $\left|c_{1,i}\right|>1$ and/or $\left|c_{2,i}\right|>1$ or $C_i$ is in an epoch before another convergence group, then $C_i$ cannot be in the tip epoch. For other convergence groups, whether they are in the tip epoch or not must be inferred.

For each $4$-taxon CDM with a fixed leaf labeling with a possible convergence group in the tip epoch, we determine which CDM is best supported among the two CDMs, for example, CDM $2$ versus CDM $3$ or CDM $4$ versus CDM $5$. For $C_i$, we tally the $4$-taxon CDMs displaying the given convergence group with and without the convergence group in the tip epoch.

If $C_i$ corresponds to a greater proportion of $4$-taxon CDMs with the convergence group in the tip epoch than any other convergence group and the proportion is greater than some cutoff, for example, half, then we infer that $C_i$ is in the tip epoch. We retain only one possible CDM for each $4$-taxon set after the convergence group order has been assigned and it has been determined which, if any, convergence group is in the tip epoch.

Note that some convergence group orders may still be undefined. Suppose two convergence groups do not have an order defined by the edge partial order of the principal tree or the orders of other convergence groups. Suppose both convergence groups are only ever present on $4$-taxon CDMs where one convergence group is a sister convergence group. Then there will be no information to resolve the order of these two convergence groups. We leave these convergence group orders unresolved. Thus, we have a \emph{partial} order on the convergence groups. Algorithms~\ref{cdmconorder}~and~\ref{cdmdiv} for inferring convergence group orders and any convergence group in the tip epoch then follow.

\begin{algorithm}
\caption{Convergence group order inference}
\textbf{Input: }$N$-taxon CDM $\widehat{\mathcal{N}}$ comprising $N$-taxon topology of principal tree $\widehat{\mathcal{T}}$ and list of convergence groups $\widehat{\mathcal{G}}$, as well as $\binom{N-1}{3}\times{}27$ matrix of model selection criterion values $\boldsymbol{M}$ and matrix of partial edge orders $\boldsymbol{P}$.
\begin{enumerate}[label*=\arabic*.]
\item Initialize empty list of inferred $4$-taxon CDMs $L_Q$. Initialize $k\times{}k$ matrix of observed convergence group orders $\boldsymbol{O}$ as zero matrix, where $k$ is length of list $\widehat{\mathcal{G}}$. Initialize $k\times{}k$ matrix $\boldsymbol{E}$ of expected convergence group orders as convergence group orders defined by $\boldsymbol{P}$, with $\left[\boldsymbol{E}\right]_{ij}=1$ if convergence group $i$ before $j$ and $0$ otherwise.
\item For each $4$-taxon set that includes outgroup $o$, with model selection criterion, select CDM from those displayed on $\widehat{\mathcal{N}}$ and permitted by $\boldsymbol{E}$ and append to $L_Q$.
\item For all $i,j$, compute $\left[\boldsymbol{O}\right]_{ij}$ as proportion of inferred $4$-taxon CDMs displaying convergence groups $i$ and $j$, where $i$ is before $j$.
\item Compute initial sum of squared differences between elements of $\boldsymbol{O}$ and $\boldsymbol{E}$, $s=\sum_{i=1}^{k}\sum_{j=1}^{k}\left(\left[\boldsymbol{O}\right]_{ij}-\left[\boldsymbol{E}\right]_{ij}\right)^2$.
\item Assign new order between two convergence groups that minimizes $s$.
\item Update $\boldsymbol{E}$ and $s$ to reflect newly inferred convergence group order. Suppose new order is convergence group $x$ before $y$. Then all convergence groups above $x$ are also above $y$ and all convergence groups below $y$ are also below $x$. If no pairs of convergence groups left to assign orders to, terminate algorithm. 
\item Return to Step~$5$.
\end{enumerate}
\textbf{Output: }$N$-taxon CDM $\widehat{\mathcal{N}}$ comprising $N$-taxon topology principal tree $\widehat{\mathcal{T}}$ and list of convergence groups $\widehat{\mathcal{G}}$, as well as $\binom{N-1}{3}\times{}27$ matrix of model selection criterion values $\boldsymbol{M}$, matrix of partial edge orders $\boldsymbol{P}$ and matrix of expected convergence group orders $\boldsymbol{E}$.
\label{cdmconorder}
\end{algorithm}

\begin{algorithm}
\caption{Inference of convergence groups in tip epochs}
\textbf{Input: }$N$-taxon CDM $\widehat{\mathcal{N}}$ comprising $N$-taxon topology principal tree $\widehat{\mathcal{T}}$ and list of convergence groups $\widehat{\mathcal{G}}$, as well as $\binom{N-1}{3}\times{}27$ matrix of model selection criterion values $\boldsymbol{M}$, matrix of partial edge orders $\boldsymbol{P}$, matrix of expected convergence group orders $\boldsymbol{E}$ and tolerance $\tau\in\left[0,1\right]$.
\begin{enumerate}[label*=\arabic*.]
\item Initialize empty list of inferred $4$-taxon CDMs $L_Q$. Initialize vector $\boldsymbol{D}$ of length $k$ of convergence groups in tip epoch as zero vector, where $k$ is length of list $\widehat{\mathcal{G}}$.
\item For each $4$-taxon set that includes outgroup $o$, select CDM from those displayed on $\widehat{\mathcal{N}}$ and permitted by $\boldsymbol{E}$ with model selection criterion and append to $L_Q$.
\item For all $i$, if convergence group $C_i=\left\{c_{1,i},c_{2,i}\right\}$ satisfies $\left|c_{1,i}\right|=\left|c_{2,i}\right|=1$ and is not before any other convergence group of $\widehat{\mathcal{N}}$, compute $\left[\boldsymbol{D}\right]_i$ as proportion of inferred $4$-taxon CDMs with $C_i$ in tip epoch.
\item If $\max_{i\in\left\{1,2,\ldots{},k\right\}}\left[\boldsymbol{D}\right]_i=\left[\boldsymbol{D}\right]_j$ and $D_j>\tau$, set $\left[\boldsymbol{D}\right]_j=1$.
\end{enumerate}
\textbf{Output: }$N$-taxon CDM $\widehat{\mathcal{N}}$ comprising $N$-taxon topology principal tree $\widehat{\mathcal{T}}$ and list of convergence groups $\widehat{\mathcal{G}}$, as well as $\binom{N-1}{3}\times{}27$ matrix of model selection criterion values $\boldsymbol{M}$, matrix of partial edge orders $\boldsymbol{P}$, matrix of expected convergence group orders $\boldsymbol{E}$ and vector of convergence groups in tip epoch $\boldsymbol{D}$.
\label{cdmdiv}
\end{algorithm}

We do not prove consistency of inference of the convergence group partial orders from Algorithm~\ref{cdmconorder}. This is because Theorem~\ref{conscongroups} assumes that no leaf taxa belong to more than one convergence group. Thus, all $4$-taxon CDMs displayed on $\mathcal{N}$ have at most one non-sister convergence group and there are no convergence group orders to infer. Furthermore, we do not prove consistency of inference of the convergence groups in the tip epoch.

However, if all inferred $4$-taxon CDMs that include the outgroup are the $4$-taxon CDMs displayed on the generating $N$-taxon CDM after suppressing sister convergence groups, then it is straightforward to prove that Algorithm~\ref{cdmconorder} correctly infers all orders of convergence groups of the generating $N$-taxon CDM that can be determined from the displayed $4$-taxon CDMs. Furthermore, it is also straightforward to prove that Algorithm~\ref{cdmdiv} correctly infers which, if any, convergence group of the generating $N$-taxon CDM is in the tip epoch.

\section{Proof of Proposition~\ref{distident}}
\label{prop14}

\setcounter{prop}{13}

\begin{prop}
All edge lengths of the principal tree of each of CDM $1-5$ are identifiable.
\end{prop}

\begin{proof}
Using the parameterization of Appendix~\ref{CDM5appendix}, for CDM $5$, with principal tree $\left(o,\left(a,\left(b,c\right)\right)\right)$, the sums of edge lengths between leaf taxa are
\begin{align*}
\begin{cases}
d_{o,a}=&l_1+l_3+l_6+l_8+l_9+l_{11}=-\log\left(x_1x_3x_6x_8x_9x_{11}\right)=-\log\left(y_1y_3y_6y_8\right), \\
\phantom{d_{o,a}}\nphantom{$d_{o,b}$}d_{o,b}=&l_1+l_2+l_5+l_6+l_7=-\log\left(x_1x_2x_5x_6x_7\right)=-\log\left(y_1y_2y_5y_6\right), \\
\phantom{d_{o,a}}\nphantom{$d_{o,c}$}d_{o,c}=&l_1+l_2+l_4+l_9+l_{10}=-\log\left(x_1x_2x_4x_9x_{10}\right)=-\log\left(y_1y_2y_4y_8\right), \\
\phantom{d_{o,a}}\nphantom{$d_{a,b}$}d_{a,b}=&l_2+l_3+l_5+2l_6+l_7+l_8+l_9+l_{11}=-\log\left(x_2x_3x_5x_6^2x_7x_8x_9x_{11}\right) \\
\phantom{d_{o,a}}=&-\log\left(y_2y_3y_5y_6^2y_8\right), \\
\phantom{d_{o,a}}\nphantom{$d_{a,c}$}d_{a,c}=&l_2+l_3+l_4+l_6+l_8+2l_9+l_{10}+l_{11}=-\log\left(x_2x_3x_4x_6x_8x_9^2x_{10}x_{11}\right) \\
\phantom{d_{o,a}}=&-\log\left(y_2y_3y_4y_6y_8^2\right), \\
\phantom{d_{o,a}}\nphantom{$d_{b,c}$}d_{b,c}=&l_4+l_5+l_6+l_7+l_9+l_{10}=-\log\left(x_4x_5x_6x_7x_9x_{10}\right)=-\log\left(y_4y_5y_6y_8\right).
\end{cases}
\end{align*}

From Equations~\eqref{ysolns}, the set $\left\{y_1,y_2,y_3,y_4,y_5,y_6,y_7,y_8,y_9\right\}$ is identifiable. Thus, the set $\left\{d_{o,a},d_{o,b},d_{o,c},d_{a,b},d_{a,c},d_{b,c}\right\}$ is also identifiable for CDM $5$. Solving for the lengths of the edges of the principal tree,
\begin{align*}
\begin{cases}
\phantom{l_{bc}}\nphantom{$l_o$}l_o=&\frac{1}{2}\left(d_{o,a}+d_{o,b}-d_{a,b}\right), \\
\phantom{l_{bc}}\nphantom{$l_a$}l_a=&\frac{1}{2}\left(d_{o,a}-d_{o,b}
+d_{a,b}\right), \\
\phantom{l_{bc}}\nphantom{$l_b$}l_b=&\frac{1}{2}\left(d_{a,b}-d_{a,c}+d_{b,c}\right), \\
\phantom{l_{bc}}\nphantom{$l_c$}l_c=&\frac{1}{2}\left(-d_{a,b}+d_{a,c}+d_{b,c}\right), \\
l_{bc}=&\frac{1}{2}\left(-d_{o,a}+d_{o,b}+d_{a,c}-d_{b,c}\right),
\end{cases}
\end{align*}
where $l_o$ is the sum of divergence parameters along the two edges of the principal tree whose parent node is the root, $l_a$, $l_b$ and $l_c$ are the sums of divergence and possibly convergence parameters along the terminal edges whose descendent leaf taxa are $a$, $b$ and $c$ respectively and $l_{bc}$ is the sum of divergence parameters along the edge whose descendent leaf taxa are $b$ and $c$.

It follows that all edge lengths are also identifiable for CDMs $1-4$ since expressions for the sums of edge lengths are the same, except that some $y_i=1$.

\end{proof}

\section{Proof of Proposition~\ref{conident}}
\label{prop15}

\setcounter{prop}{14}

\begin{prop}
All convergence parameters of each of CDM $2-5$ are identifiable.
\end{prop}

\begin{proof}
On CDM $5$, parameters $y_6=x_6$ and $y_8=x_9$ are identifiable. Thus, the convergence parameters $l_6=a_6+b_6=-\log\left(y_6\right)$ and $l_9=a_9+b_9=-\log\left(y_8\right)$ are identifiable. Thus, for all other CDMs with these convergence parameters, they are also identifiable.
\end{proof}

\section{Proof of Proposition~\ref{rootident}}
\label{prop16}

\setcounter{prop}{15}

\begin{prop}
The root parameter $\gamma=\left[\boldsymbol{\Pi}\right]_0-\left[\boldsymbol{\Pi}\right]_1$, where $\left[\boldsymbol{\Pi}\right]_0$ and $\left[\boldsymbol{\Pi}\right]_1$ are the probabilities of states $0$ and $1$ at the root, respectively, is identifiable on each of CDM $1-5$.
\end{prop}

\begin{proof}
From Equation~\eqref{hadphylotensor} for the phylogenetic tensor of CDM $5$, $q_{0001}=q_{0010}=q_{0100}=q_{1000}=\gamma$. Thus, $\gamma$ is identifiable for CDM $5$. Since all other CDMs are nested in CDM $5$ and none correspond to generic values of $\gamma$ --- instead they correspond to some generic values of $x_i$ or $y_i$ --- $\gamma$ is also identifiable for CDMs $1-4$.
\end{proof}

\section{Proof of Theorem~\ref{consparams}}
\label{theorem17}

\setcounter{thm}{16}

\begin{thm}

Suppose CDM $\mathcal{N}$ has topology of principal tree $\mathcal{T}$, convergence groups $\mathcal{G}$, principal tree edge lengths $\boldsymbol{l}$, root parameter $\gamma$ and convergence parameters $\boldsymbol{v}$. Suppose $\mathcal{T}$, $\mathcal{G}$, convergence group partial orders and tip epoch convergence groups of $\mathcal{N}$ are input into Algorithm~\ref{algorithmmetCDM}. Suppose in step~4 of Algorithm~\ref{algorithmmetCDM} only $4$-taxon sets for which $4$-taxon CDMs displayed on $\mathcal{N}$ have no sister convergence are considered. Suppose that for each convergence group of $\mathcal{G}$ --- say $C_a=\left\{c_{1,a},c_{2,a}\right\}$ --- there is at least one $4$-taxon CDM displayed on $\mathcal{N}$ with no sister convergence where $x\in{}c_{1,a}$, $y\in{}c_{2,a}$ are non-sister leaf taxa on the displayed CDM. Suppose further that matrix $\boldsymbol{X}$ in step $6$ of Algorithm~\ref{algorithmmetCDM} has rank $2N-3$. Suppose $\widehat{\boldsymbol{l}}$, $\widehat{\gamma}$ and $\widehat{\boldsymbol{v}}$ are the estimates of $\boldsymbol{l}$, $\gamma$ and $\boldsymbol{v}$, respectively, inferred by Algorithm~\ref{algorithmmetCDM}. Then for any $\epsilon_>0$,
\begin{align*}
\lim_{n\to\infty}\mathbb{P}\left(\left|\widehat{\boldsymbol{l}}-\boldsymbol{l}\right|>\epsilon\right)=0,\quad{}\lim_{n\to\infty}\mathbb{P}\left(\left|\widehat{\gamma}-\gamma\right|>\epsilon\right)=0,\quad{}\lim_{n\to\infty}\mathbb{P}\left(\left|\widehat{\boldsymbol{v}}-\boldsymbol{v}\right|>\epsilon\right)=0,
\end{align*}
where $\left|\widehat{\boldsymbol{l}}-\boldsymbol{l}\right|$ and $\left|\widehat{\boldsymbol{v}}-\boldsymbol{v}\right|$ involve $l^1$ norms.

\end{thm}

\begin{proof}

In step $4$ of Algorithm~\ref{algorithmmetCDM}, only $4$-taxon sets that include the outgroup for which $4$-taxon CDMs displayed on $\mathcal{N}$ have no sister convergence are considered. Thus, all such $4$-taxon CDMs displayed on $\mathcal{N}$ are CDM $1-5$. Since some $4$-taxon sets may not be considered, we cannot yet assume that all parameters are identifiable. However, for a given $4$-taxon set that is considered, from the proof of Proposition~\ref{distident}, all sums of edge lengths between leaf taxa in the $4$-taxon set are identifiable. From Propositions~\ref{conident}~and~\ref{rootident}, all convergence parameters on the $4$-taxon CDM displayed on $\mathcal{N}$ and the root parameter $\gamma$ are also identifiable.

Thus, for the given $4$-taxon set, the estimates of sums of edge lengths between taxa formed from the sums of maximum likelihood estimates of parameters converge in probability to the sums of edge lengths between taxa for $\mathcal{N}$. Likewise, the maximum likelihood estimates of the convergence parameters converge in probability to the convergence parameters on $\mathcal{N}$ and the maximum likelihood estimate of $\gamma$ also converges in probability to $\gamma$. Thus, it follows that when averaging over all $4$-taxon sets that are considered, the estimates of the sums of edge lengths between taxa converge in probability to the values for $\mathcal{N}$.

Now, since the matrix $\boldsymbol{X}$ has rank $2N-3$, $\boldsymbol{X}^T\boldsymbol{X}$ is invertible. It follows that $\boldsymbol{\widehat{l}}$ also converges in probability to $\boldsymbol{l}$ in step $7$ of Algorithm~\ref{algorithmmetCDM}. By assumption, for each convergence group of $\mathcal{G}$ there is at least one $4$-taxon CDM displayed on $\mathcal{N}$ where two converging taxa of the convergence group are non-sister taxa and there is no sister convergence. Thus, each convergence parameter of $\boldsymbol{\widehat{v}}$ is estimated at least once. Thus, $\boldsymbol{\widehat{v}}$ converges in probability to $\boldsymbol{v}$. Finally, since $\gamma$ is fixed across all $4$-taxon CDMs displayed on $\mathcal{N}$, to be consistently estimated it only needs to be estimated for one $4$-taxon CDM displayed on $\mathcal{N}$. Thus, $\widehat{\gamma}$ converges in probability to $\gamma$.

\end{proof}

\putbib
\end{bibunit}

\end{appendices}

\end{document}